\documentclass[dvipsnames,11pt]{article}
\usepackage{amsmath,amssymb,amsthm,color,bm,mathrsfs,extarrows,tikz,graphicx,mathtools,enumitem,fancybox,makecell}
\usepackage[square,sort,comma,numbers]{natbib}
\usepackage{subcaption}
\usepackage[ruled]{algorithm2e}
\usepackage[margin=1in]{geometry}
\usepackage{xcolor,bbm}
\usepackage[pagebackref]{hyperref}
\usepackage{comment}
\usepackage{soul}
\usepackage{ifthen}
\usepackage{mathdots}
\usepackage{enumitem}

\allowdisplaybreaks

\setlist[itemize]{itemsep=0pt}
\setlist[enumerate]{itemsep=0pt}
\hypersetup{colorlinks=true,urlcolor=blue,linkcolor=blue,citecolor=[rgb]{.42,.56,.14},}
\usetikzlibrary{decorations.pathreplacing,decorations.markings,decorations.pathmorphing,decorations.shapes,arrows.meta,positioning,patterns}

\usepackage[capitalise,nameinlink]{cleveref}
\Crefname{lemma}{Lemma}{Lemmas}
\Crefname{fact}{Fact}{Facts}
\Crefname{theorem}{Theorem}{Theorems}
\Crefname{corollary}{Corollary}{Corollaries}
\Crefname{claim}{Claim}{Claims}
\Crefname{example}{Example}{Examples}
\Crefname{problem}{Problem}{Problems}
\Crefname{definition}{Definition}{Definitions}
\Crefname{notation}{Notation}{Notations}
\Crefname{assumption}{Assumption}{Assumptions}
\Crefname{subsection}{Subsection}{Subsections}
\Crefname{section}{Section}{Sections}
\Crefformat{equation}{(#2#1#3)}
\Crefname{figure}{Figure}{Figures}

\newtheorem{theorem}{Theorem}[section]
\newtheorem*{theorem*}{Theorem}

\newtheorem{proposition}[theorem]{Proposition}
\newtheorem*{proposition*}{Proposition}
\newtheorem{property}[theorem]{Property}
\newtheorem*{property*}{Property}
\newtheorem{lemma}[theorem]{Lemma}
\newtheorem*{lemma*}{Lemma}
\newtheorem{corollary}[theorem]{Corollary}
\newtheorem*{corollary*}{Corollary}
\newtheorem*{conjecture*}{Conjecture}
\newtheorem{fact}[theorem]{Fact}
\newtheorem*{fact*}{Fact}

\newtheorem*{exercise*}{Exercise}

\newtheorem*{hypothesis*}{Hypothesis}

\theoremstyle{definition}
\newtheorem{definition}[theorem]{Definition}

\newtheorem{exercise-easy}[theorem]{Exercise}
\newtheorem{exercise-med}[theorem]{Exercise}
\newtheorem{exercise-hard}[theorem]{Exercise$^\star$}
\newtheorem{claim}[theorem]{Claim}
\newtheorem*{claim*}{Claim}

\newtheorem{remark}[theorem]{Remark}
\newtheorem*{remark*}{Remark}

\newtheorem*{observation*}{Observation}

\DeclareSymbolFont{extraup}{U}{zavm}{m}{n}
\DeclareMathSymbol{\varheart}{\mathalpha}{extraup}{86}
\DeclareMathSymbol{\vardiamond}{\mathalpha}{extraup}{87}
\DeclareMathOperator*{\sd}{\mathrm d}
\DeclareMathOperator*{\E}{\mathbb E}

\renewcommand{\Pr}{\operatorname*{\mathbf{Pr}}}

\newcommand{\Mod}[1]{\ (\mathrm{mod}\ #1)}

\newcommand{\eps}{\varepsilon}
\newcommand{\abs}[1]{\left| #1 \right|}
\newcommand{\vabs}[1]{\left\| #1 \right\|}

\newcommand{\pbra}[1]{\left( #1 \right)}
\newcommand{\sbra}[1]{\left[ #1 \right]}
\newcommand{\cbra}[1]{\left\{ #1 \right\}}
\newcommand{\floorbra}[1]{\left\lfloor #1 \right\rfloor}
\newcommand{\ceilbra}[1]{\left\lceil #1 \right\rceil}

\renewcommand{\mid}{\,\middle\vert\,}
\newcommand{\ket}[1]{\left| #1 \right\rangle}

\newcommand{\bin}{\{0,1\}}

\newcommand{\poly}{\mathsf{poly}}

\newcommand{\str}{\mathsf{str}}

\newcommand{\indicator}{\mathbbm{1}}
\newcommand{\err}{\mathsf{err}}
\newcommand{\supp}[1]{\mathsf{supp}\pbra{#1}}
\newcommand{\tow}{\mathrm{tow}}
\newcommand{\ac}{\mathsf{AC^0}}
\newcommand{\nc}{\mathsf{NC^0}}
\newcommand{\qnc}{\mathsf{QNC^0}}

\newcommand{\Nbb}{\mathbb{N}}
\newcommand{\Rbb}{\mathbb{R}}
\newcommand{\Zbb}{\mathbb{Z}}

\newcommand{\Acal}{\mathcal{A}}
\newcommand{\Bcal}{\mathcal{B}}

\newcommand{\Dcal}{\mathcal{D}}
\newcommand{\Ecal}{\mathcal{E}}

\newcommand{\Hcal}{\mathcal{H}}

\newcommand{\Lcal}{\mathcal{L}}
\newcommand{\Mcal}{\mathcal{M}}

\newcommand{\Pcal}{\mathcal{P}}
\newcommand{\Qcal}{\mathcal{Q}}

\newcommand{\Scal}{\mathcal{S}}
\newcommand{\Tcal}{\mathcal{T}}
\newcommand{\Ucal}{\mathcal{U}}

\newcommand{\Wcal}{\mathcal{W}}

\newcommand{\tvdist}[1]{\vabs{#1}_\mathsf{TV}}

\renewcommand{\tilde}{\widetilde}
\renewcommand{\bar}{\overline}

\title{Locality Bounds for Sampling Hamming Slices}
\author{
Daniel M. Kane\thanks{University of California, San Diego. Email: \texttt{dakane@ucsd.edu}. Supported by NSF Medium Award CCF-2107547 and NSF CAREER Award CCF-1553288.}
\and
Anthony Ostuni\thanks{University of California, San Diego. Email: \texttt{aostuni@ucsd.edu}.}
\and
Kewen Wu\thanks{University of California, Berkeley. Email: \texttt{shlw\_kevin@hotmail.com}. Supported by a Sloan Research Fellowship and NSF CAREER Award CCF-2145474.}
}
\date{}

\begin{document}

\maketitle

\begin{abstract}
Spurred by the influential work of Viola (Journal of Computing 2012), the past decade has witnessed an active line of research into the complexity of (approximately) \emph{sampling} distributions, in contrast to the traditional focus on the complexity of \emph{computing} functions. 

We build upon and make explicit earlier implicit results of Viola to provide superconstant lower bounds on the locality of Boolean functions approximately sampling the uniform distribution over binary strings of particular Hamming weights, both exactly and modulo an integer, answering questions of Viola (Journal of Computing 2012) and Filmus, Leigh, Riazanov, and Sokolov (RANDOM 2023). 
Applications to data structure lower bounds and quantum-classical separations are discussed.
\end{abstract}
\thispagestyle{empty}

\newpage
\tableofcontents
\thispagestyle{empty}
\newpage

\section{Introduction}\label{sec:intro}
Historically, complexity theory has been dominated by research determining the complexity of \textit{computing} particular functions. Following the seminal work of Viola \cite{viola2012complexity} (with preliminary ideas appearing in \cite{ambainis2003quantum, goldreich2010implementation}), the past decade has seen a rise in study on the complexity of \textit{sampling} particular distributions \cite{viola2012complexity, lovett2011bounded, beck2012large, de2012extractors, viola2016quadratic, viola2020sampling, goos2020lower, chattopadhyay2022space, viola2023new, filmus2023sampling}. In this setting, the goal is to construct a circuit whose input is an infinite string of unbiased, independent random bits and whose output is a distribution close in total variation distance to a specified distribution.

As a standard motivating example, consider the parity function. H{\aa}stad's flagship result \cite{haastad1986computational} shows that $\ac$ circuits need exponentially many gates to compute parity. However, one can easily sample pairs of the form $(X, \textsc{parity}(X))$, where each output bit depends on just two input bits: output $(x_1 \oplus x_2, x_2 \oplus x_3, \ldots, x_{n-1} \oplus x_n, x_n \oplus x_1)$ on input $x_1, x_2, \ldots$ \cite{babai1987random, boppana1987one}. In fact, $\ac$ circuits can even sample $(X,f(X))$ for more complicated functions, such as inner product \cite{impagliazzo1996efficient} and symmetric functions \cite{viola2012complexity}.

Despite these examples illustrating the difficulty of sampling lower bounds, the challenge has been rewarded with results providing intuition for and applications to succinct data structures \cite{viola2012complexity, lovett2011bounded, beck2012large, viola2020sampling, chattopadhyay2022space, viola2023new}, pseudorandom generators \cite{viola2012bit, lovett2011bounded,beck2012large}, and extractors \cite{viola2012extractors, de2012extractors, viola2014extractors, chattopadhyay2016explicit, cohen2016extractors}.

Let $f\colon\bin^m\to\bin^n$ be a Boolean function of $n$ output bits. Additionally, let $\Ucal^m$ be the uniform distribution over $\bin^m$ and $f(\Ucal^m)$ be the output distribution of $f$ given a uniform input. We say $f$ is $d$-local if every output bit of $f$ depends on at most $d$ input bits.
For constant $d$, this captures $\nc$ circuits, and more generally, $d$-local functions encompass circuits of depth $\log(d)$ and bounded fan-in.
In his influential paper, Viola \cite{viola2012complexity} considered the problem of determining locality lower bounds for Boolean functions $f$ where $f(\Ucal^m)$ approximates one of the following distributions:
\begin{enumerate}
\item $\Dcal_k$ -- the uniform distribution over $x\in\bin^n$ of Hamming weight $k = \Theta(n)$.   
\item $\Mcal_n$ -- the uniform distribution over $x \in \bin^n$ conditioned on $\textsc{modmaj}_p(x)=0$ where
$$
\textsc{modmaj}_p (x) \coloneqq \mathbbm{1}[(x_1 + \cdots + x_n) \bmod p \ge p/2]
\quad\text{for some prime $p=\Theta(\log(n))$.}
$$
\end{enumerate}
In particular, an $\Omega(\log(n))$ locality lower bound was proved for both distributions\footnote{Viola actually considered the uniform distribution over $(X, \textsc{modmaj}_p(X))$ pairs, but these are essentially equivalent (see, e.g., \cite[Lemma 4.3]{viola2012complexity}).}, conditional on $f$'s input domain being sufficiently small ($m=(1+o(1))\cdot n$). Additionally, Viola gave results without this restriction, but with a worse distance bound decaying exponentially in the locality $d$.

Towards a full locality-distance trade-off, Viola asked whether lower bounds could be obtained with strong error bounds and no input length restriction. Later, the similar bound $\tilde\Omega(\log(n/k))$ was discovered for $\Dcal_k$ by Filmus, Leigh, Riazanov, and Sokolov \cite{filmus2023sampling}, answering the question for small $k$. To complement this result, Viola proved an $\Omega(\log(n))$ bound in the case of non-dyadic\footnote{Recall a number is \textit{dyadic} if it can be expressed as a fraction whose denominator is a power of two.} $k/n$ \cite{viola2023new}.\footnote{The bound was originally implicit in \cite{viola2023new} with many of the ideas appearing earlier in \cite{viola2012bit}. Very recently and after the submission of our paper, Viola posted an update making the bound explicit -- see Section 9 of \url{https://eccc.weizmann.ac.il/report/2021/073/}. All our bounds, including the non-dyadic case, were proven independently.\label{fn:non-dyadic-lit}}
However, Viola's question in the general regime remained open.

More recently, building on Viola's results on the distribution $\Mcal_n$, Watts and Parham \cite{watts2023unconditional} proved an input-independent separation between $\qnc$ and $\nc$ circuits of restricted domain size.
Such domain conditions boil down to the exact domain conditions of the locality bounds for $\Mcal_n$ in Viola's analysis, rendering their result only a partial separation.

\subsection{Our Results}\label{sec:our_results}

We resolve Viola's question in the affirmative by providing nontrivial locality lower bounds on $\Dcal_k$ in the $k = \Theta(n)$ regime, as well as for $\Mcal_n$, with no restrictions on the domain size.

Formally, we say a distribution $\Pcal$ is $\delta$-far (resp., $\delta$-close) from a distribution $\Qcal$ if the total variation distance is at least $\delta$ (resp., at most $\delta$).
Our results provide a wide range of trade-offs between locality and distance. For readability, we present them here with some particular parameter choices.

\begin{theorem}[Consequence of \Cref{thm:locality_single_hamming_slice}]\label{thm:informal_single}
Let $f\colon\bin^m\to\bin^n$ be a $d$-local function, and let $1\le k\le n-1$.
If $n$ is sufficiently large and
$$
d\le\frac{\log^*(\log^*(n))}{60}
\quad\text{and}\quad
\frac1{\log^*(\log^*(n))}\le \frac{k}{n} \le1-\frac1{\log^*(\log^*(n))},
$$
then $f(\Ucal^m)$ is $\pbra{1-\frac{\log^*(n)}{\sqrt n}}$-far from $\Dcal_k$, where $\log^*(\cdot)$ is the iterated logarithm with base $2$.
\end{theorem}

By a different application of \Cref{thm:locality_single_hamming_slice}, we obtain the following sharp bound. The tightness of \Cref{thm:informal_single_1} is discussed in \Cref{sec:single_hamming_slice}.

\begin{theorem}[Consequence of \Cref{thm:locality_single_hamming_slice}]\label{thm:informal_single_1}
Let $f\colon\bin^m\to\bin^n$ be a $d$-local function and $1\le k\le n-1$. If both $d$ and $k/n$ are constant, then $f(\Ucal^m)$ is $\pbra{1-O\pbra{1/\sqrt n}}$-far from $\Dcal_k$.
\end{theorem}

We remark that the above theorems hold in a stronger setting where $f$ is fed with some arbitrary binary\footnote{In fact, it works even in the case where the product distribution is not binary, as long as the alphabet is a constant (or slightly superconstant). This will be clear in the proof, and we will discuss it in \Cref{pg:general_input_dist}.} product distribution as an input.
Taking advantage of the fact that the input is actually unbiased coins, we prove the following bound.

\begin{theorem}[Consequence of \Cref{thm:locality_single_non-dyadic}]\label{thm:informal_single_non-dyadic}
Let $f\colon\bin^m\to\bin^n$ be a $d$-local function, where $n$ is a multiple of $3$.
Then $f(\Ucal^m)$ is $\pbra{1-\exp\cbra{-n\cdot2^{-O(d^2)}}}$-far from $\Dcal_{n/3}$.
\end{theorem}

A result (with an adaptive version) giving the bound $1-\exp\cbra{-n^{1-\Omega(1)} \cdot 2^{-d}}$ appears in \cite{viola2023new} (see \Cref{fn:non-dyadic-lit}), which for sufficiently large $d$ eclipses \Cref{thm:informal_single_non-dyadic} (as well as \Cref{thm:informal_single} and \Cref{thm:informal_single_1} when $k/n$ is non-dyadic). Given these similar bounds, we view our primary contribution here to be the generality of our statements, as they have no dyadic restrictions.

Towards lower bounds for the distribution $\Mcal_n$, we introduce the following notation.
Let $q$ be an integer and let $\Lambda\subseteq\Zbb/q\Zbb$ be a non-empty set, where $\Zbb/q\Zbb=\cbra{0,1,\ldots,q-1}$.
We define the distribution $\Dcal_{q,\Lambda}$ to be the uniform distribution over $x\in\bin^n$ conditioned on $(x_1+\cdots+x_n)\bmod q\in\Lambda$.

\begin{theorem}[Consequence of \Cref{thm:mod_slice}]\label{thm:informal_mod}
Let $f\colon\bin^m\to\bin^n$ be a $d$-local function.
Let $3\le q\le\sqrt{n/\log^*(n)}$ be an integer, and let $\Lambda\subseteq\Zbb/q\Zbb$ be a non-empty set.
If $n$ is sufficiently large and $d\le\log^*(\log^*(n))/20$, then $f(\Ucal^m)$ and $\Dcal_{q,\Lambda}$ have distance at least
$$
1-\exp\cbra{-\frac n{q^2\cdot\log^*(n)}}-\begin{cases}
|\Lambda|/q & q\text{ is odd,}\\
2\cdot\max\cbra{|\Lambda_\textsf{even}|,|\Lambda_\textsf{odd}|}/q & q\text{ is even,}
\end{cases}
$$
where $\Lambda_\textsf{even}$ is the set of even numbers in $\Lambda$ and  $\Lambda_\textsf{odd}$ is the set of odd numbers in $\Lambda$.
\end{theorem}

Taking $q=p$ and $\Lambda=\cbra{c\in\Zbb/q\Zbb\colon c\ge p/2}$, we recover $\Dcal_{q,\Lambda}=\Mcal_n$.
By setting $\Lambda=\cbra{0}$, \Cref{thm:informal_mod} also answers an open problem in \cite{filmus2023sampling} for locality lower bounds for the uniform distribution over binary strings of Hamming weight $0 \bmod q$.

We highlight that \Cref{thm:informal_mod} is essentially tight for \emph{any} choice of $q$ and $\Lambda$, and direct interested readers to the discussion in \Cref{sec:periodic_hamming_slice}.

\subsubsection{Data Structure Lower Bounds}\label{sec:dslb}

It is an active line of research to determine optimal bounds for \textit{succinct data structures}: 
structures that store their data close to the information theoretic limit, while still including sufficient redundancy to allow for efficient and meaningful queries \cite{gal2007cell, viola2012bit, viola2012complexity, lovett2011bounded, beck2012large, viola2020sampling, persiano2020tight, larsen2023optimal}. We will focus on the following setting of a binary alphabet and bit probes.

\begin{definition}[Dictionary Problem]\label{def:dictionary}
Let $\Hcal\subseteq\bin^n$ and $s,q\in\Nbb$.
The dictionary problem of $\Hcal$ with parameters $s$ and $q$ asks for a pair of algorithms $\Acal$ and $\Bcal$ such that the following holds:
\begin{itemize}
\item Given an arbitrary $a\in\Hcal$, $\Acal$ produces a data structure $\str_a\in\bin^s$.
\item Given access to $\str\in\bin^s$, for every query $i\in[n]$, $\Bcal$ produces an answer $b_i\in\bin$ with $q$ (adaptive / non-adaptive) bit probes (i.e., number of bits read) to $\str$.
\item When $\str=\str_a$, we have $b_i=a_i$ for all $i\in[n]$.
\end{itemize}
\end{definition}

We remark that this setting is static, in contrast to the dynamic setting where the data structure needs to support updates to the underlying input $a$.
As a weaker model, proving static lower bounds has traditionally been much more difficult than dynamic lower bounds.
The locality-distance trade-off provides a useful tool in establishing trade-offs between parameters in the static dictionary problem.

\begin{claim}[{\cite[Claim 1.8]{viola2012complexity}}]\label{clm:locality_to_ds}
Suppose we can solve the dictionary problem of $\Hcal\subseteq\bin^n$ with parameters $s,q$ and non-adaptive (resp., adaptive) queries.
Then there exists a $q$-local (resp., $(2^q-1)$-local) function $f\colon\bin^s\to\bin^n$ such that $f(\Ucal^s)$ is $(1-|\Hcal|/2^s)$-close to the uniform distribution over $\Hcal$.
\end{claim}

Combining \Cref{clm:locality_to_ds} with our results, we obtain a number of lower bounds. Here we highlight the most interesting ones, and interested readers are encouraged to instantiate more on their own.\footnote{It seems plausible that one could obtain similar results to \Cref{cor:dslb_mod=0}, \Cref{cor:dslb_mod=01}, and \Cref{cor:dslb_n/2} using the ideas in \cite{viola2012complexity}; however, issues arise when trying to modify their proofs, and we do not pursue this direction further. For example, their proofs rely on $\Dcal_{n/2}$ and $\Mcal_n$ failing a particular test $T_0$, which is no longer true when one of their assumptions ($\delta < 1$) is lifted, as is required to improve the storage dependence.}

\begin{corollary}[Via \Cref{thm:informal_mod}]\label{cor:dslb_mod=0}
Let $\Hcal=\cbra{a\in\bin^n\colon(a_1+\cdots+a_n)\bmod r=0}$ where $r\ge3$ is an odd constant.
The dictionary problem of $\Hcal$ needs either $s=n$ bits of storage or $q=\omega(1)$ bit probes per query.
\end{corollary}

Note that the information theoretic limit is $\ceilbra{\log(|\Hcal|)}=n-\floorbra{\log(r)}$. On the other hand, the trivial data structure that simply stores $a$ in $n$ bits can support every membership query by a single bit probe.
Hence \Cref{cor:dslb_mod=0} shows that the \emph{only} efficient data structure using constant probes is the trivial one. We can obtain a similar sharp trade-off for even $r$ by adding a modulus to reflect \Cref{thm:informal_mod}'s different quantitative behavior for odd and even moduli. 

\begin{corollary}[Via \Cref{thm:informal_mod}]\label{cor:dslb_mod=01}
Let $\Hcal=\cbra{a\in\bin^n\colon(a_1+\cdots+a_n)\bmod r\in\cbra{0,1}}$ where $r\ge3$ is an even constant.
The dictionary problem of $\Hcal$ needs either $s=n$ bits of storage or $q=\omega(1)$ bit probes per query.
\end{corollary}

In the case of $\Hcal=\cbra{a\in\bin^n\colon a_1+\cdots+a_n=n/k}$ where $n/k$ is non-dyadic, the results of \cite{viola2012bit} are superior to those we can obtain via our techniques. Specifically, they show that the dictionary problem of $\Hcal$ with $q$ (adaptive) queries requires at least $\log(|\Hcal|) + n2^{-O(q)} - \log(n)$ bits of storage. However, our generality allows us to prove trade-offs in the dyadic setting.

\begin{corollary}[Via \Cref{thm:informal_single_1}]\label{cor:dslb_n/2}
Let $\Hcal=\cbra{a\in\bin^n\colon a_1+\cdots+a_n=n/2}$ where $n$ is a multiple of two.
The dictionary problem of $\Hcal$ needs either $s=n-O(1)$ bits of storage or $q=\omega(1)$ bit probes per query.
\end{corollary}

The previous best result in the setting of \Cref{cor:dslb_n/2} is \cite{viola2012complexity}, which gives an $s=n-0.01\log(n)$ versus $q=\Omega(\log(n))$ trade-off.
Our \Cref{cor:dslb_n/2} improves the storage bound to optimal (ignoring the hidden constant in $O(1)$) at the cost of a worse bit probe bound.
It remains an interesting question whether one can get the best of the two results that further improves our bit probe bound to $\Omega(\log(n))$ without weakening the $n-O(1)$ storage bound.

Meanwhile we note that it is impossible to get an $n$-vs-$\omega(1)$ trade-off as in \Cref{cor:dslb_mod=0} and \Cref{cor:dslb_mod=01}.
Here is a simple data structure of $s=(n-1)$-bit of storage and using $q=2$ bit probes per query for $\Hcal$ in \Cref{cor:dslb_n/2}: for each $i\in[n-1]$, store the prefix sum $\str_a[i]=a_1\oplus\cdots\oplus a_i$. For $i=n$, the prefix sum is precisely $n/2\bmod2$ that we do not need to store. Then every query $i$ can be answered by the parity of the prefix sum up to $i$ and $i-1$. It would be interesting to determine if a similar structure exists with fewer than $(n-1)$-bits of storage while maintaining $O(1)$ bit probes per query.

The results compared here are by no means a complete list of data structure lower bounds for the dictionary problem. 
In particular, there are many results on cell probe lower bounds (e.g., \cite{persiano2020tight}), the dynamic dictionary setting (e.g., \cite{li2023tight}), and other natural choices of $\Hcal$ (e.g., \cite{viola2012bit}).
We refer interested readers to \cite{viola2012complexity} for a detailed discussion.

\subsubsection{Input-Independent Quantum-Classical Separation}\label{subsec:qc_sep}

A driving research direction in quantum computing is exhibiting separations between quantum and classical complexity. In the theme of our paper, we consider the problem of devising distributions that quantum circuits can efficiently sample, whereas classical circuits cannot.
Note that such a separation does not rely on a particular input. Instead, the quantum circuit is fed with a fixed initial state (ideally $\ket{0}^n$), and each qubit is measured in the computational basis at the end to produce the desired distribution over $\bin^n$.
Meanwhile, a classical circuit, which has $n$ output bits, has access to unbiased coins and aims to reproduce the distribution.

The problem of establishing such an \emph{input-independent} separation between circuit classes $\qnc$ and $\nc$ was first proposed by Bravyi, Gosset, and K{\"o}nig \cite{bravyi2018quantum}, and was later found to be connected to the complexity of quantum states \cite{watts2023unconditional} as well.
Using ideas from \cite{viola2012complexity}, Watts and Parham \cite{watts2023unconditional} gave a family of distributions over $\bin^n$ that constant-depth quantum circuits can produce within distance $1/6+o(1)$, but any $\nc$ circuit's output is at least $(1/2-o(1))$-far from, assuming the total number of random bits the $\nc$ circuit could use is $(1+o(1))\cdot n$.
The exact distributions are variants of the $\Mcal_n$ distribution above. 

However, an ideal separation result should have no restriction on the number of classical random bits, as well as a maximal quantum-classical distance gap of $1-o(1)$ or even $1-e^{-\Omega(n)}$.
Towards this goal, \cite{filmus2023sampling} suggested determining locality lower bounds for $\Mcal_n$ and related distributions.
Without diving into detail, we remark that our \Cref{thm:informal_mod} resolves this open problem, and we can lift the domain size assumption in \cite[Theorem 5]{watts2023unconditional} while still preserving the separation.

Aside from directly improving previous analysis, we note that there is a simpler distribution that produces an optimal separation. Let $\Ucal_{1/3}^n$ be the $1/3$-biased distribution over $n$ bits, where each bit is independently set as $1$ with probability $1/3$.

\begin{theorem}[Consequence of \Cref{thm:locality_biased}]\label{thm:informal_biased}
Let $f\colon\bin^m\to\bin^n$ be a $d$-local function.
Then $f(\Ucal^m)$ is $\pbra{1-\exp\cbra{-n\cdot2^{-O(d^2)}}}$-far from $\Ucal_{1/3}^n$.
\end{theorem}

Observe that, starting with $\ket{0}^n$, a $\qnc$ circuit can perfectly simulate $\Ucal_{1/3}^n$ using just one layer of single-qubit gates each mapping $\ket0$ to $\sqrt{2/3}\ket0+\sqrt{1/3}\ket1$.
Thus we obtain the following ideal input-independent quantum-classical separation.

\begin{corollary}\label{cor:qc_sep}
There exists a distribution that $\qnc$ circuits of depth one can sample without error, but any $\nc$ circuit is $\pbra{1-\exp\cbra{-\Omega(n)}}$-far from.
\end{corollary}

We suspect this corollary may be folklore, especially
after a reviewer pointed out that a bound similar to \Cref{thm:informal_biased} is implicit in \cite{viola2012bit, viola2023new}. (Specifically, one can obtain distance $1-\exp\cbra{-n^{1-\Omega(1)} \cdot 2^{-d}}$.) However, it does not seem to explicitly appear in the literature, so we hope our statement will be beneficial to future researchers.

\begin{remark}\label{rmk:qc_sep}
The quantum-classical separation result obtained in \Cref{cor:qc_sep} seems a bit dishonest, as it takes advantage of precision issues arising from the classical binary representation.
One may desire a separation where the quantum circuit is also restricted to ``binary operations'' to rule out distributions like $1/3$-biased.
One natural candidate is Clifford circuits where non-Clifford gates are not allowed. However, the sampling task there sometimes can be reduced to the search task with a constant depth overhead \cite[Section F]{grier2020interactive}, where the latter is trivial in the input-independent setting.
Hence one must be careful in formulating such a restriction on the quantum circuit.

A different way to compensate for the precision issue is to give $\nc$ circuits access to arbitrary binary product distributions. Then $\nc$ circuits can certainly generate $\Ucal_{1/3}^n$ by simply receiving $1/3$-biased coins. 
We remark that our proof of \Cref{thm:informal_mod} still holds in this setting.\footnote{\Cref{thm:informal_mod} works in the case where the input distribution is a product distribution with constant (or slightly superconstant) alphabet. This will be clear in the proof, and we will discuss it in \Cref{pg:general_input_dist}.}
Thus, even giving $\nc$ this extra power, we have a separation combining \Cref{thm:informal_mod} and \cite[Theorem 5]{watts2023unconditional}. 
One caveat here is that the $\qnc$ circuit needs to start with the $\textrm{GHZ}_n$ state. 
If it is forced with $\ket{0}^n$ as the initial state, one may want to prove locality lower bounds (without the domain assumption) for a more complicate distribution designed in \cite[Theorem 3]{watts2023unconditional}. 
Due to its similarity with the $\Mcal_n$ distribution, we believe our techniques can be used there, and we leave this as a future work.\footnote{In the new version of \cite{watts2023unconditional} and as a concurrent work, they updated an appendix section with an extension of their theorem by allowing the $\nc$ circuits to have access to (a bounded number of) biased i.i.d. Bernoulli random variables.}
\end{remark}

\subsection{Future Directions}\label{sec:future}
Beyond considering specific distributions, Filmus, Leigh, Riazanov, and Sokolov \cite{filmus2023sampling} conjectured a classification of when $\nc$ circuits can approximately sample $\Dcal_{\Lambda}$, where $\Dcal_\Lambda$ is the uniform distribution over binary strings with Hamming weights in $\Lambda$. They hypothesized that if $f$ is $O(1)$-local and $f(\Ucal^m)$ is $\eps$-close to $\Dcal_{\Lambda}$, then $f(\Ucal^m)$ is $O(\eps)$-close to $\Dcal_{\Lambda'}$ for $\Lambda'$ being one of the following:
\[
    \{0\}, \{n\}, \{0,n\}, \{0,2,4, \ldots\}, \{1,3,5,\ldots\}, [n].
\]
Our \Cref{thm:informal_single}, combined with their main results, rules out all the singleton $\Lambda'$ other than $\cbra0$ and $\cbra n$.
In addition, our \Cref{thm:informal_mod} rules out all the $q$-periodic $\Lambda'$ for $3\le q\le n^{1/2-o(1)}$.
With a number of new ideas, in an upcoming paper \cite{kaneLocally} we are able to resolve this conjecture (and a strengthening of it) affirmatively.

One question we are not able to resolve concerns the quantitative bounds derived. While our distance bounds are asymptotically optimal when locality is constant, the locality-distance trade-offs deteriorate quickly when locality becomes superconstant. We believe our trade-offs can be further improved, especially in light of the best known upper bounds (see \Cref{sec:upper_bounds}). However, in \Cref{app:tightness} we give examples to show the tightness of the parts in our analysis that create such inevitable blowup. This suggests that new ideas may be needed to get substantial improvements.

\paragraph*{Paper Organization.}
An overview of our proofs is given in \Cref{sec:overview}.
In \Cref{sec:prelim}, we define necessary notation and list standard inequalities.
In \Cref{sec:useful_lemmas}, we prove additional useful inequalities for total variation distance and bipartite graph structures.
In \Cref{sec:special}, we prove sharp lower bounds for specific distributions including biased distributions, uniform  strings of a fixed Hamming weight, and uniform strings of periodic Hamming weights.
Upper bound constructions are presented in \Cref{sec:upper_bounds}.
Missing proofs can be found in the appendices.
In addition, in \Cref{app:tightness}, we give examples to explain the barriers of improving our analysis to obtain better locality lower bounds.
\section{Proof Overview}\label{sec:overview}

Let $f$ be a $d$-local function with $n$ output bits, and let $\Dcal$ be a distribution over $\bin^n$.
The goal is to prove that $f$, fed with uniform inputs, cannot generate a distribution close to $\Dcal$.
The general recipe of establishing such a bound is as follows:
\begin{enumerate}
\item\label{itm:overview_1} 
First, we consider a simpler setting where not only does every output bit of $f$ depend on few input bits, but every input bit of $f$ influences few output bits as well.

In this case, we can find many output bits that depend on disjoint sets of input bits.
Now if the desired distribution $\Dcal$ has long-range correlation (e.g., the Hamming weight must equal $k$), we would expect a large error, since these output bits are independent and cannot coordinate with each other.
\item\label{itm:overview_2} 
Then, based on the error bound established in the first step, we aim to reduce the general case, where we may have popular input bits that many output bits depend on, to the structured case above.

At this step, we shall prove certain graph elimination results, showing that the desired structure in the first step can be obtained after deleting some input bits.
\end{enumerate}
The above description is an oversimplification of our analysis, and for each of our results in \Cref{sec:our_results} we face different issues, which we elaborate on below.
For convenience and simplicity, we will hide minor factors when stating bounds. 

The framework of viewing $f$ as a convex combination of specific, easier-to-handle restrictions was largely developed in \cite{viola2012complexity, viola2020sampling} and applied in \cite{filmus2023sampling}. Thus, our primary contributions are the specific choices of structure we reduce to and the corresponding technical analysis.

\paragraph*{The $1/3$-Biased Distribution.}
We first consider the toy example $\Dcal=\Ucal_{1/3}^n$, the $1/3$-biased distribution.
The idea here works equally well for any $\gamma$-biased distribution where $\gamma$ is non-dyadic.
Observe that $1/3$ can only be approximated up to error $\approx2^{-d}$ using integer multiples of $2^{-d}$.
Therefore the marginal distribution for every output bit of $f$ is doomed to be $2^{-d}$-far from a $1/3$-biased coin.
Since total variation distance is closed under marginal projections, this already implies a $2^{-d}$ bound.

To further boost it to $1-o(1)$, we first assume that we can find $r$ non-connected output bits, i.e., they do not depend on common input bits, which means they are independent.
Since each one of these output bits incurs $2^{-d}$ error, intuitively their error should accumulate.
We prove (\Cref{lem:tvdist_after_product}) that this is indeed the case. If we have $r$ pairs of distributions $(\Pcal_1,\Qcal_1),\ldots,(\Pcal_r,\Qcal_r)$ where each $\Pcal_i$ is $\eps$-far from $\Qcal_i$, then their products $\Pcal_1\times\cdots\times\Pcal_r$ and $\Qcal_1\times\cdots\times\Qcal_r$ are $(1-\exp\cbra{-\eps^2r})$-far from each other.
We briefly sketch the proof: each weak distance bound implies an event $\Ecal_i$ that happens $\eps$ more often in $\Pcal_i$ than $\Qcal_i$. Then by independence and standard concentration, the number of total events happening in $\Pcal_1\times\cdots\times\Pcal_r$ is typically $r\cdot\eps/2$ larger than the number in $\Qcal_1\times\cdots\times\Qcal_r$, thus establishing the bound.
Applied here, each $\Pcal_i$ corresponds to a selected output bit, and each $\Qcal_i$ is a $1/3$-biased coin.
Hence we get (\Cref{prop:biased_more_local}) a $1-\exp\cbra{-r\cdot2^{-d}}$ bound.

Now back to reality, we may not immediately find non-connected output bits, since the degrees of input bits can be unbounded. 
For example, there could be one input bit that all outputs depend on, and therefore no two output bits are independent. 
However, conditioning on this one bit would decrease the degree to $d-1$ and also fix the problem, at the cost of changing the distribution by a factor of $2$.
Since the distance bound above is sufficiently strong, we can indeed pay some loss to condition on input bits.

In particular, we show (\Cref{lem:tvdist_after_conditioning}) that the convex combination of distributions $\Pcal_1,\ldots,\Pcal_m$ is $(1-m\cdot\eps)$-far from a distribution $\Qcal$, provided that each $\Pcal_i$ is $(1-\eps)$-far from $\Qcal$.
This is proved as follows: each distance bound implies an event $\Ecal_i$ happening with probability at least $1-\eps$ in $\Pcal_i$ but at most $\eps$ in $\Qcal$. Then their disjunction will inherent the $1-\eps$ probability in any convex combination of $\Pcal_i$'s, but still happens with at most $m\cdot\eps$ probability in $\Qcal$ by the union bound, which gives the desired statement.
Applied here, each $\Pcal_i$ corresponds to the distribution of output bits conditioned on a specific assignment of $c = \log(m)$ input bits. This will add a $2^c$ overhead on top of the distance bound for the distribution after each conditioning.
Given this observation and the $1-\exp\cbra{-r\cdot2^{-d}}$ bound above, we can afford to delete roughly $r\cdot2^{-d}$ input bits as long as we can find $r$ non-connected output bits after this.

At this point, the problem is graph theoretic: given a bipartite graph $G$ where each left vertex (representing an output bit) has degree bounded by $d$, we are allowed to delete a few right vertices (representing input bits) to get many non-connected left vertices, where we say two left vertices are non-connected if they are not both adjacent to the same right vertex.
More precisely, we are allowed to delete at most $r\cdot2^{-d}$ right vertices to get at least $r$ non-connected left vertices. In addition, we would like to maximize $r$, since the final bound will be roughly $1-\exp\cbra{-r\cdot2^{-d}}$.
It turns out that this can be achieved (\Cref{prop:biased_structure}) with $r=n/2^{d^2}$, which explains our bounds in \Cref{thm:informal_biased}.
The starting point of the proof is the following naive attempt: if we remove all right vertices of degree at least $\ell$, then we obtain a bipartite graph with left degree $d$ and right degree $\ell$, which readily gives $n/(d\cdot\ell)$ non-connected left vertices. 
Hence if the desired bound does not hold, the number of right vertices of degree at least $\ell$ is larger than $r\cdot2^{-d}\ge n/(d2^d\cdot\ell)$.
Then summing over all $\ell$ up to roughly $2^{2^d}$, we will find the right-hand side of above will be a sum of harmonic series and larger than $d\cdot n$, whereas the left-hand side of above is still upper bounded by the number of total edges, which is at most $d\cdot n$. This forms a contradiction.
By analyzing more carefully, we can improve (\Cref{cor:non-adj_vtx_new}) the $2^{2^d}$ to $2^{d^2}$.
This turns out to be sharp (\Cref{app:tightness_vtx}).

\paragraph*{The Hamming Slice of Weight $n/3$.}
Now we move on to the single Hamming slice case $\Dcal=\Dcal_k$, the uniform distribution over $n$-bit binary strings of Hamming weight $k$.
A simpler case here is still when $k/n$ has precision issues -- think of $k=n/3$ for now.
Then every bit in $\Dcal_{n/3}$ is supposed to be $1/3$-biased, whereas every output bit in the produced distribution is still $2^{-d}$-far from it.

While we largely follow the analysis above, the caveat here is that being far from $\Ucal_{1/3}^n$ does not imply being far from $\Dcal_{n/3}$, as the distance between $\Ucal_{1/3}^n$ and $\Dcal_{n/3}$ is itself $1-1/\sqrt n$.
More precisely, this issue arises when we try to aggregate the errors from $m$ independent output bits.
In the previous argument, we compared $\Pcal_i$ (representing the true output bits) and $\Qcal_i=\Ucal_{1/3}$ (representing the desired marginal distributions), then showed that the weak individual error can be boosted to $1-o(1)$ between their product distributions, where the issue kicks in as the product of $\Qcal_i$'s is $\Ucal_{1/3}^n$ instead of (and actually far from) $\Dcal_{n/3}$.

To get around this, we strengthen (\Cref{prop:single_non-dyadic_more_local}) the above argument to use $\Qcal_i$'s as a proxy between the actual distribution and the desired distribution.
Notice that, despite being far in the total variation distance metric, $\Ucal_{1/3}^n$ is close to $\Dcal_{n/3}$ in the pointwise multiplicative error sense.
More formally, every string in the support of $\Dcal_{n/3}$ has density $\binom n{n/3}^{-1}$, and has density $(1/3)^{n/3}(2/3)^{2n/3}$ under $\Ucal_{1/3}^n$.
These two quantities are only off by a $\sqrt n$ multiplicative factor, which means every event of probability at most $\eps$ under $\Ucal_{1/3}^n$ will have probability at most $\eps\sqrt n$ under $\Dcal_{n/3}$.
Therefore we can modify (\Cref{lem:tvdist_after_product}) the previous analysis to show that there is an event of probability at least $1-\exp\cbra{-r\cdot2^{-d}}$ under $\Pcal_1\times\cdots\times\Pcal_r$ but of probability at most $\sqrt n\cdot\exp\cbra{-r\cdot2^{-d}}$ under $\Dcal_{n/3}$, thus establishing a strong distance bound between the actual and desired distributions.
Since later in the graph theoretic task we will set $r\approx n/2^{d^2}$, this $\poly(n)$ loss is affordable when $d$ is not particularly large.

\paragraph*{The Hamming Slice of Weight $n/2$.}
The above analysis works well when individual output bits have inevitable error against the marginal of the desired distribution. As such, we need new ideas if we want to establish lower bounds for the general case.
For simplicity let us focus on the $k=n/2$ case, as the analysis will generalize to any $k$ that is not too close to $0$ or $n$.

If we can find $r$ independent output bits that are $\eps$-far from being unbiased, then we can use the same argument to boost them to a $1-\exp\cbra{-\eps^2r}$ bound.
Otherwise we need to exploit the long-range correlation of $\Dcal_{n/2}$ that the Hamming weight must sum to exactly $n/2$.
One possible exploitation is through anticoncentration inequalities, which have played an important role in the analysis of similar problems \cite{viola2012complexity, chattopadhyay2022space}.
In particular, if there are $r$ independent output bits that are actually unbiased, then by Littlewood-Offord anticoncentration \cite{littlewood1943number,erdos1945lemma}, they cannot sum to any particular value with probability more than $1/\sqrt r$, which seemingly means the distribution is still $(1-1/\sqrt r)$-far from $\Dcal_{n/2}$.
The issue with this argument is that the $r$ independent output bits can correlate with many other output bits, which might be able to force the total Hamming weight to a fixed value.
For example, one can consider the construction $(X_1,1-X_1,X_2,1-X_2,\ldots,X_{n/2},1-X_{n/2})$, where we have $n/2$ independent bits but the total sum is always $n/2$ and every individual bit is unbiased.

To address this problem, we need to take into account the neighborhood of each output bit.
Define the neighborhood $N(i)$ of an output bit $i$ as the set of output bits that depend on some input bit that $i$ also depends on.
We will exploit a key tension between two facts about $N(i)$'s distribution.
Firstly, every small neighborhood should be unbiased, since the marginals of $\Dcal_{n/2}$ restricted to any small number of bits are $1/\poly(n)$-close to the uniform distribution over those bits. Secondly, resampling the input bits on which $i$ depends should not change the Hamming weight of the output (and thus does not change the Hamming weight of $N(i)$). However, since the output of $i$ depends only on these inputs, the second property implies the distribution over Hamming weights of $N(i)$ conditioned on $i=0$ would be the same as the distribution over $i=1$, which contradicts the first property. Note this argument has no issue with the above construction.

Let $\eps$ be a parameter to be optimized later.
We classify each neighborhood as \textsf{Type-1} if it is $\eps$-far from being unbiased, and as \textsf{Type-2} if it is $\eps$-close to unbiased coins.
Mimicking the previous analysis, we say two neighborhoods $N(i),N(j)$ are non-connected if all pairs $(i', j') \in N(i) \times N(j)$ are non-connected.
Thus by the same argument (\Cref{lem:type-1_single}), if we have $r$ non-connected neighborhoods of \textsf{Type-1}, then our distribution is at distance $1-\sqrt n\cdot\exp\cbra{-\eps^2r}$ from $\Dcal_{n/2}$.

Now suppose we have $r$ non-connected neighborhoods of \textsf{Type-2}, each of size at most $t$. We would like to use anticoncentration inequalities to argue that with high probability the Hamming weight does not sum up to $n/2$.
Assume the neighborhoods are $N(1),\ldots,N(r)$ and $I(i)$ is the set of input bits the $i$-th output bit depends on.
We fix all the input bits outside $I(1)\cup\cdots\cup I(r)$ as $\rho$. Then all the output bits outside $N(1)\cup\cdots\cup N(r)$ are fixed, and moreover, the neighborhoods $N(1),\ldots,N(r)$ are independent to each other.
At this point, if the Hamming sum of each $N(i)$ is still not fixed, we can apply anticoncentration (\Cref{fct:ushakov}) to obtain the desired bound.

To this end, we use the property that $N(i)$ is \textsf{Type-2}, i.e., it is roughly unbiased under random $\rho$.
Say $N(i)$ has size $t$. 
Then under a uniform random input, the Hamming sum of $N(i)$ is distributed like a binomial distribution of $t$ coins.
If we resample the input bits in $I(i)$, with half probability the $i$-th output bit is flipped, whereas the Hamming sum of $N(i)\setminus i$ is a binomial distribution of $t-1$ coins.
This implies that such an experiment has $1/\sqrt t-\eps$ probability of changing the Hamming sum of $N(i)$, where $1/\sqrt t$ comes from the total variation distance between a binomial distribution of $t$ coins and its shift, and $\eps$ comes from the error between the actual distribution of $N(i)$ and $\Ucal^t$.
Meanwhile, since $\rho$ does not touch $I(i)$, we cannot change the Hamming sum by simply resampling $I(i)$ if the Hamming sum is already fixed by $\rho$.
Hence as long as $\eps\le1/\sqrt t$, we show (\Cref{clm:lem:type-2_single_1}) that the Hamming sum of $N(i)$ is not fixed under random (and thus a typical) $\rho$.
Since these neighborhoods are independent, by standard concentration many neighborhoods will enjoy this property simultaneously for a typical $\rho$.
Then we can apply anticoncentration and obtain a bound of roughly $1-1/\sqrt r$ (\Cref{lem:type-2_single}).

Set $\eps=1/(2\sqrt t)$.
To summarize (\Cref{prop:single_slice_more_local}), if we have $r$ non-connected neighborhoods of size at most $t$, then 
\begin{itemize}
\item either $r/2$ of them are \textsf{Type-1}, which implies a distance bound of $1-\sqrt n\cdot\exp\cbra{-r/t}$; 
\item or $r/2$ of them are \textsf{Type-2}, which implies a distance bound of $1-1/\sqrt r$.
\end{itemize}
Following the previous argument, this means that we can afford conditioning on $\min\cbra{r/t,\log r}$ input bits to get the above structure. This seems too stringent and impossible, even without considering the undesirable loss in the final bound.
Instead, we observe that the distance bound from the second case actually tells more; it is proved in the stronger sense that our output distribution hits any point in the support of $\Dcal_{n/2}$ with probability at most $1/\sqrt r$.
Hence we can refine (\Cref{lem:tvdist_after_conditioning}) the previous analysis as follows:
any convex combination of $\Pcal_1,\ldots,\Pcal_m$ is $(1-m\cdot\eps_1-\eps_2)$-far from $\Qcal$, provided that each $\Pcal_i$ is either $(1-\eps_1)$-far from $\Qcal$, or hits the support of $\Qcal$ with probability at most $\eps_2$.
The proof is not much different, and we simply merge the event ``not hitting the support'' into the previous union bound.
Therefore we can remove $r/t$ input bits now.

Finally we need to handle the graph theoretic task: given a bipartite graph $G$ with left degree at most $d$, show we can obtain $r$ non-connected left neighborhoods (representing the neighborhoods of output bits) of size $t$ by removing $r/t$ right vertices.
The left neighborhood of a left vertex is the set of left vertices reachable from it with two edges.
Two left neighborhoods are non-connected if they do not connect to common right vertices.
In addition, we aim to maximize $r$ and minimize $t$, since the final distance bound will be $1-1/\sqrt r-\sqrt n\cdot\exp\cbra{-r/t}$.
This task is significantly more challenging than the previous one, as now we need to eliminate the dependency of the neighborhoods too. Consequently, we only get a bound with tower-type dependence on $d$. That is, we show (\Cref{prop:biased_structure}) that the problem can be solved with $r=n/\tow_2(d)$ and $t=\tow_2(d)$.\footnote{$\tow_2(d)=2^{2^{2^{\iddots}}}$ is the tower of $2$'s of height $d$.}
Perhaps surprisingly, this tower-type dependency is in fact necessarily (\Cref{app:tightness_neigh}).

Here we briefly sketch the proof. As before, assume towards contradiction it is false.
Then we follow the previous approach and argue that we will have too many right vertices of large degree, which will imply the following structural result (\Cref{lem:graph_elim_non-adj_neigh}): if we have removed $n/\alpha$ right vertices from the graph where $\alpha\ge C$ is sufficiently large, we can additionally remove $n/\log(\alpha)$ right vertices to shave $n$ edges from the graph.
Then we arrive at a contradiction, as the graph has at most $d\cdot n$ edges and thus can support the elimination process up to $d$ times. However repeating $d$ times only removes roughly $n/\log^{(d)}(n)$ many right vertices in total, which means the elimination process should continue if $\log^{(d)}(n)\ge C$.\footnote{$\log^{(d)}(n)=\log(\log(\log(\cdots(n))))$ is the iterated logarithm of order $d$.} 

\paragraph*{Hamming Slices of Weight $0$ Modulo $3$.}
We note that the analysis for $\Dcal_{n/2}$ also works for the union of multiple Hamming slices, since the main place where we use $n/2$ is that it is \emph{one} fixed value and thus has $1/\sqrt n$ bound via anticoncentration. Beyond a single slice, the $1/\sqrt n$ bound simply scales with the number of slices.
Nevertheless, this does not go beyond $\sqrt n$ slices. Here we demonstrate that our framework is robust enough to handle $\Omega(n)$ periodic slices.

For simplicity, we consider the case where $\Dcal$ equals the uniform distribution over $n$-bit binary strings with Hamming weight $0$ modulo $3$.
Note that this distribution consists of roughly $n/3$ Hamming slices and has marginal distribution almost unbiased.
We follow the proof of $\Dcal_{n/2}$. Let $\eps$ be a parameter measuring the distance between the marginal distribution of neighborhoods and the unbiased distribution.
Similarly we classify each neighborhood as \textsf{Type-1} if the distance is at least $\eps$, and as \textsf{Type-2} if otherwise.
Once we have $r$ non-connected neighborhoods of \textsf{Type-1}, we readily get a $1-\exp\cbra{-\eps^2r}$ distance bound following the same argument.

On the other hand, if we have $r$ non-connected \textsf{Type-2} neighborhoods of size at most $t$, then we use anticoncentration inequalities (in fact, a local limit theorem) to show that with certain probability we cannot have Hamming weight equal to $0$ modulo $3$.
Recall that in the single Hamming slice case, we argue that a \textsf{Type-2} neighborhood $N(i)$ is not fixed after a typical restriction $\rho$ which does not touch $I(i)$ (the input bits that the $i$-th output bit depends on).
This is proved via a thought experiment where we resample the input bits in $I(i)$ and compare the binomial distribution of $t$ coins with its shift.
Here we need a similar statement (\Cref{clm:lem:type-2_mod_1}) that a \textsf{Type-2} neighborhood is not fixed \emph{modulo $3$} after a typical restriction $\rho$.
The only difference is that now we need to compare the binomial-modulo-$3$ distribution with its shift.
Since $3$ does not divide $2$, the binomial-modulo-$3$ distribution can never be uniform over $\cbra{0,1,2}$. 
In fact, by granularity, it is $2^{-t}$-far from its shift, which means that the Hamming sum modulo $3$ of $N(i)$ is typically not fixed as long as $\eps\le2^{-t}/2$.
Then using a local limit theorem (which is an almost tight Littlewood-Offord-type anticoncentration) on the additive group modulo $3$ (\Cref{lem:mod_llt}), we obtain that under typical $\rho$, the Hamming sum modulo $3$ is roughly uniform over $\cbra{0,1,2}$, thus it hits any particular value with probability $1/3+o(1)$.

Set $\eps=2^{-t}/2$.
To summarize (\Cref{prop:mod_slice_more_local}), if we have $r$ non-connected neighborhoods of size at most $t$, then 
\begin{itemize}
\item either $r/2$ of them are \textsf{Type-1}, which implies a distance bound of $1-\exp\cbra{-r/2^t}$; 
\item or $r/2$ of them are \textsf{Type-2}, then hitting $0$ modulo $3$ has probability at most $1/3+o(1)$.
\end{itemize}
By the same reasoning, we seek the above structure at the cost of removing at most $r/2^t$ input bits, while simultaneously maximizing $r$ and minimizing $t$.
It turns out that this (\Cref{prop:mod_slice_structure}) is still manageable with a tower-type loss on $d$ via a similar graph elimination argument.

At last, we mention that the local limit theorem used for analyzing \textsf{Type-2} neighborhoods holds generally for all modulus (\Cref{thm:mod_llt}) including $2$.
However, the comparison between the binomial-modulo-$q$ distribution and its shift can only be done for modulus $q\ge3$.
This is because the binomial-modulo-$2$ distribution is indeed uniform over $\cbra{0,1}$.
Thus for even $q$'s (i.e., $q$'s not coprime with $2$), there will be an additional contributing factor (\Cref{lem:mod_llt}), which results in a different bound for even $q$'s in \Cref{thm:informal_mod}.

\paragraph*{More General Input Distributions.}\label{pg:general_input_dist}

Now we briefly discuss how to modify our analysis to prove similar lower bounds when the input distribution changes from unbiased coins to general product distributions.
While this is not true for the $1/3$-biased distribution (or any $\gamma$-biased in general), it works for the Hamming slices setting. 
Since it is standard that a Boolean circuit takes unbiased coins as input, we focus on this case and leave the following more general treatment for interested readers as an exercise.

Recall that our analysis starts with a simpler setting where we can find many small non-connected neighborhoods.
In this case, we prove distance lower bounds by comparing the marginal distributions of these non-connected neighborhoods with the desired marginal distribution (unbiased or $1/3$-biased coins).
Then we classify them into \textsf{Type-1} and \textsf{Type-2} and argue the final distance bound separately.
The analysis in this part has nothing to do with the \emph{input distribution}, since the only property we need is the non-connectivity of output neighborhoods in the input-output relation, which generalizes trivially when we view the input ``bits'' as taking values in a larger alphabet.
Then we reduce the general setting to the above simpler setting by removing a few input bits.
This part is also oblivious to the alphabet of the input as it works in a purely graph theoretic sense where the input-output dependency is defined in an abstract way regardless the alphabet.

The only problematic part is where we put the above two steps together (\Cref{lem:tvdist_after_conditioning}). 
There, we have to pay a union bound of $m$ for all the possible conditioning (or equivalently, the number of different distributions after conditioning), since the true output distribution is a convex combination of them.
If the alphabet of the input is $\Sigma$ and we need to remove $t$ input bits, we will need to set $m=|\Sigma|^t$.
To compensate this loss, the graph theoretic problem needs to be slightly reformulated, but it will still be manageable if $|\Sigma|$ is a constant or even slightly superconstant.
For example, in the setting of $\Dcal=\Dcal_{n/2}$, previously we needed to obtain $r$ non-connected neighborhoods of size $t$ after removing $r/t$ input bits; now we need to obtain $r$ non-connected neighborhoods of size $t$ after removing $r/(t\cdot\log(|\Sigma|))$ input ``bits''.

Finally we remark that an extremely general result, where the bounds have no restriction on the alphabet, is simply not true. One can use a $1$-local function to sample any distribution if the input alphabet is large enough to include all possible outcomes and is dubbed just with the desired distribution.

\section{Preliminaries}\label{sec:prelim}

For a positive integer $n$, we use $[n]$ to denote the set $\cbra{1,2,\ldots,n}$.
We use $\Rbb$ to denote the set of real numbers, use $\Nbb=\cbra{0,1,2,\ldots}$ to denote the set of natural numbers, and use $\Zbb$ to denote the set of integers.
For a positive integer $q$, we use $\Zbb/q\Zbb=\cbra{0,1,\ldots,q-1}$ to denote the additive group modulo $q$. 
For a binary string $x$, we use $|x|$ to denote its Hamming weight.

We use $\log(x)$ and $\ln(x)$ to denote the logarithm with base $2$ and $e\approx2.71828\ldots$ respectively.
We use $\log^*(x)$ to denote the iterated logarithm with base $e$:
$$
\log^*(x)=\begin{cases}
0 & 0\le x\le1,\\
1+\log^*(\log(x)) & x>1.
\end{cases}
$$
For $a>0$ and $b\in\Nbb$, we use $\tow_a(b)$ to denote the power tower of base $a$ and order $b$, where 
$$
\tow_a(b)=\begin{cases}
1 & b=0,\\
a^{\tow_a(b-1)} & b\ge1.
\end{cases}
$$
Note that $\log^*(\tow_2(b))=b$ and $x\le\tow_2(\log^*(x))\le2^x$.

\paragraph*{Asymptotics.}
We use the standard $O(\cdot), \Omega(\cdot), \Theta(\cdot)$ notation, and emphasize that in this paper they only hide universal positive constants that do not depend on any parameter.

\paragraph*{Probability.}
We reserve $\Ucal$ to denote the uniform distribution over $\bin$, and more generally for $\gamma\in[0,1]$, reserve $\Ucal_\gamma$ to denote the $\gamma$-biased distribution, i.e., $\Ucal_\gamma(1)=\gamma=1-\Ucal_\gamma(0)$.
Note that $\Ucal=\Ucal_{1/2}$.

Let $\Pcal$ be a (discrete) distribution. We use $x\sim\Pcal$ to denote a random sample $x$ drawn from the distribution $\Pcal$.
If $\Pcal$ is a distribution over a product space, then we say $\Pcal$ is a product distribution if its coordinates are independent.
In addition, for any non-empty set $S\subseteq[n]$, we use $\Pcal|_S$ to denote the marginal distribution of $\Pcal$ on coordinates in $S$.
For a deterministic function $f$, we use $f(\Pcal)$ to denote the output distribution of $f(x)$ given a random $x\sim\Pcal$.

For every event $\Ecal$, we define $\Pcal(\Ecal)$ to be the probability that $\Ecal$ happens under distribution $\Pcal$.
In addition, we use $\Pcal(x)$ to denote the probability mass of $x$ under $\Pcal$, and use $\supp{\Pcal}=\cbra{x\colon\Pcal(x)>0}$ to denote the support of $\Pcal$.

Let $\Qcal$ be a distribution. We use $\tvdist{\Pcal-\Qcal}=\frac12\sum_x\abs{\Pcal(x)-\Qcal(x)}$ to denote their total variation distance.\footnote{To evaluate total variation distance, we need two distributions to have the same sample space. This will be clear throughout the paper and thus we omit it for simplicity.}
We say $\Pcal$ is $\eps$-close to $\Qcal$ if $\tvdist{\Pcal(x)-\Qcal(x)}\le\eps$, and $\eps$-far otherwise.

\begin{fact}\label{fct:tvdist}
Total variation distance has the following equivalent characterizations:
$$
\tvdist{\Pcal-\Qcal}=\max_{\text{event }\Ecal}\Pcal(\Ecal)-\Qcal(\Ecal)=\min_{\substack{\text{random variable }(X,Y)\\\text{$X$ has marginal $\Pcal$ and $Y$ has marginal $\Qcal$}}}\Pr\sbra{X\neq Y}.
$$
\end{fact}

Let $\Pcal_1,\ldots,\Pcal_t$ be distributions.
Then $\Pcal_1\times\cdots\times\Pcal_t$ is a distribution denoting the product of $\Pcal_1,\ldots,\Pcal_t$.
We also use $\Pcal^t$ to denote $\Pcal_1\times\cdots\times\Pcal_t$ if each $\Pcal_i$ is the same as $\Pcal$.
For a finite set $S$, we use $\Pcal^S$ to emphasize that coordinates of $\Pcal^{|S|}$ are indexed by elements in $S$.
We say distribution $\Pcal$ is a convex combination of $\Pcal_1,\ldots,\Pcal_t$ if there exist $\alpha_1,\ldots,\alpha_t\in[0,1]$ such that $\sum_{i\in[t]}\alpha_i=1$ and $\Pcal=\sum_{i\in[t]}\alpha_i\cdot\Pcal_i$.

\paragraph*{Locality.}
Let $f\colon\bin^m\to\bin^n$. For each output bit $i\in[n]$, we use $I_f(i)\subseteq[m]$ to denote the set of input bits that the $i$-th output bit depends on.
We say $f$ is a $d$-local function if $|I_f(i)|\le d$ holds for all $i\in[n]$.
Define $N_f(i)=\cbra{i'\in[n]\colon I_f(i)\cap I_f(i')\neq\emptyset}$ to be the neighborhood of $i$, which contains all the output bits that have potential correlation with the $i$-th output bit.
For each input bit $j\in[m]$, we use $\deg_f(j)=\abs{\cbra{i\in[n]\colon j\in I_f(i)}}$ to denote the number of output bits that it influences.

We say output bit $i_1$ is connected to $i_2$ if $I_f(i_1)\cap I_f(i_2)\neq\emptyset$.
We say neighborhood $N_f(i_1)$ is connected to $N_f(i_2)$ if there exist $i_1'\in N_f(i_1)$ and $i_2'\in N_f(i_2)$ such that $I_f(i_1')\cap I_f(i_2')\neq\emptyset$.
As such, every output bit is independent of any non-connected output bit, and the output of a neighborhood has no correlation with any non-connected neighborhood of it.
When $f$ is clear from the context, we will drop subscripts in $I_f(i),N_f(i),\deg_f(j)$ and simply use $I(i),N(i),\deg(j)$.

\paragraph*{Bipartite Graphs.}
We sometimes take an alternative view, using bipartite graphs to model the dependency relations in $f$.
Let $G=(V_1,V_2,E)$ be an undirected bipartite graph.
For each $i\in V_1$, we use $I_G(i)\subseteq V_2$ to denote the set of adjacent vertices in $V_2$.
We say $G$ is $d$-left-bounded if $|I_G(i)|\le d$ holds for all $i\in V_1$.
Define $N_G(i)=\cbra{i'\in V_1\colon I_G(i)\cap I_G(i')\neq\emptyset}$ to be the left neighborhood of $i$.

We say left vertex $i_1$ is connected to $i_2$ if $I_G(i_1)\cap I_G(i_2)\neq\emptyset$.
We say left neighborhood $N_G(i_1)$ is connected to $N_G(i_2)$ if there exist $i_1'\in N_G(i_1)$ and $i_2'\in N_G(i_2)$ such that $I_G(i_1')\cap I_G(i_2')\neq\emptyset$.
For each $j\in V_2$, we use $\deg_G(j)=\abs{\cbra{i\in V_1\colon j\in I_G(i)}}$ to denote its degree.
When $G$ is clear from the context, we will drop subscripts in $I_G(i),N_G(i),\deg_G(j)$ and simply use $I(i),N(i),\deg(j)$.

It is easy to see that the dependency relation in $f\colon\bin^m\to\bin^n$ can be visualized as a bipartite graph $G=G_f$ where $[n]$ is the left vertices (representing output bits of $f$) and $[m]$ is the right vertices (representing input bits of $f$), and an edge $(i,j)\in[n]\times[m]$ exists if and only if $j\in I_f(i)$.
The notation and definitions of $I_f(i),N_f(i),\deg_f(j)$ are then equivalent to those of $I_G(i),N_G(i),\deg_G(j)$.

\paragraph*{Concentration and Anticoncentration.}
We will use the following standard concentration inequalities.

\begin{fact}[Hoeffding's Inequality]\label{fct:hoeffding}
Assume $X_1,\ldots,X_n$ are independent random variables such that $a\le X_i\le b$ holds for all $i\in[n]$.
Then for all $\delta\ge0$, we have
$$
\max\cbra{\Pr\sbra{\frac1n\sum_{i\in[n]}\pbra{X_i-\E[X_i]}\ge\delta},\Pr\sbra{\frac1n\sum_{i\in[n]}\pbra{X_i-\E[X_i]}\le-\delta}}
\le\exp\cbra{-\frac{2n\delta^2}{(b-a)^2}}.
$$
\end{fact}

\begin{fact}[Chernoff's Inequality]\label{fct:chernoff}
Assume $X_1,\ldots,X_n$ are independent random variables such that $X_i\in[0,1]$ holds for all $i\in[n]$.
Let $\mu=\sum_{i\in[n]}\E[X_i]$.
Then for all $\delta\in[0,1]$, we have
$$
\Pr\sbra{\sum_{i\in[n]}X_i\le(1-\delta)\mu}
\le\exp\cbra{-\frac{\delta^2\mu}2}.
$$
\end{fact}

We also need the following version of the Littlewood-Offord-type anticoncentration inequality, which uniformly bounds the probability density function of the sum of independent random variables.
\begin{fact}[{\cite[Theorem 3]{ushakov1986upper}}]\label{fct:ushakov}
Assume $X_1,\ldots,X_n$ are independent random variables in $\Rbb$.
For each $i\in[n]$, define $p_i=\max_{x\in\Rbb}\Pr\sbra{X_i=x}$.
Then there exists a universal constant $C>0$ such that
$$
\Pr\sbra{\sum_{i\in[n]}X_i=x}\le\frac C{\sqrt{\sum_{i\in[n]}(1-p_i)}}
\quad\text{ holds for any $x\in\Rbb$.}
$$
\end{fact}

\paragraph*{Binomials and Entropy.}
Let $\Hcal(x)=x\cdot\log\pbra{\frac1x}+(1-x)\cdot\log\pbra{\frac1{1-x}}$ be the binary entropy function.
We will frequently use the following estimates regarding binomial coefficients and the entropy function.

\begin{fact}[{See e.g., \cite[Lemma 17.5.1]{cover2006elements}}]\label{fct:individual_binom}
For $1\le k\le n-1$, we have
$$
\frac{2^{n\cdot\Hcal(k/n)}}{\sqrt{8k(1-k/n)}}\le\binom nk\le\frac{2^{n\cdot\Hcal(k/n)}}{\sqrt{\pi k(1-k/n)}}.
$$
\end{fact}

\begin{fact}[{See e.g., \cite{wiki:Binary_entropy_function}}]\label{fct:entropy}
For any $x\in[-1,1]$, we have 
$$1-x^2\le\Hcal\pbra{\frac{1+x}2}\le1-\frac{x^2}{2\ln(2)}.$$ 
\end{fact}

\paragraph*{Local Limit Theorems.}
Local limit theorems provide sharp estimates for the probability density function of the sum of independent random variables, strengthening the usual (anti-)concentration inequalities and central limit theorems. 
We refer interested readers to a recent survey by Szewczak and Weber \cite{szewczak2022classical}.

We will require the following local limit result in the additive group modulo $q$, which is a special case of the more general statement \Cref{thm:mod_llt}.
The proof of \Cref{lem:mod_llt} is deferred to \Cref{app:mod_llt}, where we also discuss its tightness.
\begin{lemma}\label{lem:mod_llt}
Let $q\ge3$ be an integer, and let $X_1,\ldots,X_n$ be independent random variables in $\Zbb$.
For each $i\in[n]$ and $r\ge1$, define $p_{r,i}=\max_{x\in\Zbb}\Pr\sbra{X_i\equiv x\Mod r}$ and assume
$$
\sum_{i\in[n]}(1-p_{r,i})\ge L>0
\quad\text{holds for all $r\ge3$ dividing $q$.}
$$
Then for any $\Lambda\subseteq\Zbb/q\Zbb$, we have
$$
\Pr\sbra{\sum_{i\in[n]}X_i\mod q\in\Lambda}\le q\cdot e^{-2L/q^2}+\begin{cases}
|\Lambda|/q & q\text{ is odd,}\\
2\cdot\max\cbra{|\Lambda_\textsf{even}|,|\Lambda_\textsf{odd}|}/q & q\text{ is even,}
\end{cases}
$$
where $\Lambda_\textsf{even}=\cbra{\text{even numbers in }\Lambda}$ and $\Lambda_\textsf{odd}=\cbra{\text{odd numbers in }\Lambda}$.
\end{lemma}
\section{Useful Lemmas}\label{sec:useful_lemmas}

In this section, we prove additional useful lemmas which will appear multiple times with various parameter choices in later sections.
In such generality, they may be of independent interest elsewhere.

\subsection{Total Variation Bounds}\label{sec:tvd_bounds}

Here we prove various lemmas to control total variation bounds.

The following fact is standard, showing that two distributions close to each other remain close after conditioning. For completeness, we include a proof in \Cref{app:proof_of_fct:mult_apx}.

\begin{fact}\label{fct:mult_apx}
Assume $\Pcal$ is $\eps$-close to $\Qcal$, and let $\Pcal',\Qcal'$ be the distributions of $\Pcal,\Qcal$ conditioned on some event $\Ecal$, respectively. Then for any function $f$,
\[
\tvdist{f(\Pcal')-f(\Qcal')}\le\frac{2\eps}{\Qcal(\Ecal)}.
\]
\end{fact}

Intuitively if the marginals of two product distributions do not match, the two distributions in general should be extremely far apart.
This intuition is generalized and formalized as the following lemma, where we actually prove a strengthening of the above intuition that works even for non-product distributions.

\begin{lemma}\label{lem:tvdist_after_product}
Let $\Pcal$, $\Qcal$, and $\Wcal$ be distributions over an $n$-dimensional product space, and let $S\subseteq[n]$ be a non-empty set of size $s$.
Assume
\begin{itemize}
\item $\Pcal|_S$ and $\Wcal|_S$ are two product distributions,
\item $\tvdist{\Pcal|_{\cbra{i}}-\Wcal|_{\cbra i}}\ge\eps$ holds for all $i\in S$, and
\item $\Wcal(x)\ge\eta\cdot\Qcal(x)$ holds for some $\eta>0$ and all $x$.
\end{itemize}
Then
$$
\tvdist{\Pcal-\Qcal}\ge1-2\cdot e^{-\eps^2s/2}/\eta.
$$
\end{lemma}
\begin{proof}
By \Cref{fct:tvdist}, for each $i\in S$ there exists an event $\Ecal_i$ such that $\Pcal|_{\cbra{i}}(\Ecal_i)-\Wcal|_{\cbra{i}}(\Ecal_i)\ge\eps$.
Let $\indicator_{\Ecal_i}\in\bin$ be the indicator of event $\Ecal_i$.
Now define event $\Ecal$ such that
\begin{center}
$\Ecal$ happens if and only if $\frac1s\sum_{i\in S}\pbra{\indicator_{\Ecal_i}-\Pcal|_{\cbra{i}}(\Ecal_i)}\ge-\eps/2$.
\end{center}
Then
\begin{align*}
\Pcal(\Ecal)
&=\Pcal|_S(\Ecal)
=\Pr_{\Pcal|_S}\sbra{\frac1s\sum_{i\in S}\pbra{\indicator_{\Ecal_i}-\Pcal|_{\cbra{i}}(\Ecal_i)}\ge-\eps/2}\\
&=1-\Pr_{\Pcal|_S}\sbra{\frac1s\sum_{i\in S}\pbra{\indicator_{\Ecal_i}-\Pcal|_{\cbra{i}}(\Ecal_i)}<-\eps/2}\\
&\ge1-e^{-\eps^2s/2}.
\tag{since $\Pcal|_S$ is a product distribution and by \Cref{fct:hoeffding}}
\end{align*}
We also have
\begin{align*}
\Wcal(\Ecal)
&=\Wcal|_S(\Ecal)=\Pr_{\Wcal|_S}\sbra{\frac1s\sum_{i\in S}\pbra{\indicator_{\Ecal_i}-\Pcal|_{\cbra{i}}(\Ecal_i)}\ge-\eps/2}\\
&\le\Pr_{\Wcal|_S}\sbra{\frac1s\sum_{i\in S}\pbra{\indicator_{\Ecal_i}-\Wcal|_{\cbra{i}}(\Ecal_i)}\ge\eps/2}
\tag{since $\Pcal_i(\Ecal_i)\ge\Wcal|_{\cbra{i}}(\Ecal_i)+\eps$}\\
&\le e^{-\eps^2s/2}.
\tag{since $\Wcal|_S$ is a product distribution and by \Cref{fct:hoeffding}}
\end{align*}
Since $\Wcal(x)\ge\eta\cdot\Qcal(x)>0$ for all $x$, we have
$$
\Wcal(\Ecal)=\sum_{x:\Ecal\text{ happens}}\Wcal(x)\ge\sum_{x:\Ecal\text{ happens}}\Qcal(x)\cdot\eta=\eta\cdot\Qcal(\Ecal),
$$
which then implies $\Qcal(\Ecal)\le e^{-\eps^2s/2}/\eta$.
Therefore by \Cref{fct:tvdist}, we obtain the desired bound.
\end{proof}

Suppose we can prove distance bounds from a distribution to a set of distributions. This should establish distance bounds from the former distribution to any distribution inside the convex hull of the latter set of distributions.
This is characterized by \Cref{lem:tvdist_after_conditioning}, a special case of which appears in \cite[Section 4.1]{viola2020sampling}.

\begin{lemma}\label{lem:tvdist_after_conditioning}
Let $\Pcal_1,\ldots,\Pcal_t$ and $\Qcal$ be distributions.
Assume there exists an event $\Ecal$ and values $\eps_1,\eps_2,\eps_3$ such that for each $i\in[t]$,
\begin{itemize}
\item either $\tvdist{\Pcal_i-\Qcal}\ge1-\eps_1$ holds,
\item or $\Pcal_i(\Ecal)\le\eps_2$ and $\Qcal(\Ecal)\ge1-\eps_3$ hold.
\end{itemize}
Then for any distribution $\Pcal$ as a convex combination of $\Pcal_1,\ldots,\Pcal_t$, we have
$$
\tvdist{\Pcal-\Qcal}\ge1-(t+1)\cdot\eps_1-\eps_2-\eps_3.
$$
\end{lemma}
\begin{proof}
Let $T\subseteq[t]$ be the set of distributions such that $\tvdist{\Pcal_i-\Qcal}\ge1-\eps_1$.
By \Cref{fct:tvdist}, for each $i\in T$ there exists an event $\Ecal_i$ such that $\Pcal_i(\Ecal_i)-\Qcal(\Ecal_i)\ge1-\eps_1$.
This means 
$$
\Pcal_i(\Ecal_i)\ge1-\eps_1
\quad\text{and}\quad
\Qcal(\Ecal_i)\le\eps_1
\quad\text{for $i\in T$.}
$$
Now define the event $\Ecal'=(\neg\Ecal)\lor\bigvee_{i\in T}\Ecal_i$.
Assume $\Pcal=\sum_{i\in[t]}\alpha_i\cdot\Pcal_i$ is the convex combination.
Then
$$
\Pcal(\Ecal')\ge\sum_{i\in T}\alpha_i\cdot\Pcal_i(\Ecal_i)+\sum_{i\notin T}\alpha_i\cdot\Pcal_i(\neg\Ecal)\ge(1-\eps_1)\cdot\sum_{i\in T}\alpha_i+(1-\eps_2)\cdot\sum_{i\notin T}\alpha_i\ge1-\eps_1-\eps_2,
$$
since $\sum_{i\in[t]}\alpha_i=1$.
In addition,
$$
\Qcal(\Ecal')\le\Qcal(\neg\Ecal)+\sum_{i\in T}\Qcal(\Ecal_i)\le\eps_3+t\cdot\eps_1.
$$
Then the desired bound follows from \Cref{fct:tvdist}.
\end{proof}

The next lemma shows that if two coupled random vectors are both individually $\gamma$-biased, they will still have Hamming weight mismatch (even modulo an integer) as long as parts of their entries are independent.

\begin{lemma}\label{lem:anticoncentration_after_coupling_all}
Let $(X,Y,Z,W)$ be a random variable where $X,Z\in\bin$ and $Y,W\in\bin^{t-1}$.
Let $q\ge\min\cbra{3,t+1}$ be an integer.\footnote{If $q\ge t+1$, then one may instead apply \Cref{lem:anticoncentration_after_coupling_all} with modulus $t+1$, since $X+|Y|\equiv Z+|W|\Mod q$ is equivalent to $X+|Y|=Z+|W|$ for $q\ge t+1$.}
Assume
\begin{itemize}
\item $X$ is independent from $(Z,W)$ and $Z$ is independent from $(X,Y)$,
\item $(X,Y)$ and $(Z,W)$ have the same marginal distribution and are $\eps$-close to $\Ucal_\gamma^t$ for some $\gamma\in(0,1/2]$\footnote{\Cref{lem:anticoncentration_after_coupling_all} holds for $\gamma\in[1/2,1)$ as well, with $\gamma$ replaced by $1-\gamma$ in the bounds. This can be achieved by simply flipping zeros and ones of $(X,Y,Z,W)$. This trick carries over the $\eps$-closeness to $\Ucal_{1-\gamma}^t$ and preserves the congruence.} and 
$$
\eps\le\frac\gamma{4q}\cdot2^{-50\gamma(t-1)/q^2}.
$$
\end{itemize}
Then we have
$$
\Pr\sbra{X+|Y|\equiv Z+|W|\Mod q}\le1-\frac\gamma{2q}\cdot2^{-50\gamma(t-1)/q^2}.
$$
\end{lemma}
\begin{proof}
If $t=1$ then we observe that $\Pr\sbra{X+|Y|=Z+|W|}=\Pr\sbra{X=Z}$ as $q\ge2$.
Since $X$ and $Z$ are independent and of the same distribution $\eps$-close to $\Ucal_\gamma^1$, we have $\Pr\sbra{X=1}=\Pr\sbra{Z=1}\in[\gamma-\eps,\gamma+\eps]$.
Hence 
\begin{equation}\label{eq:lem:anticoncentration_after_coupling_2}
\Pr\sbra{X=Z}=\Pr\sbra{X=1}^2+(1-\Pr\sbra{X=1})^2\le(\gamma-\eps)^2+(1-\gamma+\eps)^2\le1-\gamma/2,
\end{equation}
where we use the fact that $\gamma\in(0,1/2]$ and $\eps\le\gamma/2$.

Now we assume $t\ge2$ and $q\ge3$.
Expand $\Pr\sbra{X+|Y|\equiv Z+|W|\Mod q}$ as
\begin{align}\label{eq:lem:anticoncentration_after_coupling_1}
\sum_{x,z\in\bin}\Pr\sbra{X=x,Z=z}\Pr\sbra{x+|Y|\equiv z+|W|\Mod q\mid X=x,Z=z}.
\end{align}
For fixed $x$ and $z$, consider the distribution of $x+|Y|\bmod q$ conditioned on $X=x,Z=z$.
Since $Z$ is independent from $(X,Y)$, it is the same as the distribution, denoted by $\Pcal_x$, of $x+|Y|\bmod q$ conditioned on $X=x$.
Similarly define $\Qcal_z$ as the distribution of $z+|W|\bmod q$ conditioned on $Z=z$ (or equivalently, conditioned on $Z=z,X=x$).

Since $(X,Y)$ is $\eps$-close to $\Ucal_\gamma^t$, by \Cref{fct:mult_apx}, $\Pcal_0$ is $\frac{2\eps}{1-\gamma}$-close to $\Dcal_0$, the distribution of $|V|\bmod q$ for $V\sim\Ucal_\gamma^{t-1}$.
Similarly, $\Qcal_1$ is $\frac{2\eps}{\gamma}$-close to $\Dcal_1$, the distribution of $1+|V|\bmod q$ for $V\sim\Ucal_\gamma^{t-1}$.
Hence 
\begin{align*}
\Pr\sbra{|Y|\equiv1+|W|\Mod q\mid X=0,Z=1}
&\le1-\tvdist{\Pcal_0-\Qcal_1}
\tag{by \Cref{fct:tvdist}}\\
&\le1+\frac{2\eps}\gamma+\frac{2\eps}{1-\gamma}-\tvdist{\Dcal_0-\Dcal_1}
\tag{by \Cref{fct:mult_apx}}\\
&\le1+\frac{4\eps}\gamma-\tvdist{\Dcal_0-\Dcal_1}
\tag{since $\gamma\le1/2$}\\
&\le1+\frac{4\eps}\gamma-\frac2q\cdot2^{-50\gamma(t-1)/q^2},
\end{align*}
where we apply the following claim for the last inequality.
\Cref{clm:tvdist_gamma_biased_shift} is proved in \Cref{app:proof_of_clm:tvdist_gamma_biased_shift} by Fourier analysis.

\begin{claim}\label{clm:tvdist_gamma_biased_shift}
$\tvdist{\Dcal_0-\Dcal_1}\ge\frac2q\cdot2^{-50\gamma(t-1)/q^2}$ for any $q\ge3$.
\end{claim}

By our assumption on $\eps$, we now have
\begin{align*}
\Pr\sbra{|Y|\equiv1+|W|\Mod q\mid X=0,Z=1}
\le1-\frac1q\cdot2^{-50\gamma(t-1)/q^2}.
\end{align*}
The same bound holds for $\Pr\sbra{1+|Y|\equiv|W|\Mod q\mid X=1,Z=0}$.
Plugging back into \Cref{eq:lem:anticoncentration_after_coupling_1} and using \Cref{eq:lem:anticoncentration_after_coupling_2}, we can upper bound $\Pr\sbra{X+|Y|\equiv Z+|W|\Mod q}$ by
\begin{align*}
\Pr[X=Z]+\Pr[X\neq Z]\cdot\pbra{1-\frac1q\cdot2^{-50\gamma(t-1)/q^2}}
\le1-\frac\gamma{2q}\cdot2^{-50\gamma(t-1)/q^2}
\end{align*}
as desired.
\end{proof}

We remark that \Cref{lem:anticoncentration_after_coupling_all} does not hold when $q=2$ and $t\ge2$, even if we assume $(X,Y)=(Z,W)=\Ucal^t$ and $t=2$: let $B$ be an unbiased coin independent from $X$ and $Z$. Then define $(X,Y,Z,W)=(X,X\oplus B,Z,Z\oplus B)$, and one can verify that this distribution satisfies all the conditions yet has $X+Y\equiv Z+W \Mod{2}$ always.

\subsection{Graph Elimination: Non-Connected Vertices}\label{sec:graph_elim_vtx}

In this section we prove the graph theoretic results mentioned in \Cref{sec:overview} that aim to reduce a general $d$-local function to a more structured one: a $d$-local function with many non-connected output bits.

Recall the notation and terminology for bipartite graphs from \Cref{sec:prelim}. In particular, recall that $d$-left-bounded means each of the left vertices has degree at most $d$, and two left vertices are non-connected if they are not both adjacent to the same right vertex. We show that a $d$-left-bounded bipartite graph $G=([n],[m],E)$ has many non-connected left vertices after removing few right vertices. 

Let $\beta,\lambda\ge1$ be parameters (not necessarily constant). We formalize the desired property as the following \Cref{as:non-adj_vtx} with parameters $\beta,\lambda$.

\begin{property}\label{as:non-adj_vtx}
There exists $S\subseteq[m]$ such that deleting those right vertices (and their incident edges) produces a bipartite graph with $r$ non-connected left vertices satisfying
$$
|S|\le\frac r\beta
\quad\text{and}\quad
r\ge\frac n\lambda.
$$
\end{property}

Assuming \Cref{as:non-adj_vtx} is false, we prove the following graph elimination result to show that we can remove many edges by deleting few right vertices.
Later we will iteratively apply this with a proper choice of relation between $\beta$ and $\lambda$ to show that actually \Cref{as:non-adj_vtx} always holds.

\begin{lemma}\label{lem:non-adj_vtx_new}
Assume \Cref{as:non-adj_vtx} does not hold for a particular choice of (not necessarily constant) parameters $\beta,\lambda\ge1$ and $d$-left-bounded bipartite graph $G=([n],[m],E)$ with $d \ge 1$.
Let $U\subseteq[m]$ be of size at most $n/\alpha$.
Define
\begin{equation}\label{eq:lem:non-adj_vtx_new_1}
s=\min\cbra{n,\frac\lambda{2d},\frac\alpha{2d\beta}}.
\end{equation}
If $s\ge1$, then there exists $V\subseteq[m]\setminus U$ of size at most $n/s$ and $\sum_{j\in V}\deg_G(j)\ge n/2$.
\end{lemma}

The proof exploits that unless such a $V$ exists, there are many small neighborhoods. 
If so, we can find many non-connected left vertices by a simple greedy argument, contradicting to the assumption that \Cref{as:non-adj_vtx} is false.

\begin{proof}[Proof of \Cref{lem:non-adj_vtx_new}]
We first remove right vertices (and their incident edges) in $U$ to obtain graph $G'$.
Note that $\deg_G(j)=\deg_{G'}(j)$ holds for any $j\in[m]\setminus U$. Hence for simplicity we use $\deg(j)$ to denote both of them.

For each $i\in[n]$, we say $N_{G'}(i)$ is a small left neighborhood if every $j\in I_{G'}(i)$ satisfies $\deg(j)<s$.
Since $G$ is $d$-left-bounded, each small left neighborhood has size less than $d\cdot s$.
Let $A=\cbra{j\in[m]\setminus U\colon\deg(j)\ge s}$. Then each $j\in A$ prevents $\deg(j)$ left neighborhoods from being small, which means that the number of small left neighborhoods is at least $n-\sum_{j\in A}\deg(j)$.
If $\sum_{j\in A}\deg(j)\ge n/2$, then we have the following cases:
\begin{itemize}
\item if $|A|\le n/s$, then \Cref{lem:non-adj_vtx_new} follows by setting $V=A$,
\item otherwise $|A|>n/s$. Then we pick an arbitrary set $V\subset A$ of size $\floorbra{n/s}$. Since $\deg(j)>s$ for all $j\in A\supset V$, we have $\sum_{j\in V}\deg(j)>s\cdot\floorbra{n/s}\ge n/2$ for any $0<s\le n$.
\end{itemize}
Now we assume $\sum_{j\in A}\deg(j)<n/2$, which means that we have at least $n/2$ small left neighborhoods.
We will show that this cannot happen.

Observe that two left vertices in $G'$ are non-connected if and only if one is not in the left neighborhood of the other.
Since $d,s\ge1$, among the left vertices with small left neighborhoods, we can find at least 
$$
r=\underbrace{\frac n2}_{\#i\colon N_{G'}(i)\text{ is small}}\cdot\underbrace{\frac1{d\cdot s}}_{|N_{G'}(i)|<d\cdot s}
$$ 
of them that are not connected to each other.
If \Cref{as:non-adj_vtx} is false, we must have
$$
\frac n\alpha\ge|U|>\frac r\beta=\frac n{2ds\beta}
\quad\text{or}\quad
r=\frac n{2ds}<\frac n\lambda.
$$
However by our choice \Cref{eq:lem:non-adj_vtx_new_1} of $s$, we now have a contradiction.
\end{proof}

We now show that \Cref{as:non-adj_vtx} always holds if $\lambda$ is not too small with respect to $\beta$ and $d$.

\begin{corollary}\label{cor:non-adj_vtx_new}
Let $\beta,\lambda\ge1$ be parameters (not necessarily constant), and let $G=([n],[m],E)$ be a $d$-left-bounded bipartite graph with $d \ge 1$.
If 
$$
\lambda\ge2d\cdot(2d\beta+1)^{2d},
$$
then \Cref{as:non-adj_vtx} holds for $G$.
\end{corollary}
\begin{proof}
If $n\le\lambda$, then we can simply pick an arbitrary left vertex in $G$ and set $S=\emptyset,r=1$ for \Cref{as:non-adj_vtx}.
Now we assume $n\ge\lambda\ge2d\cdot(2d\beta+1)^{2d}$ and \Cref{cor:non-adj_vtx_new} is false.

We will apply \Cref{lem:non-adj_vtx_new} iteratively.
For convenience, we define
$$
\alpha_i=(2d\beta+1)^{2d-i}\cdot2d\beta
\quad\text{and}\quad
s_i=(2d\beta+1)^{2d-i}
\quad\text{for $i=0,1,\ldots,2d$.}
$$
Notice that for each $i=0,1,\ldots,2d$, we have 
\begin{equation}\label{eq:cor:non-adj_vtx_new_1}
s_i\ge1
\quad\text{and}\quad
\min\cbra{n,\frac\lambda{2d},\frac{\alpha_i}{2d\beta}}=\frac{\alpha_i}{2d\beta}=s_i.
\end{equation}

Let $U_0=\emptyset$.
For each $i=0,1,\ldots,2d$, we apply \Cref{lem:non-adj_vtx_new} to $U_i$ with $\alpha_i$ to obtain $V_i$, then set $U_{i+1}=U_i\cup V_i$.
Now we prove by induction that $|U_i|\le n/\alpha_i$ and $|V_i|\le n/s_i$, which establishes the validity of the above process.
The base case $i=0$ is $|U_i|=0\le n/\alpha_i$ and $|V_i|\le n/s_i$ by \Cref{eq:cor:non-adj_vtx_new_1}.
For the inductive case $i\ge1$, we have
$$
|U_i|=|U_{i-1}|+|V_{i-1}|\le\frac n{\alpha_{i-1}}+\frac n{s_{i-1}}=\frac n{\alpha_i}
$$
by the induction hypothesis and our choice of $\alpha_{i-1},s_{i-1},\alpha_i$.
Then the size bound on $|V_i|$ follows again from \Cref{eq:cor:non-adj_vtx_new_1}.
This completes the induction.

Note that by \Cref{lem:non-adj_vtx_new}, we have
$$
\sum_{j\in U_{2d+1}}\deg_G(j)=\sum_{i=0}^{2d}\sum_{j\in V_i}\deg_G(j)\ge(2d+1)\cdot n/2,
$$
contradicting the fact that $G$ has at most $d\cdot n$ edges, as it is $d$-left-bounded.
Hence \Cref{cor:non-adj_vtx_new} must be true.
\end{proof}

We will apply this result in the proof of \Cref{prop:biased_structure} to show that we can find many independent output bits of a $d$-local function by conditioning on only a few input bits.
In addition, in \Cref{app:tightness_vtx}, we will show that the bound in \Cref{cor:non-adj_vtx_new} is essentially tight.

\subsection{Graph Elimination: Non-Connected Neighborhoods}\label{sec:graph_elim_neigh}

Similar to \Cref{sec:graph_elim_vtx}, here we aim to reduce a general $d$-local function to one having many non-connected neighborhoods of small size by deleting a few input bits. However the situation here is much more complicated than the one in \Cref{sec:graph_elim_neigh}, particularly because in later applications, we will impose different constraints between the number of input bits and the number of non-connected neighborhoods.

Let $\lambda,\kappa\ge1$ be parameters (not necessarily constant) and $G=([n],[m],E)$ be a $d$-left-bounded bipartite graph.
Let $F(\cdot)$ be a function to be chosen based on later applications.
We will require that $G$ has many non-connected left neighborhoods after removing few right vertices, formulated as the following property.

\begin{property}\label{as:non-adj_neigh}
There exists $S\subseteq[m]$ such that deleting those right vertices (and their incident edges) produces a bipartite graph with $r$ non-connected left neighborhoods of size at most $t$ satisfying
$$
|S|\le\frac r{F(t)}
\quad\text{and}\quad
r\ge\frac n\lambda
\quad\text{and}\quad
t\le\kappa.
$$
\end{property}

Similar to the previous section, we prove the following graph elimination result, which shows, under the condition that \Cref{as:non-adj_neigh} is false, we can remove many edges by deleting few right vertices.

\begin{lemma}\label{lem:graph_elim_non-adj_neigh}
Assume \Cref{as:non-adj_neigh} does not hold for a particular choice of (not necessarily constant) parameters $\lambda,\kappa\ge1$ and $d$-left-bounded bipartite graph $G = ([n], [m], E)$ with $d\ge 1$.
Let $U\subseteq[m]$ be of size at most $n/\alpha$, and let $s$ be another parameter.
If 
\begin{equation}\label{eq:lem:graph_elim_non-adj_neigh_1}
1\le s\le\min\cbra{n,\frac\kappa d}
\quad\text{and}\quad
1\le\alpha\le2\lambda\cdot F(d\cdot s)
\quad\text{and}\quad
\ln(\alpha\cdot d)\ge8d^4s^2\cdot F(d\cdot s),
\end{equation}
then there exists $V\subseteq[m]\setminus U$ of size at most $n/s$ and $\sum_{j\in V}\deg_G(j)\ge n/2$.
\end{lemma}

The proof is similar to the proof of \Cref{lem:non-adj_vtx_new} and exploits that unless such a $V$ exists, there are many small neighborhoods. If so, consider taking $S = \{v \in [m] : \deg(v) \ge \ell\}$ in \Cref{as:non-adj_neigh}. Since we removed the vertices of high degree, a small neighborhood cannot be connected to too many others. Hence, unless this $S$ satisfies \Cref{as:non-adj_neigh}, it must be that $S$ is large. However, this implies there are many right vertices of large degree, violating our total degree bound.

\begin{proof}[Proof of \Cref{lem:graph_elim_non-adj_neigh}]
We first remove right vertices (and their incident edges) in $U$ to obtain graph $G'$.
Note that $\deg_G(j)=\deg_{G'}(j)$ holds for all $j\in[m]\setminus U$. Hence for simplicity we use $\deg(j)$ to denote both of them.
For each $i\in[n]$, we say $N_{G'}(i)$ is a small left neighborhood if every $j\in I_{G'}(i)$ satisfies $\deg(j)<s$.
Since $G$ is $d$-left-bounded, each small left neighborhood has size less than $d\cdot s$. 
By the same argument in the proof of \Cref{lem:non-adj_vtx_new}, the lemma holds unless there are at least $n/2$ small left neighborhoods, so assume this to be the case. We will show this cannot happen.

Let $K$ be a parameter to determine later.
For $1\le\ell\le K$, let $B_\ell=\cbra{j\in[m]\setminus U\colon\deg(j)\ge\ell}$ be the set of right vertices with degree at least $\ell$.
If we remove $B_\ell$ from $G'$ and obtain $H$, every small left neighborhood $N_{H}(i)$ is connected to 
$$
<
\underbrace{d\cdot s}_{i'\in N_H(i)}
\cdot\underbrace{d}_{j'\in I_H(i')}
\cdot\underbrace{\ell}_{i''\colon I_H(i'')\ni j'}
\cdot\underbrace{d}_{j''\in I_H(i'')}
\cdot\underbrace{s}_{i'''\colon I_H(i''')\ni j''}
=d^3s^2\ell
$$
small left neighborhoods.
Since there are at least $n/2$ small left neighborhoods and $d,s,\ell\ge1$, we can find
$$
r=\frac n{2d^3s^2\ell}
$$
non-connected small left neighborhoods, each of which has size less than $d\cdot s$. Setting 
$$
K=\frac\alpha{4d^3s^2F(d\cdot s)},
$$
we obtain 
\[
r\ge\frac n{2d^3s^2K}=\frac{2n\cdot F(d\cdot s)}\alpha,
\]
so \Cref{eq:lem:graph_elim_non-adj_neigh_1} implies
\begin{equation}\label{eq:lem:graph_elim_non-adj_neigh_2}
t = d\cdot s\le\kappa
\quad\text{and}\quad
r\ge\frac n\lambda.
\end{equation}

If \Cref{as:non-adj_neigh} is false, we must have
$$
|U|+|B_\ell|>\frac r{F(t)}=\frac n{2d^3s^2F(d\cdot s)\cdot\ell},
$$
since the other conditions are satisfied by \Cref{eq:lem:graph_elim_non-adj_neigh_2}.
Therefore, by \Cref{eq:lem:graph_elim_non-adj_neigh_1} and \Cref{eq:lem:graph_elim_non-adj_neigh_2}, we have
\begin{equation*}
|B_\ell|>\frac n{2d^3s^2F(d\cdot s)\cdot\ell}-\frac n\alpha
\ge\frac n{4d^3s^2F(d\cdot s)\cdot\ell}.
\end{equation*}
where the last inequality follows from
$$
4d^3s^2F(d\cdot s)\cdot\ell\le4d^3s^2F(d\cdot s)\cdot K=\alpha.
$$
Now we sum over\footnote{In the following inequality we do not need to assume $K\ge1$, since otherwise $\text{LHS}=0>\text{RHS}$ already holds.} $1\le\ell\le\floorbra{K}$ and obtain
\begin{align*}
\sum_{1\le\ell\le\floorbra{K}}|B_\ell|
>\sum_{1\le\ell\le\floorbra{K}}\frac n{4d^3s^2F(d\cdot s)\cdot\ell}
\ge\frac n{4d^3s^2F(d\cdot s)}\int_1^K\frac1\ell\sd\ell
=\frac{n\cdot\ln(K)}{4d^3s^2F(d\cdot s)}.
\end{align*}
Since $G'$ is $d$-left-bounded, we also have
$$
\sum_{1\le\ell\le\floorbra{K}}|B_\ell|
\le\sum_{\ell\ge1}\abs{\cbra{j\in[m]\setminus U\colon\deg(j)\ge\ell}}
=\text{number of edges in $G'$}\le d\cdot n.
$$

At this point, we obtain the relation
$$
d>\frac{\ln(K)}{4d^3s^2F(d\cdot s)}=\frac1{4d^3s^2F(d\cdot s)}\cdot\ln\pbra{\frac\alpha{4d^3s^2F(d\cdot s)}},
$$
or equivalently, $K \cdot \ln(K) < \alpha d$. Since $\alpha,d\ge1$, this implies $K<\frac{2\alpha\cdot d}{\ln(\alpha\cdot d)}$, i.e.,
$$
\ln(\alpha\cdot d)<8d^4s^2\cdot F(d\cdot s),
$$
which contradicts \Cref{eq:lem:graph_elim_non-adj_neigh_1}.
\end{proof}

Similar to \Cref{cor:non-adj_vtx_new}, we also show that \Cref{as:non-adj_neigh} holds if $\lambda,\kappa$ are not too small with respect to $d$ and the function $F$.

\begin{corollary}\label{cor:graph_elim_non-adj_neigh}
Let $\lambda,\kappa\ge1$ be parameters (not necessarily constant),  $F(\cdot)$ be an increasing function, and $G=([n],[m],E)$ be a $d$-left-bounded bipartite graph with $d \ge 1$.

Define
\begin{equation}\label{eq:cor:graph_elim_non-adj_neigh_1}
\tilde F(x)=\frac1d\cdot\exp\cbra{32d^4x^2\cdot F(2d\cdot x)}.
\end{equation}
Assume $H(\cdot)$ is an increasing function and $H(x)\ge\tilde F(x)$ for all $x\ge L$ where $L\ge1$ is some parameter not necessarily constant.
If $H(x)\ge2x$ for all $x\ge L$ and
\begin{equation}\label{eq:cor:graph_elim_non-adj_neigh_0}
F(x)\ge1
\text{ holds for all $x\ge1$}
\quad\text{and}\quad
\kappa\ge\lambda\ge d\cdot H^{(2d+2)}(L),
\end{equation}
where $H^{(k)}$ is the iterated $H$ of order $k$\footnote{$H^{(1)}(x)=H(x)$ and $H^{(k)}(x)=H(H^{(k-1)}(x))$ for $k\ge2$.}, then \Cref{as:non-adj_neigh} holds for $G$.
\end{corollary}
\begin{proof}
The proof is similar to the one for \Cref{cor:non-adj_vtx_new}.
If $n\le\lambda$, then we can simply pick an arbitrary left vertex in $G$ and set $S=\emptyset,r=1,t=n$ for \Cref{as:non-adj_neigh}.
Now we assume $n\ge\lambda$ and \Cref{cor:graph_elim_non-adj_neigh} is false.

We will apply \Cref{lem:graph_elim_non-adj_neigh} iteratively. For convenience, we define
$$
\alpha_i=H^{(2d+2-i)}(L)
\quad\text{and}\quad
s_i=2\cdot H^{(2d+1-i)}(L)
\quad\text{for $i=0,1,\ldots,2d$.}
$$
Since $H$ is increasing and $H(x)\ge2x$ for $x\ge L$, we have 
\begin{equation}\label{eq:cor:graph_elim_non-adj_neigh_2}
L\le H^{(2)}(L)=\alpha_{2d}\le\alpha_{2d-1}\le\cdots\le\alpha_0=H^{(2d+2)}(L). 
\end{equation}
Similarly
\begin{equation}\label{eq:cor:graph_elim_non-adj_neigh_3}
2\cdot L\le2\cdot H(L)=s_{2d}\le s_{2d-1}\le\cdots\le s_0=2\cdot H^{(2d+1)}(L)\le H^{(2d+2)}(L).
\end{equation}

Let $U_0=\emptyset$.
For each $i=0,1,\ldots,2d$, we apply \Cref{lem:graph_elim_non-adj_neigh} to $U_i$ with $\alpha_i$ to obtain $V_i$ with $s_i$, then set $U_{i+1}=U_i\cup V_i$.
To show the validity of the process, we first verify the following relations:
\begin{equation}\label{eq:cor:graph_elim_non-adj_neigh_4}
1\le s_i\le\min\cbra{n,\frac\kappa d}
\quad\text{and}\quad
1\le\alpha_i\le2\lambda\cdot F(d\cdot s_i)
\quad\text{and}\quad
\ln(\alpha_i\cdot d)\ge8d^4s_i^2\cdot F(d\cdot s_i).
\end{equation}
\begin{itemize}
\item The first one is due to $1\le2\cdot L\le s_{i}\le2\cdot H^{(2d+1)}(L)\le\lambda\le n$ and $\kappa/d\ge H^{(2d+2)}(L)\ge s_i$ by \Cref{eq:cor:graph_elim_non-adj_neigh_3} and \Cref{eq:cor:graph_elim_non-adj_neigh_0}.
\item The second one is due to $1\le L\le\alpha_i\le H^{(2d+2)}(L)\le\lambda$ and $F(d\cdot s_i)\ge1$ as $d\cdot s_i\ge1$ by \Cref{eq:cor:graph_elim_non-adj_neigh_2} and \Cref{eq:cor:graph_elim_non-adj_neigh_0}.
\item The third one is equivalent to verifying 
$$
\alpha_i\ge\frac1d\cdot\exp\cbra{8d^4s_i^2\cdot F(d\cdot s_i)}=\tilde F(s_i/2),
$$
where we recall the definition of $\tilde F$ from \Cref{eq:cor:graph_elim_non-adj_neigh_1}.
Since $H(x)\ge\tilde F(x)$ for $x\ge L$ and $s_i/2\ge L$ from \Cref{eq:cor:graph_elim_non-adj_neigh_3}, we have
$$
\tilde F(s_i/2)\le H(s_i/2)=H^{(2d+2-i)}(L)=\alpha_i
$$
as desired.
\end{itemize}
Given \Cref{eq:cor:graph_elim_non-adj_neigh_4}, we prove by induction that $|U_i|\le n/\alpha_i$ and $|V_i|\le n/s_i$.
The base case $i=0$ is $|U_i|=0\le n/\alpha_i$ and $|V_i|\le n/s_i$ by \Cref{eq:cor:graph_elim_non-adj_neigh_4}.
For the inductive case $i\ge1$, we have
$$
|U_i|=|U_{i-1}|+|V_{i-1}|\le\frac n{\alpha_{i-1}}+\frac n{s_{i-1}}=\frac n{H^{(2d+3-i)}(L)}+\frac n{2\cdot H^{(2d+2-i)}(L)}\le\frac n{H^{(2d+2-i)}(L)}=\frac n{\alpha_i},
$$
where the second inequality used the fact that $H(x)\ge2x$ for $x\ge L$.
Then the size bound on $|V_i|$ follows again from \Cref{eq:cor:graph_elim_non-adj_neigh_4}.
This completes the induction.
Finally we obtain the same contradiction from the total number of edges as in the proof of \Cref{cor:non-adj_vtx_new}. Hence \Cref{cor:graph_elim_non-adj_neigh} must be true.
\end{proof}

Observe that in \Cref{cor:graph_elim_non-adj_neigh}, even if $F$ is a constant function, $\tilde F$ (and hence $H$) will grow faster than an exponential function. 
This implies that the lower bound on $\kappa$ and $\lambda$ will (at least) be a tower-type blowup in $d$.
We emphasize that this is surprisingly inevitable and will be elaborated in \Cref{app:tightness_neigh}.
\section{Lower Bounds}\label{sec:special}

In this section, we prove lower bounds for a variety of distributions related to Hamming slices. \Cref{sec:biased} contains lower bounds for $\gamma$-biased distributions. \Cref{sec:single_non-dyadic} contains lower bounds for single Hamming slices of weight $\gamma n$ when $\gamma$ has binary representation error, and in \Cref{sec:single_hamming_slice}, we extend the analysis to the general case of $\gamma$.
We conclude by proving lower bounds for sampling from periodic Hamming slices in \Cref{sec:periodic_hamming_slice}.

\subsection{Biased Distributions}\label{sec:biased}

Recall that $\Dcal_\gamma^n$ is the $\gamma$-biased distribution on $\bin^n$.
In this section we show that if $\gamma$ is not close to a dyadic number, then local functions cannot produce distributions close to $\Dcal_\gamma^n$.
After proving this result, we learned a similar one is implicit in \cite{viola2012bit, viola2023new}. However, we do not find references explicitly giving such bounds, and the techniques used in proving this result will be generalized to other cases.
Therefore we include a complete proof here.

\begin{definition}[Binary Representation Error]\label{def:abs_binary_rep_err}
For each $t\in\Nbb$, we use $\err(\gamma,t)$ to denote the minimum distance of $\gamma$ to an integer multiple of $2^{-t}$.
In particular, given a binary representation of $\gamma$ as $\gamma=\sum_{i\in\Zbb}a_i\cdot 2^i$ where each $a_i\in\bin$, we have
$$
\err(\gamma,t)=\min\cbra{\sum_{i<-t}a_i\cdot2^i,\sum_{i<-t}(1-a_i)\cdot2^i}.
$$
\end{definition}

It is easy to see that $0\le\err(\gamma,t)\le2^{-t-1}$.
A concrete non-trivial example is $\err(1/3,t)\ge2^{-t-2}$ for all $t\ge0$, since $1/3$ has binary representation $\sum_{i<0}2^{2i}$.

\begin{fact}\label{fct:binary_rep_err}
Let $f\colon\bin^m\to\bin^n$ be a $d$-local function.
Then the marginal distribution of $f(\Ucal^m)$ on any single output bit is $\err(\gamma,d)$-far from $\Ucal_\gamma^1$.
\end{fact}

\Cref{fct:binary_rep_err} already shows that $f(\Ucal^m)$ is $\err(\gamma,d)$-far from $\Ucal_\gamma^n$.
Our goal is to boost the distance closer to $1$.

\begin{theorem}\label{thm:locality_biased}
Let $f\colon\bin^m\to\bin^n$ be a $d$-local function, and let $0\le\gamma\le1$ be a parameter.
If $\err(\gamma,d)\ge\delta$ for some $\delta>0$, then
$$
\tvdist{f(\Ucal^m)-\Ucal_\gamma^n}\ge1-4\cdot\exp\cbra{-n\cdot\delta^{40d}}.
$$
\end{theorem}

We first consider a simple case where we can find many output bits that do not correlate with each other.
For example, this happens when every input bit influences few output bits.

\begin{definition}[$(d,r)$-Local Function]\label{def:dyadic_local}
We say $g\colon\bin^m\to\bin^n$ is a $(d,r)$-local function if $g$ is a $d$-local function with $r$ non-connected output bits.
\end{definition}

\begin{proposition}\label{prop:biased_more_local}
Let $g\colon\bin^m\to\bin^n$ be a $(d,r)$-local function.
Then 
$$
\tvdist{g(\Ucal^m)-\Ucal_\gamma^n}\ge1-2\cdot\exp\cbra{-\frac{\err(\gamma,d)^2\cdot r}2}.
$$
\end{proposition}
\begin{proof}
The bound is trivial when $r<1$. Hence we assume $r\ge1$.
By rearranging indices, we assume without loss of generality that $1,2,\ldots,r$ are non-connected output bits.
We will apply \Cref{lem:tvdist_after_product} with
$$
\Pcal=g(\Ucal^m),\quad
\Qcal=\Wcal=\Ucal_\gamma^n,\quad
S=[r].
$$

Observe that $\Pcal$ is a product distribution marginally on $S$, since $1,2,\ldots,r$ are non-connected.
Additionally by \Cref{fct:binary_rep_err}, we have
$$
\tvdist{\Pcal|_{\cbra i}-\Wcal|_{\cbra i}}\ge\err(\gamma,d)~\reflectbox{$\coloneqq$}~\eps.
$$
Then the desired bound follows from \Cref{lem:tvdist_after_product} with $\eta=1$.
\end{proof}

We next show that any $d$-local function $f$ can be made $(d,r)$-local by restricting a few input bits.

\begin{proposition}\label{prop:biased_structure}
Assume $\err(\gamma,d)>0$.
Let $f\colon\bin^m\to\bin^n$ be a $d$-local function with $d \ge 1$.
Then there exists a set $S\subseteq[m]$ such that any fixing of input bits in $S$ reduces $f$ to a $(d,r)$-local function $g$, where
$$
|S|\le\frac{\err(\gamma,d)^2\cdot r}4
\quad\text{and}\quad
r\ge n\cdot\pbra{\frac{\err(\gamma,d)^2}{16d}}^{2d+1}.
$$
\end{proposition}
\begin{proof}
Recall the graph theoretic view of the dependency relations in $f$. 
We apply \Cref{cor:non-adj_vtx_new} with $\beta=4/\err(\gamma,d)^2$ and $\lambda=(4d\beta)^{2d+1}$.
\end{proof}

Now we prove \Cref{thm:locality_biased}.
\begin{proof}[Proof of \Cref{thm:locality_biased}]
Recall that $\err(\gamma,d)\ge\delta>0$ and $\err(\gamma,d)\le2^{-d-1}$. We assume $d \ge 1$, as otherwise the theorem is trivial.
By \Cref{prop:biased_structure}, we find a set $S\subseteq[m]$ such that any fixing $\rho$ of input bits in $S$ reduces $f$ to a $(d,r)$-local function $f_\rho$ where
$$
|S|\le\frac{\delta^2\cdot r}4
\quad\text{and}\quad
r\ge\frac n{(16d/\delta^2)^{2d+1}}.
$$
Then for each $f_\rho$, we apply \Cref{prop:biased_more_local} and obtain that
$$
\tvdist{f_\rho(\Ucal^{[m]\setminus S})-\Ucal_\gamma^n}\ge1-2\cdot e^{-\delta^2\cdot r/2}.
$$
Note that $f(\Ucal^m)=\E_\rho\sbra{f_\rho(\Ucal^{[m]\setminus S})}$ where $\rho\sim\Ucal^S$.
By \Cref{lem:tvdist_after_conditioning} with $\cbra{f_\rho(\Ucal^{[m]\setminus S})}_\rho,\Ucal_\gamma^n$, and 
$$
\eps_1=2\cdot e^{-\delta^2\cdot r/2},\quad
\eps_2=\eps_3=0,\quad
\Ecal=\emptyset,
$$
we have
\begin{align*}
\tvdist{f(\Ucal^m)-\Ucal_\gamma^n}
&\ge1-\pbra{2^{|S|}+1}\cdot2\cdot e^{-\delta^2\cdot r/2}
\ge1-4\cdot e^{-\delta^2\cdot r/4}\\
&\ge1-4\cdot\exp\cbra{-\frac{\delta^2\cdot n}{4\cdot (16d/\delta^2)^{2d+1}}}\\
&\ge1-4\cdot\exp\cbra{-n\cdot\delta^{40d}}
\tag{since $d\ge1$ and $\delta\le2^{-d-1}$}
\end{align*}
as desired.
\end{proof}

\subsection{A Single Hamming Slice: The Non-Dyadic Case}\label{sec:single_non-dyadic}

The argument in \Cref{sec:biased} works beyond $\gamma$-biased distributions.
Here we generalize it to the Hamming slice setting, the proof of which introduces new ideas to handle non-product distributions and will be useful later.

Let $\Dcal_k$ be the uniform distribution of binary strings of length $n$ with Hamming weight $k$.
Define $\gamma=k/n$. 
Our goal here is to prove local function cannot sample from $\Dcal_k$ when $\gamma$ has large binary representation error. This is similar to \Cref{thm:locality_biased} which replaces $\Dcal_k$ by the $\gamma$-biased distribution.

\begin{theorem}\label{thm:locality_single_non-dyadic}
Let $f\colon\bin^m\to\bin^n$ be a $d$-local function, and let $1\le k\le n-1$ be an integer.
Define $\gamma=k/n$. 
If $\err(\gamma,d)\ge\delta$ for some $\delta>0$, then
$$
\tvdist{f(\Ucal^m)-\Dcal_k}\ge1-4\sqrt{2n}\cdot\exp\cbra{-n\cdot\delta^{40d}}.
$$
\end{theorem}

Following the previous framework, we first prove the distance bound for $(d,r)$-local functions analogous to \Cref{prop:biased_more_local}.
However unlike $\Ucal_\gamma^n$ there, we have $\Dcal_k$ here. 
Though marginally every bit in $\Dcal_k$ is exactly $\gamma$-biased $\Ucal_\gamma^1$, the distance between $\Dcal_k|_S$ and $\Ucal_\gamma^S$ enlarges quickly when $|S|$ grows.
Nevertheless, we can use $\Ucal_\gamma^n$ as a proxy between our distribution and $\Dcal_k$.
The crucial point is that $\Ucal_\gamma^n$ and $\Dcal_k$ are not far from each other in the multiplicative sense, though they have total variation distance roughly $1-1/\sqrt n$ (for constant $\gamma$).

\begin{proposition}\label{prop:single_non-dyadic_more_local}
Let $g\colon\bin^m\to\bin^n$ be a $(d,r)$-local function.
Then
$$
\tvdist{g(\Ucal^m)-\Dcal_k}\ge1-2\sqrt{2n}\cdot\exp\cbra{-\frac{\err(\gamma,d)^2\cdot r}2}.
$$
\end{proposition}
\begin{proof}
The bound is trivial when $r<1$. Hence we assume $r\ge1$.
By rearranging indices, we assume without loss of generality that $1,2,\ldots,r$ are non-connected output bits.
We will apply \Cref{lem:tvdist_after_product} with
$$
\Pcal=g(\Ucal^m),\quad
\Qcal=\Dcal_k,\quad
\Wcal=\Ucal_\gamma^n,\quad
S=[r].
$$
Note that $\Pcal$ is a product distribution on $S$.
By \Cref{fct:binary_rep_err}, we have
$$
\tvdist{\Pcal|_{\cbra i}-\Wcal|_{\cbra i}}\ge\err(\gamma,d)~\reflectbox{$\coloneqq$}~\eps.
$$
To apply \Cref{lem:tvdist_after_product}, it remains to bound $\eta\coloneqq\min_{x\in\supp\Qcal}\Wcal(x)/\Qcal(x)$.
For any $x\in\supp\Qcal$, we have
$$
\Qcal(x)=\frac1{\binom nk}=\frac1{\binom n{\gamma n}}
\quad\text{and}\quad
\Wcal(x)=\gamma^k\cdot(1-\gamma)^{n-k}=\gamma^{\gamma n}\cdot(1-\gamma)^{(1-\gamma)n}.
$$
By \Cref{fct:individual_binom}, we have
$$
\eta=\min_{x\in\supp\Qcal}\frac{\Wcal(x)}{\Qcal(x)}\ge\frac1{\sqrt{8\gamma n(1-\gamma)}}\ge\frac1{\sqrt{2n}}.
$$
Applying \Cref{lem:tvdist_after_product} gives the desired bound.
\end{proof}

\begin{remark}\label{rmk:single_non-dyadic_more_local}
It is possible to improve the construction of $\Wcal$ in \Cref{prop:single_non-dyadic_more_local} to get a better bound of $\eta$.
To see this, we can obtain $\Wcal$ as follows: first we sample bits in $S$ according to $\Ucal_\gamma^S$, then we complete the other coordinates $[n]\setminus S$ by the distribution of $\Qcal=\Dcal_k$ conditioned on the sampled bits in $S$.

As such, $\eta$ will be the minimum ratio of $\Ucal_\gamma^S(x)$ and $\Dcal_k|_S(x)$ for $x\in\bin^S$.
Note that if $|S|$ is not too large, then $\eta$ will not be too small.
For example, if $\gamma$ is constant and $|S| \ll n$, then one can show that $\eta\ge\Omega(1)$.
Since later in \Cref{prop:biased_structure} we indeed have relatively small $|S|$, this is the typical case that matters to us.

We choose our current presentation for simplicity. Moreover this improvement is only a factor of $\poly(n)$ in terms of applications after combining everything, ultimately subsumed by $\exp\cbra{-\Omega_d(n)}$.
\end{remark}

We use the same \Cref{prop:biased_structure} to convert $d$-local functions to $(d,r)$-local, and apply \Cref{lem:tvdist_after_conditioning} to put them together.

\begin{proof}[Proof of \Cref{prop:single_non-dyadic_more_local}]
The argument is almost identical to the proof of \Cref{thm:locality_biased}, except that now we have
$$
\eps_1=2\sqrt{2n}\cdot e^{-\delta^2\cdot r/2}.
$$
Combining \Cref{prop:biased_structure} and \Cref{lem:tvdist_after_conditioning}, we have
\begin{equation*}
\tvdist{f(\Ucal^m)-\Dcal_k}\ge1-4\sqrt{2n}\cdot\exp\cbra{-n\cdot\delta^{40d}}.
\tag*{\qedhere}
\end{equation*}
\end{proof}

\subsection{A Single Hamming Slice: The General Case}\label{sec:single_hamming_slice}

In the previous section, we showed that local functions cannot sample from $\Dcal_k$ when $\gamma=k/n$ has large binary representation error.
In particular, this shows $\Dcal_{n/3}$ is not locally sampleable.
In this section, we aim to address the general case of $\Dcal_k$, where we do not gain advantages from the non-dyadic numbers.

A concrete example is $k=n/2$. The coordinatewise-independent version of $\Dcal_{n/2}$ is simply $\Ucal^n$, which can be exactly sampled by a $1$-local function. However this does not seem to generalize: $\Dcal_{n/2}$ is $(1-\Theta(1/\sqrt n))$-far from $\Ucal^n$.
We will show that this is actually the best possible strategy.

We will prove the following more general statement, which works for all $o(n)\le k\le n/2$.
To build intuition, specific instantiations can be found in \Cref{thm:informal_single} and \Cref{thm:informal_single_1}.  

\begin{theorem}\label{thm:locality_single_hamming_slice}
There exists a universal constant $\kappa\ge1$ such that the following holds.
Let\footnote{By flipping zero and one, sampling from $\Dcal_k$ is equivalent to sampling from $\Dcal_{n-k}$. Therefore \Cref{thm:locality_single_hamming_slice} also holds for $k\ge n/2$ with $\gamma$ replaced by $1-\gamma$.}
$1\le k\le n/2$ and let $f\colon\bin^m\to\bin^n$ be a $d$-local function. 
Define $\gamma=k/n$ and let $\theta(n)$ be arbitrary.
If $\theta(n)\ge\kappa$ and
$$
d\le\log^*(\theta(n))/60
\quad\text{and}\quad
\log^*(1/\gamma)\le\log^*(\theta(n))/2,
$$
then
$$
\tvdist{f(\Ucal^m)-\Dcal_k}\ge1-\theta(n)/\sqrt n.
$$
\end{theorem}

We remark that the analysis in this section generalizes to multiple Hamming slices, with a loss of the union bound on top of the $\sqrt n$ that scales linearly with the number of slices.
This will be clear in the anticoncentration analysis (\Cref{lem:type-2_single}) which works equally well in the generalized setting.

In addition, the bounds in \Cref{thm:locality_single_hamming_slice} are asymptotically tight when we set $\theta(n)$ to be a fixed constant sufficiently large.
This is because the $2^{-t}$-biased distribution over $n$ bits is $\pbra{1-\Theta_t(1/\sqrt n)}$-close to $\Dcal_{n/2^t}$, where the former can be sampled by a $t$-local function.

To prove \Cref{thm:locality_single_hamming_slice}, we first consider a simpler setting where we are given a $d$-local function with many non-connected neighborhoods of small size.
For intuition, one can view it as the case where every input bit influences few output bits.

\begin{definition}[$(d,r,t)$-Local Function]\label{def:more_local}
We say $g\colon\bin^m\to\bin^n$ is a $(d,r,t)$-local function if $g$ is a $d$-local function with $r$ non-connected neighborhoods of size at most $t$.
\end{definition}

We remark that the notion of $(d,r,t)$-local generalizes the notion of $(d,r)$-local in the previous section. There, the analysis depended on individual bits having incorrect bias, but now we consider a more subtle exploitation. In particular, we need both the distribution over the neighborhood $N(i)$ to be close to uniform and resampling to not substantially change the sum. (This latter property implies the distribution of the sum conditioned on $i=0$ is approximately the same as conditioned on $i=1$.) However, this trade-off alone gives only a small error in total variation distance to amplify; we need many independent neighborhoods (rather than just independent bits). This independence follows from non-connectedness.

The following proposition concerns lower bounds for $(d,r,t)$-local functions. It is similar to \Cref{prop:single_non-dyadic_more_local} and will be proved later.

\begin{proposition}\label{prop:single_slice_more_local}
Let $g\colon\bin^m\to\bin^n$ be a $(d,r,t)$-local function and define $\Pcal_g=g(\Ucal^m)$.
Then there exists a universal constant $C\ge1$ such that either
$$
\tvdist{\Pcal_g-\Dcal_k}\ge1-C\sqrt n\cdot\exp\cbra{-\frac{\gamma^2\cdot r}{C\cdot t}}
\quad\text{or}\quad
\Pcal_g(\supp{\Dcal_k})\le\frac{C\cdot t^{1/4}}{\sqrt{\gamma\cdot r}}.
$$
\end{proposition}

Then we show that any $d$-local function $f$ can be turned into a $(d,r,t)$-local function $g$ by restricting a few input bits.
Note that this is for some values of $r$ and $t$, which might depend on the function $f$.
For intuition, one can think of $d,\gamma,C$ as constants, then we will obtain $r=\Omega(n)$ non-connected neighborhoods of size at most $t=O(1)$ by restricting way fewer than $r$ input bits. 
This result is similar to \Cref{prop:biased_structure}.

\begin{proposition}\label{prop:single_slice_find_structure}
Let $C\ge1$ be an integer parameter and let $f\colon\bin^m\to\bin^n$ be a $d$-local function with $d \ge 1$.
Then there exists a set $S\subseteq[m]$ such that any fixing of input bits in $S$ reduces $f$ to a $(d,r,t)$-local function $g$ where
$$
|S|\le\frac{\gamma^2\cdot r}{2C\cdot t}
\quad\text{and}\quad
r\ge\frac n{\tow_2(20d+\log^*(1/\gamma)+C)}
\quad\text{and}\quad
t\le\tow_2(20d+\log^*(1/\gamma)+C).
$$
\end{proposition}
\begin{proof}
Recall the graph theoretic view of the dependency relations in $f$. We will apply \Cref{cor:graph_elim_non-adj_neigh}.
Setting $F(x)=2C\cdot x/\gamma^2$ gives
$$
\tilde F(x)=\frac1d\cdot\exp\cbra{\frac{128d^5C\cdot x^3}{\gamma^2}}.
$$
Define $H(x)=2^{2^x}$ and let $L=10\cdot\log(d)+30\cdot\log(1/\gamma)+2\cdot\log(C)$.
By setting 
$$
\kappa=\lambda=\tow_2(20d+\log^*(1/\gamma)+C)\ge d\cdot H^{(2d+2)}(L),
$$
the conditions in \Cref{cor:graph_elim_non-adj_neigh} are satisfied, where we used the fact that $\gamma\le1/2$ and $d,C\ge1$.
This implies that \Cref{as:non-adj_neigh} holds for the dependency graph of $f$ with parameter $\lambda,\kappa,F$.
\end{proof}

Finally, we convert lower bounds for $(d,r,t)$-local functions to $d$-local functions via \Cref{lem:tvdist_after_conditioning} as before.

\begin{proof}[Proof of \Cref{thm:locality_single_hamming_slice}]
Let $C\ge1$ be the universal constant in \Cref{prop:single_slice_more_local}.
Without loss of generality we assume it is an integer.
Define $\kappa=\tow_2(60C)$.
If $d\le C$, then we can simply set $d=C$, since a $d'$-local function is also $d$-local if $d\ge d'$.
Since we assumed that $\theta(n)\ge\kappa=\tow_2(60C)$, setting $d=C$ still satisfies the condition $d\le\log^*(\theta(n))/60$.
From now on we safely assume $d\ge C$.

By \Cref{prop:single_slice_find_structure}, we find a set $S\subseteq[m]$ such that any fixing $\rho$ of input bits in $S$ reduces $f$ to a $(d,r,t)$-local function $f_\rho$ where
$$
|S|\le\frac{\gamma^2\cdot r}{2C\cdot t}
\quad\text{and}\quad
r\ge\frac n{\tow_2(20d+\log^*(1/\gamma)+C)}
\quad\text{and}\quad
t\le\tow_2(20d+\log^*(1/\gamma)+C).
$$
Since a $(d,r,t')$-local function is also $(d,r,t)$-local if $t\ge t'$, we may assume $t = \tow_2(20d+\log^*(1/\gamma)+C)$.
Now for each $f_\rho$, we apply \Cref{prop:single_slice_more_local} and obtain that
$$
\text{either}\quad
\tvdist{\Pcal_{f_\rho}-\Dcal_k}\ge1-C\sqrt n\cdot\exp\cbra{-\frac{\gamma^2\cdot r}{C\cdot t}}
\quad\text{or}\quad
\Pcal_{f_\rho}(\supp{\Dcal_k})\le\frac{C\cdot t^{1/4}}{\sqrt{\gamma\cdot r}}.
$$
Note that $f(\Ucal^m)=\E_\rho\sbra{f_\rho(\Ucal^{[m]\setminus S})}=\E_\rho[\Pcal_{f_\rho}]$ where $\rho\sim\Ucal^S$.
By \Cref{lem:tvdist_after_conditioning} with $\cbra{\Pcal_{f_\rho}}_\rho,\Dcal_k$, and 
$$
\eps_1=C\sqrt n\cdot\exp\cbra{-\frac{\gamma^2\cdot r}{C\cdot t}},\quad
\eps_2=\frac{C\cdot t^{1/4}}{\sqrt{\gamma\cdot r}},\quad
\eps_3=0,\quad
\Ecal=\supp{\Dcal_k},
$$
we have
\begin{align*}
\tvdist{f(\Ucal^m)-\Dcal_k}
&\ge1-\pbra{2^{|S|}+1}\cdot C\sqrt n\cdot\exp\cbra{-\frac{\gamma^2\cdot r}{C\cdot t}}-\frac{C\cdot t^{1/4}}{\sqrt{\gamma\cdot r}}\\
&\ge1-2C\sqrt n\cdot\exp\cbra{-\frac{\gamma^2\cdot r}{2C\cdot t}}-\frac{C\cdot t^{1/4}}{\sqrt{\gamma\cdot r}}
\tag{since $|S|\le\frac{\gamma^2\cdot r}{2C\cdot t}$}\\
&\ge1-2C\sqrt n\cdot\exp\cbra{-\frac n{t^3}}-\frac{t^2}{\sqrt n}
\tag{since $r\ge\frac{2Cn}{\gamma^2\cdot t^2}$ and $r\ge\frac{C^2n}{\gamma^2\cdot t^2}$}\\
&\ge1-\frac{3C\cdot t^3}{\sqrt n}.
\tag{since $e^{-x}\le 1/x$}
\end{align*}
Observe that
\begin{align*}
t&\le\tow_2(21d+\log^*(1/\gamma))
\tag{since $d\ge C$}\\
&\le\tow_2(\floorbra{\log^*(\theta)\cdot0.9})
\tag{since $d\le\log^*(\theta)/60$ and $\log^*(1/\gamma)\le\log^*(\theta)/2$}\\
&\le\tow_2(\log^*(\theta)-5)
\tag{since $\theta\ge\kappa=\tow_2(60C)\ge\tow_2(60)$}\\
&=\tow_2(\log^*(\log^{(5)}(\theta)))
\le2^{\log^{(5)}(\theta)}
\tag{since $\tow_2(\log^*(x))\le2^x$}\\
&\le(\theta/3C)^{1/3}.
\tag{since $\theta\ge\tow_2(60C)$}
\end{align*}
Hence $\tvdist{f(\Ucal^m)-\Dcal_k}\ge1-\theta/\sqrt n$ as desired.
\end{proof}

Now we prove lower bounds for $(d,r,t)$-local functions.

\begin{proposition*}[\Cref{prop:single_slice_more_local} Restated]
Let $g\colon\bin^m\to\bin^n$ be a $(d,r,t)$-local function and define $\Pcal_g=g(\Ucal^m)$.
Then there exists a universal constant $C\ge1$ such that either
$$
\tvdist{\Pcal_g-\Dcal_k}\ge1-C\sqrt n\cdot\exp\cbra{-\frac{\gamma^2\cdot r}{C\cdot t}}
\quad\text{or}\quad
\Pcal_g(\supp{\Dcal_k})\le\frac{C\cdot t^{1/4}}{\sqrt{\gamma\cdot r}}.
$$
\end{proposition*}

Recall that $\gamma=k/n\in[1/n,1/2]$.
Let $\eps\in[0,1]$ be a parameter to be optimized later.
For each neighborhood $N(i)=N_g(i)$ of size $s_i=|N(i)|$, we classify it into one of the following two cases:
\begin{itemize}
\item \textsf{Type-1.}
$\Pcal_g|_{N(i)}$ is not $\eps$-close to $\Ucal_\gamma^{s_i}$.
\item \textsf{Type-2.}
$\Pcal_g|_{N(i)}$ is $\eps$-close to $\Ucal_\gamma^{s_i}$.
\end{itemize}
Intuitively a \textsf{Type-1} neighborhood means the marginal $\Pcal_g$ on $N(i)$ is far from the $\gamma$-biased distribution. If we find many \textsf{Type-1} neighborhoods, we can prove the distance bound analogous to \Cref{prop:single_non-dyadic_more_local}.

On the other hand, the output distribution of a \textsf{Type-2} neighborhood is close to $\gamma$-biased, in which case the previous argument fails. Then we show that these neighborhoods are somewhat independent and it is unlikely for them to sum to a fixed value.

We first prove that $\Pcal_g$ is far from $\Dcal_k$ if there are many small non-connected \textsf{Type-1} neighborhoods. 
The intuition is that the local view of $\Dcal_k$ should be $\gamma$-biased.

\begin{lemma}\label{lem:type-1_single}
Assume there are $r'\ge1$ non-connected \textsf{Type-1} neighborhoods.
Then 
$$
\tvdist{\Pcal_g-\Dcal_k}\ge1-2\sqrt{2n}\cdot\exp\cbra{-\eps^2r'/2}.
$$
\end{lemma}
\begin{proof}
The proof is similar to the one for \Cref{prop:single_non-dyadic_more_local}. The only change is to work with non-connected neighborhoods instead of non-connected output bits.
By rearranging indices, we assume without loss of generality that $N(1),\ldots,N(r')$ are non-connected \textsf{Type-1} neighborhoods of sizes $s_1,\ldots,s_{r'}$.

We will apply \Cref{lem:tvdist_after_product}, for which we define $\Pcal,\Qcal,S,\Wcal$. Let $R=[n]\setminus\pbra{N(1)\cup\cdots N(r')}$ be the rest of the output coordinates.
\begin{itemize}
\item Define $\Pcal$ as $\Pcal_g$ but grouping each $N(i)$ for $i\in[r']$ and $R$ as coordinates.
That is, $\Pcal$ now is a distribution over a product space of $r'+1$ coordinates, where $\Pcal|_{\cbra{i}}=\Pcal_g|_{N(i)}$ is over $\bin^{s_i}$ for $i\in[r']$ and $\Pcal|_{\cbra{r'+1}}=\Pcal_g|_R$.
\item Define $\Qcal$ as $\Dcal_k$ but also grouping each $N(i)$ for $i\in[r']$ and $R$ as coordinates.
\item Define $S=[r']$.
\item Define $\Wcal$ as $\Ucal_\gamma^n$ but also grouping each $N(i)$ for $i\in[r']$ and $R$ as coordinates.
\end{itemize}
Observe that $\Wcal|_S=\bigtimes_{i\in S}\Ucal_\gamma^{s_i}$ is a product distribution and $\Pcal|_S=\bigtimes_{i\in S}\Pcal_g|_{N(i)}$ is also a product distribution since $N(1),\ldots,N(r')$ are non-connected.

Since each $N(i)$ here is \textsf{Type-1}, we have $\tvdist{\Pcal|_{\cbra{i}}-\Wcal|_{\cbra{i}}}\ge\eps$.
Note that $\Wcal$ is the $\gamma$-biased distribution with $\gamma\in[1/n,1/2]$.
Therefore for any $x\in\supp\Qcal$, we have
$$
\frac{\Wcal(x)}{\Qcal(x)}=\gamma^{\gamma n}\cdot(1-\gamma)^{(1-\gamma)n}\cdot\binom n{\gamma n},
$$
where we recall that $\Qcal=\Dcal_k=\Dcal_{\gamma n}$.
Then by the same calculation in the proof of \Cref{prop:single_non-dyadic_more_local}, we can set $\eta=1/\sqrt{2n}$ in \Cref{lem:tvdist_after_product} and obtain the desired bound.
\end{proof}

We note that the same improvement idea described in \Cref{rmk:single_non-dyadic_more_local} also works here.
Now we turn to the second case where there are many small non-connected \textsf{Type-2} neighborhoods.
In this case, we show that with high probability the sampled binary string from $\Pcal_g$ does not have Hamming weight $k$ via anticoncentration inequalities.

\begin{lemma}\label{lem:type-2_single}
Assume there are $r'\ge1$ non-connected \textsf{Type-2} neighborhoods of size at most $t$. 
If $\eps\le\frac{\gamma}{128\sqrt t}$, then
$$
\Pcal_g(\supp{\Dcal_k})\le O\pbra{\frac{t^{1/4}}{\sqrt{\gamma\cdot r'}}}.
$$
\end{lemma}
\begin{proof}
By rearranging indices, we assume without loss of generality that $N(1),\ldots,N(r')$ are non-connected \textsf{Type-2} neighborhoods of sizes $1\le s_1,\ldots,s_{r'}\le t$.
Recall that $I(i)=I_g(i)$ is the set of input bits that the $i$-th output bit depends on.
We sample a random $Z\sim\Ucal^m$ and let $(X_1,\ldots,X_n)=g(Z)$.

Let $R=[m]\setminus\pbra{I(1)\cup\cdots\cup I(r')}$.
Since $(X_1,\ldots,X_n)$ has its marginal equal to $\Pcal_g$, we have
\begin{equation}\label{eq:lem:type-2_single_1}
\Pcal_g(\supp{\Dcal_k})=\E_\rho\sbra{\Pr\sbra{\sum_{i\in[n]}X_i=k\mid\rho}},
\end{equation}
where $\rho\sim\bin^R$ and the condition on $\rho$ means that $Z_j=\rho_j$ for all $j\in R$.
We will use \Cref{fct:ushakov} to upper bound the above probability for most $\rho$'s.
To this end, we decompose $\sum_{i\in[n]}X_i$ into $K+\sum_{\ell\in[r']}\Delta_\ell$, where
$$
K=\sum_{i\notin N(1)\cup\cdots\cup N(r')}X_i
\quad\text{and}\quad
\Delta_\ell=\sum_{i\in N(\ell)}X_i.
$$
Observe that if $i\notin N(1)\cup\cdots\cup N(r')$, then $I(i)\subseteq R$ and thus $K$ is fixed given $\rho$.
For each $\ell\in[r']$ and $\rho\in\bin^R$, define the random variable
$$
p_{\rho,\ell}=\max_c\Pr\sbra{\Delta_\ell=c\mid\rho}.
$$
We will use \Cref{lem:anticoncentration_after_coupling_all} to prove an upper bound on $p_{\rho,\ell}$ in expectation.

\begin{claim}\label{clm:lem:type-2_single_1}
$\E_\rho\sbra{(p_{\rho,\ell})^2}\le1-\frac{\gamma}{64\sqrt t}$ holds for all $\ell\in[r']$.
\end{claim}

We first conclude the proof of \Cref{lem:type-2_single} assuming \Cref{clm:lem:type-2_single_1}.
Firstly by Jensen's inequality, we have
$$
\E_\rho\sbra{1-p_{\rho,\ell}}\ge1-\sqrt{\E_\rho\sbra{(p_{\rho,\ell})^2}}\ge1-\sqrt{1-\frac{\gamma}{64\sqrt t}}\ge\frac{\gamma}{128\sqrt t}.
$$
Since $N(1),\ldots,N(\ell)$ are non-connected, each $\Delta_\ell$ depends on disjoint parts of $Z$.
Thus each $p_{\rho,\ell}$ depends on disjoint parts of $\rho$, which means they are independent.
Since each $p_{\rho,\ell}\in[0,1]$, by \Cref{fct:chernoff} with $\delta=1/2$ we have
\begin{equation}\label{eq:lem:type-2_single_2}
\Pr_\rho\sbra{\sum_{\ell\in[r']}(1-p_{\rho,\ell})\le\frac12\cdot\frac{\gamma\cdot r'}{128\sqrt t}}\le\exp\cbra{-\frac{\gamma\cdot r'}{1024\sqrt t}}\le O\pbra{\frac{t^{1/4}}{\sqrt{\gamma\cdot r'}}}.
\end{equation}
We say $\rho$ is bad if the above event happens, and good otherwise.
Then we have
\begin{align*}
\Pcal_g(\supp{\Dcal_k})
&\le O\pbra{\frac{t^{1/4}}{\sqrt{\gamma\cdot r'}}}+\Pr\sbra{K+\sum_{\ell\in[r']}\Delta_\ell=k\mid\rho\text{ is good}}
\tag{by \Cref{eq:lem:type-2_single_1} and \Cref{eq:lem:type-2_single_2}}\\
&\le O\pbra{\frac{t^{1/4}}{\sqrt{\gamma\cdot r'}}}+O\pbra{\E_\rho\sbra{\frac1{\sqrt{\sum_{\ell\in[r']}(1-p_{\rho,\ell})}}\mid\rho\text{ is good}}}
\tag{by \Cref{fct:ushakov}}\\
&=O\pbra{\frac{t^{1/4}}{\sqrt{\gamma\cdot r'}}}
\end{align*}
as desired.

Now we prove \Cref{clm:lem:type-2_single_1}.
\begin{proof}[Proof of \Cref{clm:lem:type-2_single_1}]
Recall that $(X_1,\ldots,X_n)=g(Z)$ for a random $Z\sim\Ucal^m$, and let $I_\ell=\bigcup_{i\in N(\ell)}I(i)$.
Then $\Delta_\ell=\sum_{i\in N(\ell)}X_i$ depends only on bits of $Z$ in $I_\ell$.
Since $N(1),\ldots,N(r')$ are non-connected, $I_\ell\cap I(i)=\emptyset$ holds for all $i\neq\ell$.
This means the distribution of $\Delta_\ell$ conditioned on $Z_j$'s for $j\in[m]\setminus I(\ell)$ is the same as if we only condition on $Z_j$'s for $j\in R=[m]\setminus\pbra{I(1)\cup\cdots\cup I(r')}\supseteq I_\ell\setminus I(\ell)$.
Therefore
\begin{equation}\label{eq:clm:lem:type-2_single_1_1}
\E_\rho\sbra{(p_{\rho,\ell})^2}
=\E_{Z_j:j\notin I(\ell)}\sbra{\max_c\Pr\sbra{\sum_{i\in N(\ell)}X_i=c\mid Z_j:j\notin I(\ell)}^2}.
\end{equation}

We sample $Z'\sim\Ucal^m$ conditioned on $Z'_j=Z_j$ for all $j\notin I(\ell)$.
In other words, we randomly flip bits in $I(\ell)$ of $Z$ to obtain $Z'$.
Define $(Y_1,\ldots,Y_n)=g(Z')$.
Then for any value $c$, we have
\begin{align*}
\Pr\sbra{\sum_{i\in N(\ell)}X_i=c\mid Z_j\colon j\notin I(\ell)}^2
&=\Pr\sbra{\sum_{i\in N(\ell)}X_i=\sum_{i\in N(\ell)}Y_i=c\mid Z_j\colon j\notin I(\ell)}
\tag{$X$'s and $Y$'s are conditionally independent}\\
&\le\Pr\sbra{\sum_{i\in N(\ell)}X_i=\sum_{i\in N(\ell)}Y_i\mid Z_j\colon j\notin I(\ell)}.
\end{align*}
Putting into \Cref{eq:clm:lem:type-2_single_1_1}, we have $\E_\rho\sbra{(p_{\rho,\ell})^2}
\le\Pr\sbra{\sum_{i\in N(\ell)}X_i=\sum_{i\in N(\ell)}Y_i}$.

By rearranging indices, we assume $N(\ell)=[\ell]$.
Now we apply \Cref{lem:anticoncentration_after_coupling_all} to $(A,B,C,D)$ with $q=8\ceilbra{\sqrt{\gamma t}}$, where $A=X_\ell,C=Y_\ell$ and $B=(X_1,\ldots,X_{\ell-1}),D=(Y_1,\ldots,Y_{\ell-1})$.
This holds since $I(\ell)$ is resampled in $Z'$, which decouples $A=X_\ell$ from $(C,D)$ and $C$ from $(A,B)$. In addition, $(A,B),(C,D)$ have the same marginal distribution of $\Pcal_g|_{N(\ell)}$, which is of \textsf{Type-2}, i.e., $\eps$-close to $\Ucal_\gamma^{s_\ell}$.
Since $s_\ell\le t$ and
$$
\frac\gamma{4q}\cdot2^{-50\gamma(t-1)/q^2}
\ge\frac{\gamma}{32\ceilbra{\sqrt{\gamma t}}}\cdot\frac12
\ge\frac\gamma{128\sqrt t}
\ge\eps,
$$
\Cref{lem:anticoncentration_after_coupling_all} implies
\begin{align*}
\E_\rho\sbra{(p_{\rho,\ell})^2}
&\le\Pr\sbra{\sum_{i\in N(\ell)}X_i=\sum_{i\in N(\ell)}Y_i}
=\Pr\sbra{A+|B|=C+|D|}\\
&\le\Pr\sbra{A+|B|\equiv C+|D|\Mod q}
\le1-\frac{\gamma}{64\sqrt t}
\end{align*}
as desired.
\end{proof}
\end{proof}

At this point, we are ready to prove \Cref{prop:single_slice_more_local}.
\begin{proof}[Proof of \Cref{prop:single_slice_more_local}]
Firstly we note that the bound trivially holds when $r<1$. Hence we assume now $r\ge1$.
We set $\eps=\frac{\gamma}{64\sqrt t}$ and let $C$ be a universal constant sufficiently large.
By \Cref{def:more_local}, there are $r$ non-connected neighborhoods of size at most $t$.
If $\ceilbra{r/2}$ of them are \textsf{Type-1}, then we apply \Cref{lem:type-1_single} with $r'=\ceilbra{r/2}$ and obtain
$$
\tvdist{\Pcal_g-\Dcal_k}
\ge1-2\sqrt{2n}\cdot\exp\cbra{-\eps^2r/4}
\ge1-C\sqrt n\cdot\exp\cbra{-\frac{\gamma^2\cdot r}{C\cdot t}}.
$$
Otherwise there are $\ceilbra{r/2}$ of \textsf{Type-2}, and we apply \Cref{lem:type-2_single} with $r'=\ceilbra{r/2}$ to obtain
\begin{equation*}
\Pcal_g(\supp{\Dcal_k})\le\frac{C\cdot t^{1/4}}{\sqrt{\gamma\cdot r}}.
\tag*{\qedhere}
\end{equation*}
\end{proof}

\subsection{Periodic Hamming Slices}\label{sec:periodic_hamming_slice}

In the last section, we proved lower bounds for sampling a single Hamming slice.
Our technique is robust enough that it also works for uniform distributions over multiple Hamming slices.
Here we illustrate with periodic Hamming slices.

Let $q\ge3$ be an integer, and let $\Lambda\subseteq\Zbb/q\Zbb$ be a non-empty set.
We define the distribution $\Dcal_{q,\Lambda}$ to be the uniform distribution over $x\in\bin^n$ conditioned on $|x|\bmod q\in\Lambda$.
We will show that, under moderate conditions on $q$ and $\Lambda$, local functions cannot effectively sample from $\Dcal_{q,\Lambda}$.

\begin{theorem}\label{thm:mod_slice}
Let $q\ge3$ be an integer, and let $\Lambda\subseteq\Zbb/q\Zbb$ not empty.
Define $\eta=|\supp{\Dcal_{q,\Lambda}}|\cdot2^{-n}$.
Let $f\colon\bin^m\to\bin^n$ be a $d$-local function.
Then
$$
\tvdist{f(\Ucal^m)-\Dcal_{q,\Lambda}}\ge
1-\frac{6q}\eta\cdot\exp\cbra{-\frac n{q^2\cdot\tow_2(18d)}}-\begin{cases}
|\Lambda|/q & q\text{ is odd},\\
2\cdot\max\cbra{|\Lambda_\textsf{even}|,|\Lambda_\textsf{odd}|}/q & q\text{ is even.}
\end{cases}
$$
\end{theorem}

We note the following simple lower bound for $\eta$, which implies that \Cref{thm:mod_slice} gives non-trivial bounds for all $q\le O_d(\sqrt n)$.

\begin{claim}\label{clm:mod_slice_eta}
$\eta\ge2^{-q^2/n}/\sqrt{2n}$.
\end{claim}
\begin{proof}
The bound is trivial when $q=n$. Hence now we assume $q\le n-1$.
Since $\Lambda\neq\emptyset$, at least one Hamming slice with weight in $[(n-q)/2,(n+q)/2]$ will be included in $\supp{\Dcal_{q,\Lambda}}$.
By \Cref{fct:individual_binom} and \Cref{fct:entropy}, we obtain
\begin{equation*}
\eta\ge2^{-n}\cdot\binom n{(n-q)/2}\ge2^{\pbra{\Hcal\pbra{\frac12-\frac q{2n}}-1}\cdot n}/\sqrt{2n}
\ge2^{-q^2/n}/\sqrt{2n}.
\tag*{\qedhere}
\end{equation*}
\end{proof}

We also remark that the bound in \Cref{thm:mod_slice} is essentially tight for $q$ not extremely large:
\begin{itemize}
\item If $q=2$, then we can perfectly produce $\Dcal_{2,\cbra{0}}$ by a $2$-local function $(x_1\oplus x_2,x_2\oplus x_3,\ldots,x_{n-1}\oplus x_n,x_n\oplus x_1)$, and similarly for $\Dcal_{2,\cbra1}$.
\item If $q\ge3$ is odd, then we can produce $\Ucal^n$ by a $1$-local function, which hits $\Lambda$ with probability roughly $|\Lambda|/q$.
\item If $q\ge3$ is even, then we can produce $\Dcal_{2,\cbra0}$ as described above, which is roughly uniform over even numbers in $\Zbb/q\Zbb$ after modulo $q$. Thus we hit $\Lambda_\textsf{even}$ with probability $|\Lambda_\textsf{even}|/(q/2)$. Similar construction using $\Dcal_{2,\cbra1}$ will achieve the distance bound $1-|\Lambda_\textsf{odd}|/(q/2)$.
\end{itemize}

The proof of \Cref{thm:mod_slice} follows the same paradigm as the previous section.
We first give bounds for $(d,r,t)$-local functions. The proof of the following proposition is presented at the end of the section.

\begin{proposition}\label{prop:mod_slice_more_local}
Let $g\colon\bin^m\to\bin^n$ be a $(d,r,t)$-local function and define $\Pcal_g=g(\Ucal^m)$.
Then either
$$
\tvdist{\Pcal_g-\Dcal_{q,\Lambda}}\ge1-\frac2\eta\cdot\exp\cbra{-2^{-28t+19}\cdot r}.
$$
or
$$
\Pcal_g(\supp{\Dcal_{q,\Lambda}})\le2q\cdot\exp\cbra{-\frac{r\cdot2^{-14t+10}}{q^2}}+\begin{cases}
|\Lambda|/q & q\text{ odd},\\
2\cdot\max\cbra{|\Lambda_\textsf{even}|,|\Lambda_\textsf{odd}|}/q & q\text{ even,}
\end{cases}
$$
\end{proposition}

Then we prove the following graph elimination result tailored for the parameters in \Cref{prop:mod_slice_more_local}.

\begin{proposition}\label{prop:mod_slice_structure}
There exists a set $S\subseteq[m]$ such that any fixing of input bits in $S$ reduces $f$ to a $(d,r,t)$-local function $g$ where
$$
|S|\le\frac r{2^{28t-18}}
\quad\text{and}\quad
r\ge\frac n{\tow_2(16d)}
\quad\text{and}\quad
t\le\tow_2(16d).
$$
\end{proposition}
\begin{proof}
The statement is trivial when $d=0$ since then we can set $S=\emptyset,r=n,t=0$.
For $d\ge1$, we apply \Cref{cor:graph_elim_non-adj_neigh}.
Set $F(x)=2^{28t-18}$.
Then
$$
\tilde F(x)=\frac1d\cdot\exp\cbra{32d^4x^2\cdot2^{56dx-18}}.
$$
Define $H(x)=2^{2^{2^x}}$ and let $L=10d$.
By setting 
$$
\kappa=\lambda=\tow_2(16d)\ge d\cdot H^{(2d+2)}(L),
$$
the conditions in \Cref{cor:graph_elim_non-adj_neigh} are satisfied.
This implies that \Cref{as:non-adj_neigh} holds for the dependency graph of $f$ with parameter $\lambda,\kappa,F$.
\end{proof}

Finally we use the above graph elimination results to lift the lower bounds of $(d,r,t)$-local functions to $d$-local functions.
\begin{proof}[Proof of \Cref{thm:mod_slice}]
We assume $d \ge 1$, as otherwise $f$ must be constant, and one can verify the bound holds. By \Cref{prop:mod_slice_structure}, we find a set $S\subseteq[m]$ such that any fixing $\rho$ of input bits in $S$ reduces $f$ to a $(d,r,t)$-local function $f_\rho$ where
$$
|S|\le\frac r{2^{28t-18}}
\quad\text{and}\quad
r\ge\frac n{\tow_2(16d)}
\quad\text{and}\quad
t\le\tow_2(16d).
$$
Now for each $f_\rho$, we apply \Cref{prop:mod_slice_more_local} and obtain that either
$$
\tvdist{\Pcal_{f_\rho}-\Dcal_{q,\Lambda}}\ge1-\underbrace{\frac2\eta\cdot\exp\cbra{-2^{-28t+19}\cdot r}}_{\eps_1}
$$
or
$$
\Pcal_{f_\rho}(\supp{\Dcal_{q,\Lambda}})\le\eps_2\coloneqq 2q\cdot\exp\cbra{-\frac{r\cdot2^{-14t+10}}{q^2}}+\begin{cases}
|\Lambda|/q & q\text{ is odd},\\
2\cdot\max\cbra{|\Lambda_\textsf{even}|,|\Lambda_\textsf{odd}|}/q & q\text{ is even.}
\end{cases}
$$
Note that $f(\Ucal^m)$ is a convex combination of $\Pcal_{f_\rho}$'s.
By \Cref{lem:tvdist_after_conditioning} with $\cbra{\Pcal_{f_\rho}}_\rho,\Dcal_{q,\Lambda}$, and $\eps_3=0,\Ecal=\supp{\Dcal_{q,\Lambda}}$ and $\eps_1,\eps_2$ defined above, we have
\begin{align*}
\tvdist{f(\Ucal^m)-\Dcal_{q,\Lambda}}
&\ge1-\pbra{2^{|S|}+1}\cdot\eps_1-\eps_2
\ge1-\frac4\eta\cdot\exp\cbra{-\frac r{2^{28t-18}}}-\eps_2\\
&\ge1-\frac{6q}\eta\cdot\exp\cbra{-\frac r{q^2\cdot2^{28t-18}}}-\begin{cases}
|\Lambda|/q & q\text{ is odd},\\
2\cdot\max\cbra{|\Lambda_\textsf{even}|,|\Lambda_\textsf{odd}|}/q & q\text{ is even.}
\end{cases}
\end{align*}
Since $t\le\tow_2(16d)$ and $r\ge n/\tow_2(16d)$, we can bound
$$
\frac r{2^{28t-18}}\ge\frac n{\tow_2(16d)\cdot2^{28\cdot\tow_2(16d)-18}}\ge\frac n{\tow_2(18d)},
$$
which gives the desired bound in \Cref{thm:mod_slice}.
\end{proof}

Now we prove \Cref{prop:mod_slice_more_local}. The road-map is similar to the proof of \Cref{prop:single_slice_more_local}: we first classify each non-connected neighborhood dependent on whether its marginal distribution is far from unbiased.
If most are far, we use \Cref{lem:tvdist_after_product} to show the distance bound.
Otherwise we use local limit results to show that with certain probability the Hamming weight modulo $q$ cannot fall into $\Lambda$.

Let $\eps\in[0,1]$ be a parameter to be optimized later.
For each neighborhood $N(i)=N_g(i)$ of size $s_i=|N(i)|$, we classify it into one of the following two cases:
\begin{itemize}
\item \textsf{Type-1.} $\Pcal_g|_{N(i)}$ is not $\eps$-close to $\Ucal^{s_i}$.
\item \textsf{Type-2.} $\Pcal_g|_{N(i)}$ is $\eps$-close to $\Ucal^{s_i}$.
\end{itemize}

By almost identical reasoning as \Cref{lem:type-1_single} (except that we fix $\gamma=1/2$ here), we obtain a large distance bound when there are many small \textsf{Type-1} neighborhoods.

\begin{lemma}\label{lem:type-1_mod}
Assume there are $r'\ge1$ non-connected \textsf{Type-1} neighborhoods.
Then
$$
\tvdist{\Pcal_g-\Dcal_{q,\Lambda}}\ge1-\frac2\eta\cdot\exp\cbra{-\eps^2r'/2}.
$$
\end{lemma}

Now we turn to the second case where we have many small \textsf{Type-2} neighborhoods.
The analysis for this setting is also similar to the proof of \Cref{lem:type-2_single}, except that we now use a local limit theorem in additive groups instead of anticoncentration inequalities over the real numbers.

\begin{lemma}\label{lem:type-2_mod}
Assume there are $r'\ge1$ non-connected \textsf{Type-2} neighborhoods of size at most $t$.
If $\eps\le2^{-14t+10}$, then
$$
\Pcal_g(\supp{\Dcal_{q,\Lambda}})\le2q\cdot\exp\cbra{-\frac{r'\cdot2^{-14t+11}}{q^2}}+\begin{cases}
|\Lambda|/q & q\text{ odd},\\
2\cdot\max\cbra{|\Lambda_\textsf{even}|,|\Lambda_\textsf{odd}|}/q & q\text{ even,}
\end{cases}
$$
\end{lemma}
\begin{proof}
We inherit the notation $(X_i)_{i\in[n]},\rho,K,(\Delta_\ell)_{\ell\in[r']}$ from \Cref{lem:type-2_single} with one minor change:
here for each integer $a\ge3$, we define
$$
p_{\rho,a,\ell}=\max_{c\in\Zbb}\Pr\sbra{\Delta_\ell \equiv c \Mod{a}\mid\rho}.
$$
By a similar analysis as in \Cref{clm:lem:type-2_single_1}, we obtain the following claim.
\begin{claim}\label{clm:lem:type-2_mod_1}
$\E_\rho[(p_{\rho,a,\ell})^2]\le1-2^{-14t+11}$ holds for any $\ell\in[r']$ and $a\ge3$.
\end{claim}
\begin{proof}
We only highlight the difference from the proof of \Cref{clm:lem:type-2_single_1}.
We apply \Cref{lem:anticoncentration_after_coupling_all} with $\gamma=1/2$ and modulus $\bar q=\min\cbra{a,t+1}$ here.
Then $2\le\bar q\le t+1$ and
$$
\frac\gamma{4\bar q}\cdot2^{-50\gamma(t-1)/{\bar q}^2}
\ge\frac1{8(t+1)}\cdot2^{-13(t-1)}
\ge2^{-14t+10}\ge\eps,
$$
which implies $\E_\rho[(p_{\rho,a,\ell})^2]\le1-2^{-14t+11}$ by \Cref{lem:anticoncentration_after_coupling_all}.
\end{proof}

Then we have $\E_\rho[1-p_{\rho,a,\ell}]\ge2^{-14t+10}$.
As before, since $p_{\rho,a,\ell}$'s are independent for fixed $a$, by \Cref{fct:chernoff} we have
$$
\Pr\sbra{\sum_{\ell\in[r']}(1-p_{\rho,a,\ell})\le2^{-14t+9}\cdot r'}\le\exp\cbra{-2^{-14t+6}\cdot r'}.
$$
Now we say $\rho$ is bad if for some $a\ge3$ dividing $q$ the above event happens, and good otherwise.
Since there are at most $q$ possible such $a$'s, by the union bound we have
\begin{equation}\label{eq:lem:type-2_mod_1}
\Pr\sbra{\rho\text{ is bad}}\le q\cdot\exp\cbra{-2^{-14t+6}\cdot r'}.
\end{equation}
Thus,
\begin{align*}
&\Pcal_g(\supp{\Dcal_{q,\Lambda}}\\
\le&\Pr\sbra{\rho\text{ is bad}}+\Pr\sbra{K+\sum_{\ell\in[r']}\Delta_\ell\bmod q\in\Lambda\mid\rho\text{ is good}}\\
\le& q\cdot\exp\cbra{-2^{-14t+6}r'}+q\cdot\exp\cbra{-\frac{2^{-14t+11}r'}{q^2}}+\begin{cases}
|\Lambda|/q & q\text{ odd},\\
2\cdot\max\cbra{|\Lambda_\textsf{even}|,|\Lambda_\textsf{odd}|}/q & q\text{ even,}
\end{cases}
\end{align*}
where for the last inequality we use \Cref{eq:lem:type-2_mod_1} and \Cref{lem:mod_llt} with $L=2^{-14t+10}\cdot r'$ and the observation that $\Lambda$ shifted by $K$ still has the same maximum of even and odd numbers.
\end{proof}

Finally we choose $\eps$ and conclude the proof of \Cref{prop:mod_slice_more_local}.
\begin{proof}[Proof of \Cref{prop:mod_slice_more_local}]
Firstly the bound is trivial if $r<1$.
For $r\ge1$, we set $\eps=2^{-14t+10}$.
By \Cref{def:more_local}, there are $r$ non-connected neighborhoods of size at most $t$.
If $\ceilbra{r/2}$ of them are \textsf{Type-1}, then we apply \Cref{lem:type-1_mod} with $r'=\ceilbra{r/2}$ and obtain
$$
\tvdist{\Pcal_g-\Dcal_{q,\Lambda}}\ge1-\frac2\eta\cdot\exp\cbra{-2^{-28t+19}\cdot r}.
$$
Otherwise there are $\ceilbra{r/2}$ of \textsf{Type-2}, and we apply \Cref{lem:type-2_mod} with $r'=\ceilbra{r/2}$ to obtain
\begin{equation*}
\Pcal_g(\supp{\Dcal_{q,\Lambda}})\le2q\cdot\exp\cbra{-\frac{r\cdot2^{-14t+10}}{q^2}}+\begin{cases}
|\Lambda|/q & q\text{ odd},\\
2\cdot\max\cbra{|\Lambda_\textsf{even}|,|\Lambda_\textsf{odd}|}/q & q\text{ even,}
\end{cases}
\tag*{\qedhere}
\end{equation*}
\end{proof}
\section{Upper Bounds}\label{sec:upper_bounds}

In this section, we provide upper bounds on the locality of functions sampling specific distributions. Upper bounds for $\Dcal_k$  when the input length is nearly the information-theoretic minimum appear in \cite{viola2012complexity} with weak error bounds, as well as for $\Ucal_{\gamma}^n$ with strong error bounds. To the best of our knowledge, no upper bounds appear in the literature for $\Dcal_k$ with arbitrarily small error, so we provide two incomparable ones in \Cref{subsec:upper_singleslice}.

\begin{theorem}\label{thm:combinedSliceUpper}
    For all $k \le n$, there exists a $d$-local
    Boolean function $f : \bin^m \to \bin^n$ such that $f(\Ucal^m)$ is $\eps$-close to $\Dcal_{k}$, where
    $$
    d=O\left(\min\left\{\log(n)\cdot \log(n/\eps), \log(n/k) + \log^2(k/\eps)\right\}\right).
    $$
\end{theorem}

\Cref{subsec:upper_periodslice} uses the previous theorem to prove an upper bound on $\Dcal_{q, \Lambda}$. This bound, as well as the sampling procedure, is in \cite{viola2012complexity}, and we include them only for completeness.
\begin{theorem}[\cite{viola2012complexity}]\label{thm:upper_mod}
    For all $q \in \Nbb$ and non-empty $\Lambda \subseteq \Zbb/q\Zbb$, there exists an $O(q^2 \cdot \log^2(n/\eps))$-local Boolean function $f : \bin^m \to \bin^n$ such that $f(\Ucal^m)$ is $\eps$-close to $\Dcal_{q, \Lambda}$.
\end{theorem}

Throughout, we will need to sample from various distributions. The following standard lemma allows us to do so approximately with low locality.
\begin{lemma}\label{lem:distr_approx}
    Any discrete distribution of support size $m$ can be approximated to $\eps$ error in total variation distance with $\lceil \log(m/\eps) \rceil$ uniform random bits.
\end{lemma}
\begin{proof}
    Let $\Dcal$ be a discrete distribution with support $\{x_1, \ldots, x_m\}$, and let $B \coloneq 2^{\lceil \log(m/\eps) \rceil}$. Furthermore, define $\Tcal$ to be the distribution by discretizing $\Dcal$'s probability density function:
    \[ \Tcal(x_i) = \begin{cases} 
          \frac{1}{B} \cdot \lfloor B \cdot \Dcal(x_i) \rfloor & i \in [m-1], \\
          1 - \sum_{i=1}^{m-1} \Tcal(x_i) & i = m.
       \end{cases}
    \]
    Then,
    \[
    \tvdist{\Tcal - \Dcal} = \Tcal(x_m) - \Dcal(x_m)= 1 - \pbra{\sum_{i=1}^{m-1}\frac{1}{B} \cdot \lfloor B \cdot \Dcal(x_i) \rfloor} - \Dcal(x_m) \le \eps. \qedhere
    \]
\end{proof}

\paragraph*{Decision Forest Depth.}
Our work focuses on quantifying a circuit's complexity by its locality. Another common measure is \textit{decision forest depth}, which can be viewed as an adaptive variant of locality. A function $f : \bin^m \to \bin^n$ is computable by a depth-$d$ decision forest if every output bit of $f$ is the result of a depth-$d$ decision tree of the input bits. 
Observe that such a function is $2^d$-local.
This notion was studied in earlier works, such as \cite{viola2012complexity, filmus2023sampling}. In fact, several of their aforementioned lower bounds also hold in this stronger model.

It is known that there exists a sampler for $\Dcal_{k}$ of depth $O(\log n)$ \cite{viola2012complexity} using a (randomized) switching network construction of \cite{czumaj2015random}.
Our \Cref{thm:combinedSliceUpper} provides comparable bounds in terms of the weaker locality parameter. 
Viola also gives a sampler for $\Dcal_{q, \Lambda}$ of depth $O(\log^2 n)$. In fact, the bound holds for more general distributions, and we direct the reader to \cite{viola2012complexity} for these and more related results.

\subsection{A Single Hamming Slice}\label{subsec:upper_singleslice}

In this subsection, we prove \Cref{thm:combinedSliceUpper} in two parts: \Cref{thm:upper_largeslice} gives an $O(\log^2 n)$ bound, while \Cref{thm:upper_smallslice} gives an $O(\log(n/k) + \log^2 k)$ bound. 

\begin{theorem}\label{thm:upper_largeslice}
    For all $k \le n$, there exists an $O(\log(n)\cdot \log(n/\eps))$-local Boolean function $f : \bin^m \to \bin^n$ such that $f(\Ucal^m)$ is $\eps$-close to $\Dcal_{k}$.
\end{theorem}
\begin{proof}
    We sample $\Dcal_{k}$ by iteratively partitioning the interval $[n]$ into two (essentially) equally sized intervals, which contain $a$ and $k-a$ ones to place, respectively, for an $a$ sampled from the appropriate hypergeometric distribution. More precisely, $\Dcal_{k}$ can be viewed as the resulting distribution of the following process.
    \begin{enumerate}
        \item Consider the interval $S \coloneq [n]$. There are $\ell \coloneq k$ ones to place inside, and the remaining entries are zeros.

        \item Partition $S$ into two consecutive intervals $S_1$ and $S_2$ of sizes $\lfloor |S|/2 \rfloor$ and $\lceil |S|/2 \rceil$.

         \item Pick an element $a \in \{0, \ldots, \ell\}$ according to the hypergeometric distribution $\Hcal_{|\Scal|, \ell, |\Scal_1|}$:
         \[
         \Pr[a = r] = \frac{\binom{\ell}{r} \cdot \binom{|S|-\ell}{|S_1|-r}}{\binom{|S|}{|S_1|}}.
         \]

        \item Repeat the process with $S_1$ and $S_2$ given $a$ and $\ell-a$ ones, respectively, to place in their intervals.
    \end{enumerate}

    We would like to approximate $\Dcal_{k}$ with a distribution $\Lcal_{k} \coloneq f(\Ucal^m)$ produced by a function $f : \bin^m \to \bin^n$ of small locality. By \Cref{lem:distr_approx}, we can sample from $\Hcal_{|\Scal|, \ell, |\Scal_1|}$ with at most $\lceil \log(k/\eps') \rceil$ bits of locality and error $\eps'$. The key insight is that each output bit depends only on the interval containing it in each of the at most $\lceil \log n \rceil$ steps, so the total locality is $O(\log(n) \cdot \log(k/\eps'))$.

    Since each low locality approximation to the hypergeometric distribution sampled in the above process is $\eps'$-close to the true distribution, $\tvdist{\Lcal_{k}-\Dcal_{k}} \le \eps' n$ by the union bound. Setting $\eps'=\eps/n$, we find that $f$ is an $O(\log(n)\cdot \log(n/\eps))$-local function with $f(\Ucal^m)$ $\eps$-close to $\Dcal_{k}$.
\end{proof}

We can use the above theorem to derive a tighter bound in the case of small $k$. 

\begin{theorem}\label{thm:upper_smallslice}
    For all $k \le n$, there exists an $O(\log(n/k) + \log^2(k/\eps))$-local Boolean function $f : \bin^m \to \bin^n$ such that $f(\Ucal^m)$ is $\eps$-close to $\Dcal_{k}$.
\end{theorem}
For clarity in the proof, we will use $\Dcal_{s,t}$ to denote the uniform distribution over $x \in \bin^t$ of Hamming weight $s$. When only one parameter is provided, $\Dcal_s$ denotes the uniform distribution over $x \in \bin^n$ of Hamming weight $s$, as in prior contexts.
\begin{proof}
    Assume $k \le \sqrt{\eps n}$, as otherwise the bound follows from \Cref{thm:upper_largeslice}.
    We first describe a randomized process to approximately generate $\Dcal_{k}$: split $[n]$ into (essentially) equally sized parts, pick $k$ at random, and place a single 1 randomly in each of the chosen parts. The number of parts in the original split determines the accuracy of this approximation. More precisely, let $\Pcal_{k}$ be the resulting distribution of the following process.
    \begin{enumerate}
        \item Divide $[n]$ into $t$ intervals, each of size $\lceil n/t \rceil$ or $\lfloor n/t \rfloor$ for a $t$ to be determined later.
        \item Pick $k$ distinct intervals according to the distribution $\Dcal_{k,t}$.
        \item Put a 1 in each of the selected intervals according to the distribution $\Dcal_{1,\lceil n/t \rceil}$ or $\Dcal_{1,\lfloor n/t \rfloor}$, depending on the block size. (The remaining entries are zeros.)
    \end{enumerate}

    By direct calculation, we find 
    \begin{align*}
        \tvdist{\Pcal_{k} - \Dcal_{k}} &= \Pr_{x \sim \Dcal_{k}}\sbra{x \text{ has two ones in the same interval}} \\
        &\le \E_{x \sim \Dcal_{k}}[\text{number of pairs of ones in $x$ in the same interval}] \\
        &\le k^2/t.
    \end{align*}
    Hence, as long as $t \ge \lceil k^2/\eps \rceil$, the two distributions are $\eps$-close.

    We would like to approximate $\Pcal_{k}$ with a distribution $\Lcal_{k} \coloneq f(\Ucal^m)$ produced by a function $f : \bin^m \to \bin^n$ of small locality. Each output bit $b$ depends on the choice of intervals and the location of the 1 within the interval containing $b$. Using \Cref{thm:upper_largeslice}, we can sample from $\Dcal_{k,t}$ with $O(\log(t)\cdot\log(t/\eps_1))$ bits of locality and error $\eps_1$. Similarly using \Cref{lem:distr_approx}, we can sample from $\Dcal_{1,\lceil n/t \rceil}$ or $\Dcal_{1,\lfloor n/t \rfloor}$ with $O(\log(n/(t\eps_2)))$ bits of locality and error $\eps_2$, so the total locality is $O(\log(t)\cdot\log(t/\eps_1) + \log(n/t\eps_2))$.

    Union bounding over the errors implies $\tvdist{\Lcal_{k}-\Pcal_{k}} \le \eps_1 + k\eps_2$. Setting $\eps_1 = \eps/4$ and $\eps_2 = \eps/4k$, we find that $\Lcal_{k}$ is $(\eps/2)$-close to $\Pcal_{k}$. Finally setting $t = \lceil k^2 / (\eps/2) \rceil$ yields that $\Pcal_{k}$ is $(\eps/2)$-close to $\Dcal_{k}$, so we conclude that $f$ is an $O(\log(n/k) + \log^2(k/\eps))$-local function with $f(\Ucal^m)$ $\eps$-close to $\Dcal_{k}$.
\end{proof}

\subsection{Periodic Hamming Slices}\label{subsec:upper_periodslice}

In this subsection, we prove \Cref{thm:upper_mod}. Recall that $\Dcal_{q, \Lambda}$ for $\Lambda \in \Zbb/q\Zbb$ is the uniform distribution over $n$-bit strings of Hamming weight modulo $q$ in $\Lambda$. We emphasize the given proof is extremely similar to that of \cite{viola2012complexity}. We provide it for completeness and to demonstrate the connection to the results in \Cref{subsec:upper_singleslice}.
\begin{theorem*}[\Cref{thm:upper_mod} Restated \cite{viola2012complexity}]
    For all $q \in \Nbb$ and non-empty $\Lambda \subseteq \Zbb/q\Zbb$, there exists an $O(q^2 \cdot \log^2(n/\eps))$-local Boolean function $f : \bin^m \to \bin^n$ such that $f(\Ucal^m)$ is $\eps$-close to $\Dcal_{q, \Lambda}$.
\end{theorem*}

We start with the case of $|\Lambda|=1$, where our approach is similar to the \textsc{parity} example in the introduction.

\begin{lemma}\label{lem:upper_mod}
    For all $q \in \Nbb$ and $r \in \Zbb/q\Zbb$, there exists an $O(q^2 \cdot \log^2(n/\eps))$-local Boolean function $f : \bin^m \to \bin^n$ such that $f(\Ucal^m)$ is $\eps$-close to $\Dcal_{q, \{r\}}$.
\end{lemma}

We again will use $\Dcal_{s,t}$ to denote the uniform distribution over $x \in \bin^t$ of Hamming weight $s$ in cases where $t$ may not equal $n$.

\begin{proof}
    Assume $q = o(\sqrt{n/\log(n/\eps)})$, as otherwise the bound follows from \Cref{lem:distr_approx}. We can view $\Ucal^n$ as the resulting distribution of the following processing. 
    
    \begin{enumerate}
        \item Divide $[n]$ into $\lfloor n/t \rfloor$ blocks, each of size either $t$ or $t+1$. 
        
        (This is the reverse of what was done in the proof of \Cref{thm:upper_smallslice}, where there were $t$ blocks. It will be more convenient for later calculations.)
        
        \item Sample $y_1, \ldots, y_{\lfloor n/t \rfloor} \in \Zbb/q\Zbb$ according to the binomial distribution $\Bcal_{q,1/2}$:
        \[
        \Pr[y_i = j] = \frac{1}{2^q} \cdot \binom{q}{j}.
        \]
        
        \item Suppose the $i$-th block has size $h$. Sample $w_i$ according to the binomial distribution $\Bcal_{h,1/2}$ conditioned on the result being $y_i \Mod{q}$.
        
        \item For each $i$, put $w_i$ ones in block $i$ according to the distribution $\Dcal_{w_i, h}$.
\end{enumerate}

    We claim that the $y_i$'s are approximately uniform over $\Zbb/q\Zbb$. Thus, if we perform a variant of the above process where each $y_i$ is chosen uniformly at random from $\Zbb/q\Zbb$, we can get an $(\eps/2)$-approximation to the uniform distribution over $\bin^n$.
    \begin{claim}
        Setting $t = \Omega(q^2 \cdot \log(n /\eps))$ implies the uniform distribution over $\bin^n$ restricted to a block has weight mod $q$ $(\eps/2n)$-close to the uniform distribution over $\Zbb/q\Zbb$.
    \end{claim}
    \begin{proof}
        It suffices to show that $\abs{\Pr_{X \sim \bin^t}\sbra{\sum_{j\in[t]}X_j\bmod q = v} - 1/q} \le e^{-t/q^2}$ for any $v \in \Zbb/q\Zbb$. By Fourier expansion,
        $$
            \Pr\sbra{\sum_{j\in[t]}X_j\bmod q = v}
            =\E\sbra{\frac1q\sum_{a\in\Zbb/q\Zbb}\omega_q^{a\cdot(\sum_jX_j-v)}}
            = \frac{1}{q}+\frac1q\sum_{a\in\Zbb/q\Zbb, a\ne 0}\pbra{\frac{1+\omega_q^a}{2}}^t\cdot\omega_q^{-a\cdot v},
        $$
        where $\omega_q$ is the $q$-th unit root.

        By \Cref{clm:root_power}, we can bound the error by
        \begin{align*}
            \left|\frac1q\sum_{a\in\Zbb/q\Zbb, a\ne 0}\pbra{\frac{1+\omega_q^a}{2}}^t\cdot\omega_q^{-a\cdot v}\right| &\le \max_{a\in\Zbb/q\Zbb, a\ne 0} \left|\frac{1+\omega_q^a}{2}\right|^t
            \le \left(1 - \frac{1}{q^2}\right)^t
            \le e^{-t/q^2}. \qedhere
        \end{align*}
    \end{proof}

    Observe that $\Dcal_{q, \{r\}}$ is the uniform distribution over $\bin^n$ conditioned on the $y_i$'s summing to $r \Mod{q}$. We can sample $y_i$'s with this property by choosing $x_1, \ldots, x_{\lfloor n/t \rfloor} \in \Zbb/q\Zbb$ uniformly at random and setting 
    \[
    y_1 \coloneq x_1-x_2,\quad y_2 \coloneq x_2 - x_3,\quad \ldots,\quad y_{\lfloor n/t \rfloor-1} \coloneq x_{\lfloor n/t \rfloor-1}-x_{\lfloor n/t \rfloor},\quad y_{\lfloor n/t \rfloor} \coloneq x_{\lfloor n/t \rfloor}-x_1 + r.
    \]
    Call the resulting distribution of the above process with this modification $\Pcal$. Then we have $\tvdist{\Pcal - \Dcal_{q, \{r\}}} \le \eps/2$.

    We would like to approximate $\Pcal$ with a distribution $\Lcal \coloneq f(\Ucal^m)$ produced by a function $f : \bin^m \to \bin^n$ of small locality. By \Cref{lem:distr_approx}, we can sample from the uniform distribution over $\Zbb/q\Zbb$ with $\lceil \log(q/\eps') \rceil$ bits of locality and error $\eps'$. Similarly, we can sample from the uniform distribution over $\bin^{t+1}$ (or $\bin^t$) conditioned on the $y_i$'s summing to $0 \Mod{q}$ with $\lceil (t+1)\log(2/\eps') \rceil$ bits of locality and error $\eps'$. Finally, we can sample from $\Dcal_{w_i,h}$ with $O(\log(t)\cdot \log(t/\eps'))$ bits of locality and error $\eps'$ using \Cref{thm:combinedSliceUpper}. Output bits in block $i$ depend on the two $x_j$'s that affect $y_i$, $w_i$, and the placement of the $w_i$ ones, so the total locality is $O(\log(q/\eps')+t\log(2/\eps') +\log(t)\cdot \log(t/\eps'))$.

    Union bounding over the errors implies $\tvdist{\Lcal-\Pcal} \le \frac{n}{t} \cdot 3\eps'$.
    Setting $\eps' = t\eps/6n$ implies $\Lcal$ depends on $O(t\log(n/t\eps) + \log(t)\cdot\log(n/\eps))$ bits of locality and is $\eps/2$-close to $\Pcal$. Finally, setting $t = \Theta(q^2 \cdot \log(n/\eps))$, we find that $f$ is an $O(q^2 \cdot \log^2(n/\eps))$-local function with $f(\Ucal^m)$ $\eps$-close to $\Dcal_{q, \{r\}}$.
\end{proof}

By randomly choosing a Hamming weight from $\Lambda \subseteq \Zbb/q\Zbb$, we get \Cref{thm:upper_mod}.

\begin{proof}[Proof of \Cref{thm:upper_mod}]
    Randomly choose an element $r \in \Lambda$ with probability $\frac{\abs{\supp{\Dcal_{q,\{r\}}}}}{\abs{\supp{\Dcal_{q,\Lambda}}}}$, and then apply \Cref{lem:upper_mod} with error $\eps/2$. \Cref{lem:distr_approx} implies we can sample $r$ with locality $\lceil \log(2q/\eps) \rceil$ and error $\eps/2$. Applying the union bound concludes the proof. 
\end{proof}

\section*{Acknowledgements}
We thank Yuval Filmus for pointing this topic to us and answering questions regarding \cite{filmus2023sampling,watts2023unconditional}, thank Shachar Lovett for his support, thank Emanuele Viola for answering questions regarding \cite{viola2023new} and clarifying prior literature, and thank Adam Bene Watts and Natalie Parham for informing us of their updated version of \cite{watts2023unconditional}.
KW wants to thank Daniel Grier, Jiaqing Jiang, and Jack Morris for answering questions in quantum computing, and thank Mingmou Liu for questions in data structure lower bounds. AO wants to thank Max Hopkins for helpful conversations on the upper bounds. We are also indebted to anonymous reviewers for their useful feedback, especially their discussion on prior techniques and results.

Part of this work was done while AO and KW were visiting the Simons Institute for the Theory of Computing.

\bibliographystyle{alpha} 
\bibliography{biblio}

\appendix
\section{Tightness of the Graph Elimination Results}\label{app:tightness}

In this section, we show the tightness of our graph elimination results, which are the main bottlenecks preventing better locality lower bounds.
Recall that these graph elimination reductions aim to remove a few right vertices in a $d$-left-bounded bipartite graph to obtain non-connected left vertices or neighborhoods, which corresponds to conditioning on a few input bits in a $d$-local function to obtain a $(d,r)$-local or $(d,r,t)$-local function.

\subsection{Non-Connected Vertices}\label{app:tightness_vtx}

We start with the tightness of the graph elimination result for non-connected vertices, i.e., \Cref{as:non-adj_vtx} and \Cref{cor:non-adj_vtx_new}.
This is used in the lower bounds for sampling biased distributions (\Cref{thm:locality_biased}) and single Hamming slices of non-dyadic weight (\Cref{thm:locality_single_non-dyadic}).

\begin{property*}[\Cref{as:non-adj_vtx} Restated]
There exists $S\subseteq[m]$ such that deleting those right vertices (and their incident edges) produces a bipartite graph with $r$ non-connected left vertices satisfying
$$
|S|\le\frac r\beta
\quad\text{and}\quad
r\ge\frac n\lambda.
$$
\end{property*}

In light of \Cref{as:non-adj_vtx}, we prove the following statement.

\begin{lemma}\label{lem:tightness_vtx}
Let $\beta\ge2$ be an integer parameter (not necessarily constant).
If $n=(\beta-1)^d$, then there exists a $d$-left-bounded bipartite graph $G=([n],[m],E)$ such that for any $S\subseteq[m]$, deleting those right vertices (and their incident edges) gives at most $\max\cbra{1,(\beta-1)|S|}$ non-connected left vertices.
\end{lemma}

The instance $G$ above does not satisfy \Cref{as:non-adj_vtx} if $\lambda<n=(\beta-1)^d$.
Recall that \Cref{cor:non-adj_vtx_new} shows that \Cref{as:non-adj_vtx} holds as long as $\lambda\ge(d\beta)^{\Omega(d)}$.
Hence \Cref{lem:tightness_vtx} provides a sharp example for \Cref{cor:non-adj_vtx_new} when $\beta\ge d^{\Omega(1)}$, which is the typical setting for us.

In the whole analysis, it implies a barrier for improving the $\delta^{O(d)}$ factor in \Cref{thm:locality_biased} and \Cref{thm:locality_single_non-dyadic} to $\delta^{o(d)}$.
Put concretely, the $2^{O(d^2)}$ factor in \Cref{thm:informal_single_non-dyadic} and \Cref{thm:informal_biased} is inevitable in our analysis framework.

Now we proceed to the proof of \Cref{lem:tightness_vtx}.
\begin{proof}[Proof of \Cref{lem:tightness_vtx}]
The right vertices of $G$ will form a complete $(\beta-1)$-ary tree on top of the left vertices.
To be more precise, let $\Tcal_d$ be a complete $(\beta-1)$-ary tree of depth $d$, where the root has depth zero.
We identify the $(\beta-1)^d$ leaves of $\Tcal_d$ as the left vertices $[n]$, and identify the internal nodes of $\Tcal_d$ as the right vertices $[m]$.
From now on, we will use internal nodes for right vertices and leaves for left vertices.

In the bipartite graph $G$, each internal node is connected with all the leaves below it in $\Tcal_d$.
It is clear that $G$ is $d$-left-bounded as $\Tcal_d$ has depth $d$.
Suppose we removed internal nodes $S$ and obtained leaves $T$ that are not connected to each other in $G$.
If $|S|=0$, then clearly $|T|=1$. Hence now we assume $|S|\ge1$.
Observe that for any distinct leaves in $T$, their common ancestors in $\Tcal_d$ must be removed in $S$ to disconnect them.
Therefore, for any $v\in S$ and its child $v'\in\Tcal_d$, if $v'\notin S$ then at most one leaf in the sub-tree rooted from $v'$ can be contained in $T$.
This means that the size of $T$ is upper bounded by the number of leaves in a $(\beta-1)$-ary tree with at most $|S|$ internal nodes, where the latter is at most $(\beta-1)|S|$ when $|S|\ge1$.
\end{proof}

\subsection{Non-Connected Neighborhoods}\label{app:tightness_neigh}

Now we turn to the tightness of the graph elimination result for non-connected neighborhoods, i.e., \Cref{as:non-adj_neigh} and \Cref{cor:graph_elim_non-adj_neigh}.
This accounts for the gigantic dependence on $d$ in the lower bounds for sampling general (\Cref{thm:locality_single_hamming_slice}) and periodic (\Cref{thm:mod_slice}) Hamming slices.

\begin{property*}[\Cref{as:non-adj_neigh} Restated]
There exists $S\subseteq[m]$ such that deleting those right vertices (and their incident edges) produces a bipartite graph with $r$ non-connected left neighborhoods of size at most $t$ satisfying
$$
|S|\le\frac r{F(t)}
\quad\text{and}\quad
r\ge\frac n\lambda
\quad\text{and}\quad
t\le\kappa.
$$
\end{property*}

In light of \Cref{as:non-adj_neigh}, we present the following construction.

\begin{lemma}\label{lem:tightness_neigh}
Assume $d\ge2$ is an even number and $n=\tow_2(d/2)$.
There exists a $d$-left-bounded bipartite graph $G=([n],[m],E)$ such that for any $S\subseteq[m]$, deleting those right vertices (and their incident edges) gives at most $\max\cbra{1,|S|}$ non-connected left neighborhoods.
\end{lemma}

The instance $G$ above serves as a counterexample to \Cref{as:non-adj_neigh} if $F(t)>1$ and $\lambda<n=\tow_2(d/2)$.
In addition, by ``open-boxing'' the construction, there is a right vertex $v^*$ in $G$ incident to all left vertices.
Hence $G$ is also a counterexample to \Cref{as:non-adj_neigh} if $F(t)>1$ and $\kappa<n=\tow_2(d/2)$, since $t\le\kappa<n$ enforces that $v^*$ must be removed.
On the other hand, setting $F(t)=O(1)$ in \Cref{cor:graph_elim_non-adj_neigh}, we know that \Cref{as:non-adj_neigh} holds if $\lambda$ and $\kappa$ are indeed a tower of $d$. 
Hence \Cref{lem:tightness_neigh} shows that \Cref{cor:graph_elim_non-adj_neigh} is surprisingly tight.

Recall that in our actual applications, \Cref{thm:locality_single_hamming_slice} needs to set $F(t)=\Omega(t)$ and \Cref{thm:mod_slice} needs to set $F(t)=2^{\Omega(t)}$. Hence the assumption that $F(t)>1$ is extremely weak.
Yet we still cannot hope for a bound without a tower-type blowup on $d$.

Now we prove \Cref{lem:tightness_neigh}. Its construction is similar to the one in \Cref{lem:tightness_vtx}.
\begin{proof}[Proof of \Cref{lem:tightness_neigh}]
Define $k=d/2$.
Let $\Lcal_k$ be a rooted tree of depth $k$, where the root has depth zero. Each internal node in $\Lcal_k$ of depth $i\le k-1$ has arity $a_i=\frac{\log^{(i)}(n)}{\log^{(i+1)}(n)}=\frac{\tow_2(k-i)}{\tow_2(k-i-1)}$.\footnote{$\log^{(i)}(\cdot)$ is the iterated logarithm of order $i$, where $\log^{(0)}(x)=x$ and $\log^{(i)}(x)=\log(\log^{(i-1)}(x))$ for $i\ge1$.}
By construction, there are $b_i=\tow_2(k-i)$ leaves in a sub-tree rooted from a node of $\Lcal_k$ at depth $i$.

\begin{figure}[ht]
\includegraphics[scale=0.5]{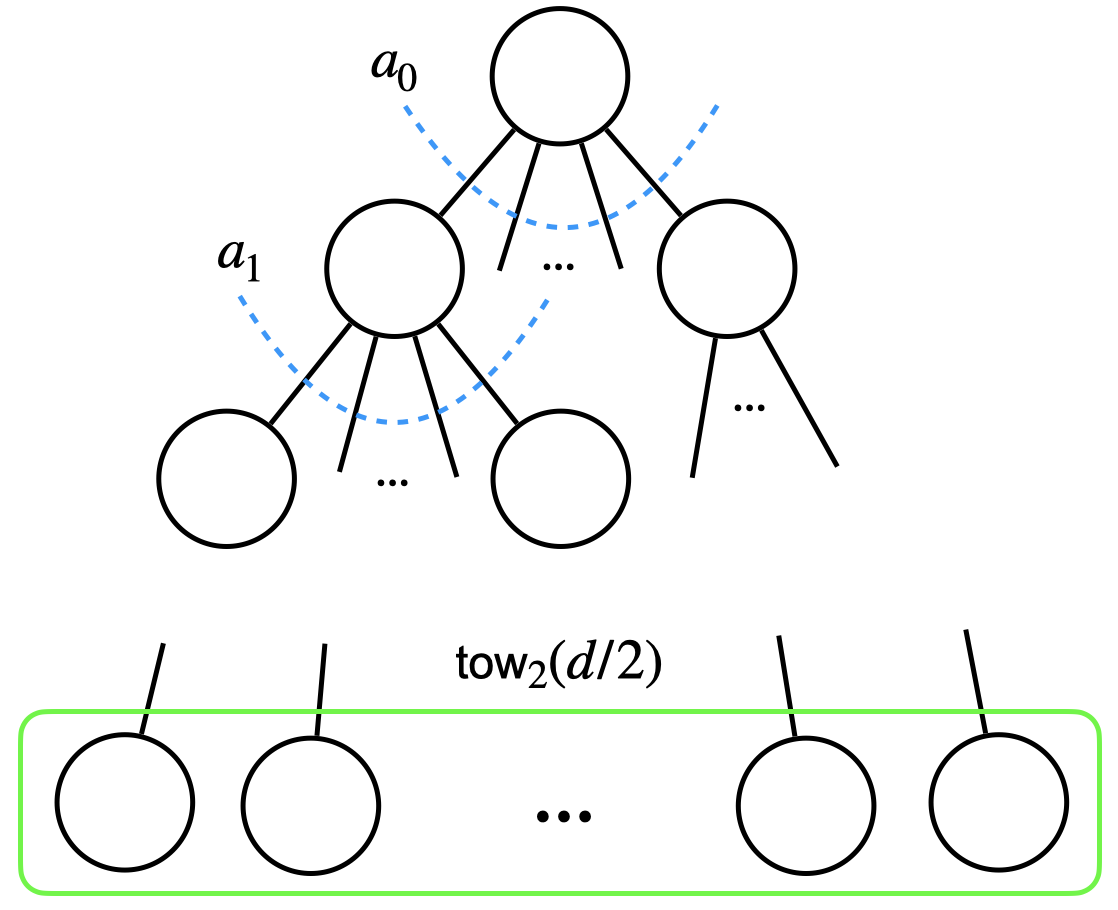}
\centering
\caption{The tree $\Lcal_k$. The vertices in the green box are identified as the left vertices of the bipartite graph $G$, while those outside the box correspond to the right vertices.}
\label{fig:Bv_2}
\end{figure}

Now similar to the proof of \Cref{lem:tightness_vtx}, we identify the $\tow_2(k)$ leaves of $\Lcal_k$ as the left vertices $[n]$, and put the internal nodes of $\Lcal_k$ as right vertices (see \Cref{fig:Bv_2}).
Then in the bipartite graph $G$, each internal node is connected with all the leaves below it in $\Lcal_k$.
It is clear that, at this point, $G$ is $k$-left-bounded as $\Lcal_k$ has depth $k$.

Suppose that we want to ensure that leaves $T\subseteq[n]$ have non-connected left \emph{neighborhoods}.
Then at least, we need $T$ to be a set of non-connected left \emph{vertices}.
Hence, analogous to the proof of \Cref{lem:tightness_vtx}, this implies the following claim.
\begin{claim}\label{clm:lem:tightness_neigh}
For any distinct leaves in $T$, their common ancestors in $\Lcal_k$ must be removed.
\end{claim}
Unfortunately there are not many such common ancestors, since the arities of the internal nodes of $\Lcal_k$ are extremely large.
To strengthen the connectivity among left neighborhoods, we now introduce more right vertices.

Fix an arbitrary internal node $v\in\Lcal_k$.
Let $0\le i\le k-1$ be its depth, and let $v_1,v_2,\ldots,v_{a_i}\in\Lcal_k$ be its children, where recall that $a_i=\frac{\tow_2(k-i)}{\tow_2(k-i-1)}$ is its arity.
In the current construction, we have only put $v_1,\ldots,v_{a_i}$ as right vertices on top of their \emph{respective} leaves. In other words, $v_1,\ldots,v_{a_i}$ support on disjoint sets of leaves.
To enhance the connectivity, we will now create more right vertices to link between these disjoint set of leaves.
For each $j\in[a_i]$, we list (in an arbitrary order) the leaves in the sub-tree rooted from $v_j$ as $u_{j,0},\ldots,u_{j,b_{i+1}-1}$, where recall that $b_{i+1}=\tow_2(k-i-1)$.
Then we will build a complete binary tree for $v_1,\ldots,v_{a_i}$, and use the $u_{j,\ell}$'s to spread out the degrees.
More precisely, let $\Bcal_v$ be a complete binary tree with $a_i$ leaves, where the $j$-th leaf is equipped with $u_{j,0},\ldots,u_{j,b_{i+1}-1}$.
Since $a_i=\frac{\tow_2(k-i)}{\tow_2(k-i-1)}$ and $0\le i\le k-1$, $\Bcal_v$ has depth $d_i=\log(a_i)\le b_{i+1}$ where $d_i\ge1$ is an integer and the root of $\Bcal_v$ root has depth zero.
For each internal node $w_v\in\Bcal_v$ of depth $0\le\ell\le d_i-1\le b_{i+1}-1$, we identify it as a new right vertex in $G$ and add an edge between $w_v$ and $u_{j,\ell}$ for each leaf $j$ below $w_v$ in $\Bcal_v$ (see \Cref{fig:Bv_1}).
Observe that different $w_v$ in the same depth will use different index $j$ of the $u$'s, and different $w_v$ of different depths will use different index $\ell$ of the $u$'s.
Hence these new right vertices are incident to disjoint sets of input vertices.

\begin{figure}[ht]
\begin{tabular}{ll}
\includegraphics[scale=0.34]{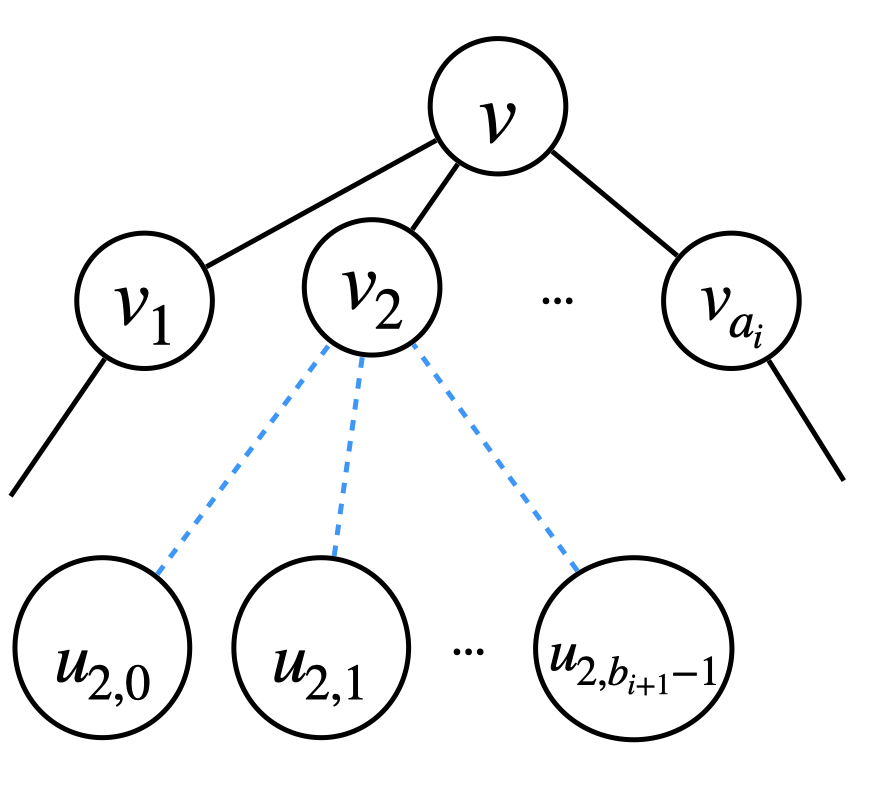}
&
\includegraphics[scale=0.34]{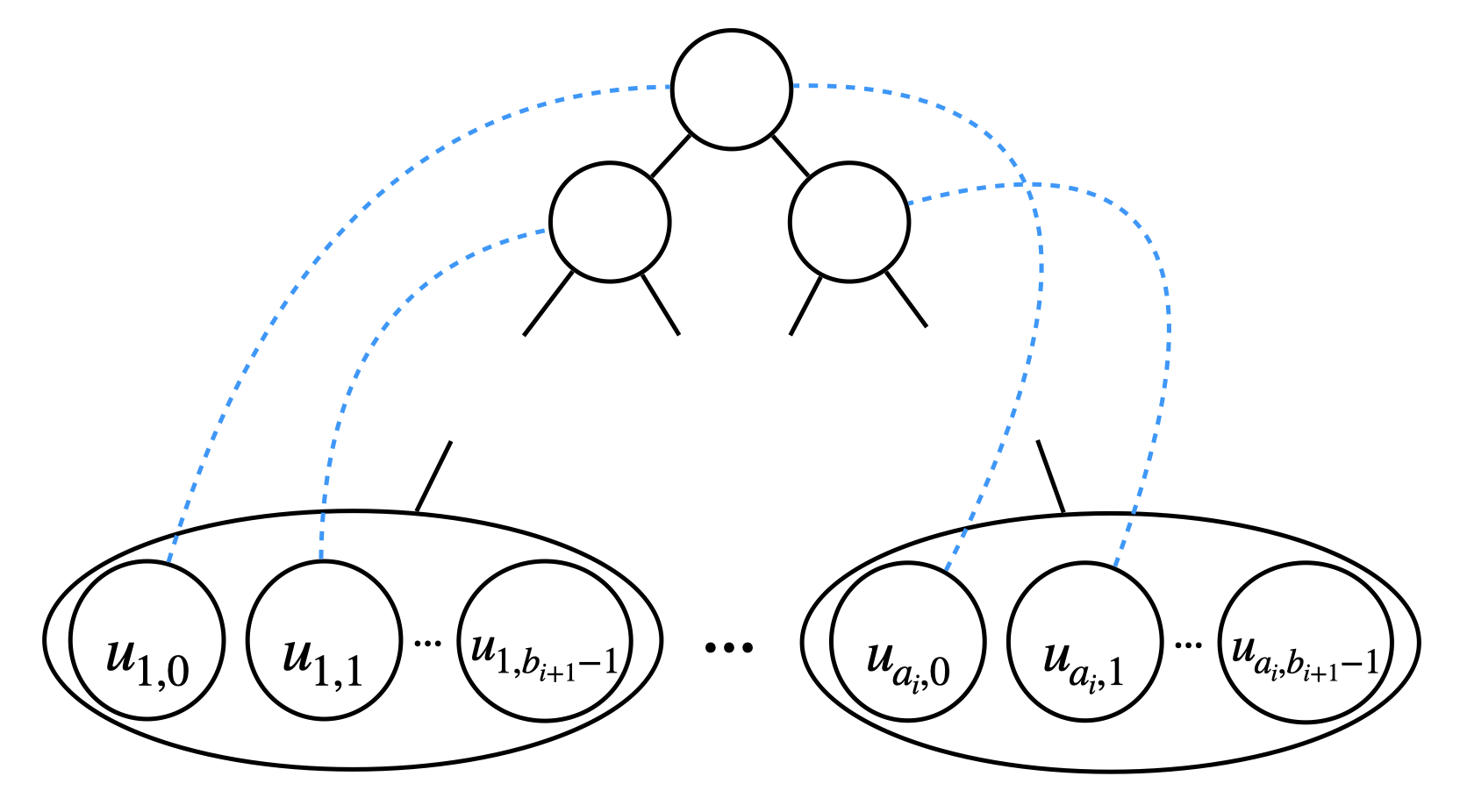}
\end{tabular}
\caption{Left: The sub-tree rooted at a node $v \in \Lcal_k$ of depth $i$. The blue dashed edges represent that the $u_{2,j}$'s are leaves in the sub-tree rooted at $v_2$. Right: The corresponding $\Bcal_v$. The blue dashed edges correspond to edges in the bipartite graph $G$.}
\label{fig:Bv_1}
\end{figure}

The construction of $G$ is completed by performing the above procedure for each internal node $v\in\Lcal_k$.
During the procedure for $v\in\Lcal_k$, only the degree of leaves (left vertices of $G$) below $v$ gets increased by at most one.
Therefore each left vertex of $G$ has total degree at most $2k=d$ as desired.

Recall that we aim to ensure that leaves $T\subseteq[n]$ have non-connected left neighborhoods.
Assume that we removed right vertices in $S$ to get this.
If $|T|\le1$, then clearly one does not need to remove any right vertex, i.e., $S=\emptyset$ suffices.
Now we assume $|T|\ge2$.
For each $u\in T$, we define $v_u\in\Lcal_k$ to be the ancestor of $u$ that is not removed in $S$ and has the smallest depth.
Note that if all ancestors of $u$ are removed, then $v_u=u$ itself.
Let $p_u$ be the parent node of $v_u\in\Lcal_k$. Since $|T|\ge2$ and by \Cref{clm:lem:tightness_neigh}, we at least need to remove the root of $\Lcal_k$ and thus $p_u$ always exists.

Now for a fixed internal node $p\in\Lcal_k$, we analyze all $u\in T$ with $p_u=p$.
Let $T_p=\cbra{u\in T\colon p_u=p}$.
Recall that we also constructed right vertices based on $\Bcal_p$ to enhance the connectivity of the first-step construction based on $\Lcal_k$.
Define $S_p=\cbra{v\in S\colon v=p\lor v\in\Bcal_p}$. Then it suffices to show $|T_p|\le|S_p|$, since
$$
|S|=\sum_{\text{internal node }p\in\Lcal_k}|S_p|\ge\sum_{\text{internal node }p\in\Lcal_k}|T_p|=|T|.
$$
Without loss of generality, assume $|T_p|\ge1$.
Firstly, $p\in S_p$ by definition.
Now we look at the right vertices in $\Bcal_p$.
By \Cref{clm:lem:tightness_neigh}, different $u$ have different $v_u$.
Therefore, $v_u,v_{u'}$ for distinct $u,u'\in T_p$ correspond to distinct leaves in $\Bcal_p$.
Let $w\in\Bcal_p$ be a common ancestor of $v_u,v_{u'}$.
Since both $v_u,v_{u'}$ are not removed, $w$ must be removed; otherwise the left neighborhoods of $u$ and $u'$ are connected.

Indeed, the sub-trees (in $\Lcal_k$) rooted at $v_u$ and $v_{u'}$ must have leaves $c_{u}$ and $c_{u'}$ with edges $(c_u, w)$ and $(c_{u'}, w)$, respectively. Then $u$ connects with $u'$ via $v_u, c_u, w, c_{u'}, v_{u'}$. 
This means that any common ancestors of $v_u,v_{u'}$ in $\Bcal_p$ for pairs of distinct $u,u'\in T_p$ must be contained in $S_p$.
Since $|T_p|\ge1$ and $\Bcal_p$ is a binary tree, we have at least $|T_p|-1$ such ancestors.
In summary, $|S_p|\ge1+(|T_p|-1)=|T_p|$ as claimed.
\end{proof}
\section[Local Limit Theorems on the Additive Group Modulo q]{Local Limit Theorems on the Additive Group Modulo $q$}\label{app:mod_llt}

In this section, we prove \Cref{lem:mod_llt}, which is a local limit result for $\Zbb/q\Zbb$.
One can also view it as the variant of \Cref{fct:ushakov} in the cyclic groups with sharp estimates.
The basic strategy is to use Fourier analysis to bound the sum of random variables, where the non-constancy assumption is used to show that certain coefficients are small.

We start by proving a more general statement.
\begin{theorem}\label{thm:mod_llt}
Let $q\ge2$ be an integer, and let $X_1,\ldots,X_n$ be independent random variables in $\Zbb$.
For each $j\in[n]$ and $r\ge1$, define $p_{r,j}=\max_{x\in\Zbb}\Pr\sbra{X_j\equiv x\Mod r}$.
Then for any $\Lambda\subseteq\Zbb/q\Zbb$, we have
$$
\Pr\sbra{\sum_{j\in[n]}X_j\mod q\in \Lambda}\le\frac1q\sum_{s\in[q]}\pbra{\sum_{a\in T_s}\abs{\sum_{c\in\Lambda}\omega_q^{-a\cdot c}}}\cdot\exp\cbra{-\frac{2\cdot\sum_{j\in[n]}(1-p_{q/s,j})}{q^2}},
$$
where $\omega_q=e^{2\pi i/q}$ is the primitive $q$-th root of unit, $\gcd(\cdot,\cdot)$ is the greatest common divider function, and $T_s=\cbra{a\in\Zbb/q\Zbb\colon\gcd(a,q)=s}$.
\end{theorem}

As a consequence, we finish the proof of \Cref{lem:mod_llt}.

\begin{lemma*}[\Cref{lem:mod_llt} Restated]
Under the same conditions as \Cref{thm:mod_llt}, assume further $\sum_{j\in[n]}(1-p_{r,j})\ge L$ holds for all $r\ge3$ dividing $q$.
Then for $\Lambda\subseteq\Zbb/q\Zbb$, we have
$$
\Pr\sbra{\sum_{j\in[n]}X_j\mod q\in\Lambda}\le q\cdot e^{-2L/q^2}+\begin{cases}
|\Lambda|/q & q\text{ is odd,}\\
2\cdot\max\cbra{|\Lambda_\textsf{even}|,|\Lambda_\textsf{odd}|}/q & q\text{ is even,}
\end{cases}
$$
where $\Lambda_\textsf{even}=\cbra{\text{even numbers in }\Lambda}$ and $\Lambda_\textsf{odd}=\cbra{\text{odd numbers in }\Lambda}$.
\end{lemma*}
\begin{proof}
We follow the notation convention in \Cref{thm:mod_llt} and define for each $s\in[q]$
$$
A_s=\pbra{\sum_{a\in T_s}\abs{\sum_{c\in\Lambda}\omega_q^{-a\cdot c}}}\cdot\exp\cbra{-\frac{2\cdot\sum_{j\in[n]}(1-p_{q/s,j})}{q^2}}.
$$
For $s=q$, we clearly have $A_s=|\Lambda|$.
If $q$ is even and $s=q/2$, we have $T_s=\cbra{q/2}$ and 
\begin{equation*}
A_s\le\abs{\sum_{c\in\Lambda}\omega_q^{-c\cdot q/2}}=\abs{|\Lambda_\textsf{even}|-|\Lambda_\textsf{odd}|}.
\tag{since $\omega_q^{q/2}=-1$}
\end{equation*}
For any $s\in[q]$ dividing $q$ and not equal to $q$ or $q/2$, we have $q/s\ge3$ and thus $A_s\le|T_s|\cdot|\Lambda|\cdot e^{-2L/q^2}$.
Then the desired bound follows from \Cref{thm:mod_llt} by observing that $|\Lambda|+\abs{|\Lambda_\textsf{even}|-|\Lambda_\textsf{odd}|}=2\cdot\max\cbra{|\Lambda_\textsf{even}|,|\Lambda_\textsf{odd}|}$ and $\sum_{s\in[q]}|T_s|\cdot|\Lambda|=q\cdot|\Lambda|\le q^2$.
\end{proof}

\begin{remark}
We remark that \Cref{lem:mod_llt} is roughly tight:
\begin{itemize}
\item If $q$ is odd, we let each $X_j$ be an unbiased coin.
Then $\sum_jX_j\bmod q$ converges to the uniform distribution over $\Zbb/q\Zbb$ as $n\to+\infty$, which means the LHS converges to $|\Lambda|/q$.
\item If $q$ is even, we let each $X_j$ be uniform in $\cbra{0,2}$.
Then $\sum_jX_j\bmod q$ converges to the uniform distribution over even numbers in $\Zbb/q\Zbb$, which means the LHS converges to $2\cdot|\Lambda_\textsf{even}|/q$. By changing $X_1$ to be uniform in $\cbra{1,3}$, we get the other bound $2\cdot|\Lambda_\textsf{odd}|/q$.
\end{itemize}
\end{remark}

Now we prove \Cref{thm:mod_llt}.
\begin{proof}[Proof of \Cref{thm:mod_llt}]
We first note the following Fourier estimate, which will be proved later.
\begin{claim}\label{clm:root_power}
For each $a\in\Zbb/q\Zbb$, let $s=\gcd(a,q)$. We have
$$
\abs{\E_{X_j}\sbra{\omega_q^{a\cdot X_j}}}\le1-\frac{2\cdot\pbra{1-p_{q/s,j}}}{q^2}
\quad\text{for each }j\in[n].
$$
\end{claim}
Assuming \Cref{clm:root_power}, we now complete the proof of \Cref{thm:mod_llt}.
Observe that
\begin{align*}
\Pr\sbra{\sum_{j\in[n]}X_j\mod q\in \Lambda}
&=\E\sbra{\sum_{c\in \Lambda}\frac1q\sum_{a\in\Zbb/q\Zbb}\omega_q^{a\cdot(\sum_jX_j-c)}}\\
&=\frac1q\sum_{a\in\Zbb/q\Zbb}\pbra{\prod_{j\in[n]}\E_{X_j}\sbra{\omega_q^{a\cdot X_j}}}\pbra{\sum_{c\in \Lambda}\omega_q^{-a\cdot c}}\\
&\le\frac1q\sum_{a\in\Zbb/q\Zbb}\abs{\prod_{j\in[n]}\E_{X_j}\sbra{\omega_q^{a\cdot X_j}}}\cdot\abs{\sum_{c\in \Lambda}\omega_q^{-a\cdot c}}.
\end{align*}
Now for each $a\in\Zbb/q\Zbb$, by \Cref{clm:root_power} we have
$$
\abs{\prod_{j\in[n]}\E_{X_j}\sbra{\omega_q^{a\cdot X_j}}}
\le\prod_{j\in[n]}\pbra{1-\frac{2\cdot\pbra{1-p_{q/\gcd(a,q),j}}}{q^2}}
\le\exp\cbra{-\frac{2\cdot\sum_{j\in[n]}(1-p_{q/\gcd(a,q),j})}{q^2}},
$$
which implies
\begin{equation*}
\Pr\sbra{\sum_{j\in[n]}X_j\mod q\in \Lambda}
\le\frac1q\sum_{s\in[q]}\pbra{\sum_{a\in T_s}\abs{\sum_{c\in\Lambda}\omega_q^{-a\cdot c}}}\cdot\exp\cbra{-\frac{2\cdot\sum_{j\in[n]}(1-p_{q/s,j})}{q^2}}.
\tag*{\qedhere}
\end{equation*}
\end{proof}

Finally we prove \Cref{clm:root_power}.
\begin{proof}[Proof of \Cref{clm:root_power}]
Let $z=\omega_q^\theta$ where $\theta\in\Rbb$ is an arbitrary number.
To prove \Cref{clm:root_power}, it suffices to show 
\begin{equation}\label{eq:clm:root_power_1}
\Re\pbra{z\cdot\E\sbra{\omega_q^{a\cdot X_j}}}\le1-\frac{2\cdot\pbra{1-p_{q/s,j}}}{q^2}
\quad\text{for all $z$,}
\end{equation}
where $\Re(\cdot)$ is the real part of a complex number.

Observe that 
\begin{equation}\label{eq:clm:root_power_2}
\Re\pbra{z\cdot\E\sbra{\omega_q^{a\cdot X_j}}}
=\E\sbra{\cos\pbra{\frac{2\pi(a\cdot X_j+\theta)}q}}
=1-2\cdot\E\sbra{\sin^2\pbra{\frac{\pi(a\cdot X_j+\theta)}q}}.
\end{equation}
Now it suffices to lower bound $\E\sbra{\sin^2\pbra{\frac{\pi(a\cdot X_j+\theta)}q}}$.
Let $\Ecal$ be the event that $(a\cdot X_j+\theta)\bmod q\in(-1/2,1/2]$, where the $\bmod$ here is over the reals.
If $\Ecal$ does not happen, we have
\begin{equation}\label{eq:clm:root_power_3}
\sin^2\pbra{\frac{\pi\pbra{a\cdot X_j+\theta}}q}\ge\sin^2\pbra{\frac\pi{2q}}\ge\frac1{q^2},
\end{equation}
where we used the fact that $\sin(\pi x)\ge2x$ for $0\le x\le1/2$.
Hence 
$$
\E\sbra{\sin^2\pbra{\frac{\pi(a\cdot X_j+\theta)}q}}\ge\frac1{q^2}\cdot\Pr\sbra{\neg\Ecal}.
$$
On the other hand, since both $a\cdot X_j$ and $q$ are integers, there is a unique value $b\in\Zbb/q\Zbb$ such that $\Ecal$ holds if and only if $a\cdot X_j\equiv b\Mod q$.
Now, if $s=\gcd(a,q)$ does not divide $b$, this can never happen and hence $\Pr[\Ecal]=0$.
Otherwise, setting $a'=a/s$ and $b'=b/s$, we have
$$
\Pr\sbra{\Ecal}
=\Pr\sbra{a'\cdot X_j\equiv b'\Mod{q/s}}
=\Pr\sbra{X_j\equiv (a')^{-1}\cdot b'\Mod{q/s}}
\le p_{q/s,j},
$$
where $(a')^{-1}$ is the inverse of $a'$ modulo $q/s$ and the last inequality uses the assumption in \Cref{thm:mod_llt}.
Therefore \Cref{eq:clm:root_power_1} follows by combining the above bound with \Cref{eq:clm:root_power_2} and \Cref{eq:clm:root_power_3}.
\end{proof}
\section[Missing Proofs in Section 4]{Missing Proofs in \Cref{sec:useful_lemmas}}\label{app:missing_proof_sec:useful_lemmas}

Here, we put omitted proofs from \Cref{sec:useful_lemmas}.

\subsection*{Proof of \Cref{fct:mult_apx}}\label{app:proof_of_fct:mult_apx}

\begin{proof}[Proof of \Cref{fct:mult_apx}]
Fix an event $\Ecal'$ that attains $\Qcal'(\Ecal')-\Pcal'(\Ecal')=\tvdist{\Pcal'-\Qcal'}$. Then we have
\begin{align*}
\Qcal'(\Ecal')-\Pcal'(\Ecal')
&=\frac{\Qcal(\Ecal\land\Ecal')}{\Qcal(\Ecal)}-\frac{\Pcal(\Ecal\land\Ecal')}{\Pcal(\Ecal)}
=\frac{\Qcal(\Ecal\land\Ecal')}{\Qcal(\Ecal)}-\frac{\Pcal(\Ecal\land\Ecal')}{\Qcal(\Ecal)}+\frac{\Pcal(\Ecal\land\Ecal')\cdot(\Pcal(\Ecal)-\Qcal(\Ecal))}{\Qcal(\Ecal)\cdot\Pcal(\Ecal)}\\
&\le\frac{\Qcal(\Ecal\land\Ecal')}{\Qcal(\Ecal)}-\frac{\Pcal(\Ecal\land\Ecal')}{\Qcal(\Ecal)}+\frac{\eps\cdot\Pcal(\Ecal\land\Ecal')}{\Qcal(\Ecal)\cdot\Pcal(\Ecal)}
\tag{since $\Pcal(\Ecal)-\Qcal(\Ecal)\le\tvdist{\Pcal-\Qcal}\le\eps$}\\
&\le\frac{\Qcal(\Ecal\land\Ecal')}{\Qcal(\Ecal)}-\frac{\Pcal(\Ecal\land\Ecal')}{\Qcal(\Ecal)}+\frac\eps{\Qcal(\Ecal)}
\tag{since $\Pcal(\Ecal\land\Ecal')\le\Pcal(\Ecal)$}\\
&\le\frac\eps{\Qcal(\Ecal)}+\frac\eps{\Qcal(\Ecal)}
\tag{since $\Qcal(\Ecal\land\Ecal')-\Pcal(\Ecal\land\Ecal')\le\tvdist{\Pcal-\Qcal}\le\eps$}\\
&=\frac{2\eps}{\Qcal(\Ecal)}.
\end{align*}
Applying the data processing inequality concludes the proof.
\end{proof}

\subsection*{Proof of \Cref{clm:tvdist_gamma_biased_shift}}\label{app:proof_of_clm:tvdist_gamma_biased_shift}

\begin{proof}[Proof of \Cref{clm:tvdist_gamma_biased_shift}]
Let $\omega_q=e^{2\pi i/q}$ be the primitive $q$-th root of unit.
We consider the following quantity
$$
Q=\abs{\E_{X\sim\Dcal_0}\sbra{\omega_q^X}-\E_{X\sim\Dcal_1}\sbra{\omega_q^X}}.
$$
On the one hand, we have
\begin{equation}\label{eq:clm:tvdist_gamma_biased_shift_1}
Q\le\sum_{c\in\Zbb/q\Zbb}\abs{\omega_q^c\cdot\pbra{\Dcal_0(c)-\Dcal_1(c)}}=\sum_{c\in\Zbb/q\Zbb}\abs{\Dcal_0(c)-\Dcal_1(c)}=2\cdot\tvdist{\Dcal_0-\Dcal_1}.
\end{equation}
On the other hand, we have
\begin{align}
Q
&=\abs{\pbra{1-\gamma+\gamma\cdot\omega_q}^{t-1}-\omega_q\cdot\pbra{1-\gamma+\gamma\cdot\omega_q}^{t-1}}
\tag{by the definition of $\Dcal_0$ and $\Dcal_1$}\\
&=\abs{1-\omega_q}\cdot\abs{1-\gamma+\gamma\cdot\omega_q}^{t-1}.
\label{eq:clm:tvdist_gamma_biased_shift_2}
\end{align}

Let $r=\sin^2\pbra{\frac\pi q}$.
Then
$$
\abs{1-\omega_q}=\sqrt{\pbra{1-\cos\pbra{\frac{2\pi}q}}^2+\sin^2\pbra{\frac{2\pi}q}}
=2\cdot\abs{\sin\pbra{\frac\pi q}}
=2\sqrt r
$$
and
\begin{align*}
\abs{1-\gamma+\gamma\cdot\omega_q}
&=\sqrt{\pbra{1-\gamma+\gamma\cdot\cos\pbra{\frac{2\pi}q}}^2+\gamma^2\cdot\sin^2\pbra{\frac{2\pi}q}}
\notag\\
&=\sqrt{1-4\gamma(1-\gamma)\cdot\sin^2\pbra{\frac\pi q}}
=\sqrt{1-4\gamma(1-\gamma)r}.
\end{align*}
Combining these with \Cref{eq:clm:tvdist_gamma_biased_shift_1} and \Cref{eq:clm:tvdist_gamma_biased_shift_2}, we have
\begin{equation}\label{eq:clm:tvdist_gamma_biased_shift_3}
\tvdist{\Dcal_0-\Dcal_1}
\ge\sqrt{r\cdot\pbra{1-4\gamma(1-\gamma)r}^{t-1}}~\reflectbox{$\coloneqq$}~\sqrt{F}.
\end{equation}

If $q=3$, then $r=\sin^2(\pi/q)=3/4$. 
Hence
$$
F=\frac34\cdot\pbra{1-3\gamma(1-\gamma)}^{t-1}\ge\frac34\cdot2^{-8\gamma(1-\gamma)(t-1)}\ge\frac34\cdot2^{-8\gamma(t-1)}\ge\frac49\cdot2^{-100\gamma(t-1)/9}
$$
where we used the fact that $3\gamma(1-\gamma)\le3/4$ and $(1-x)^{1/x}\ge2^{-8/3}$ for $x\le3/4$.
Otherwise $q\ge4$. 
We use a different inequality that $2x\le\sin(\pi x)\le\pi x$ for $0\le x\le1/2$, and obtain $4/q^2\le r\le10/q^2$.
Hence
$$
F\ge\frac4{q^2}\cdot\pbra{1-\frac{40\gamma(1-\gamma)}{q^2}}^{t-1}
\ge\frac4{q^2}\cdot2^{-100\gamma(1-\gamma)(t-1)/q^2}
\ge\frac4{q^2}\cdot2^{-100\gamma(t-1)/q^2},
$$
where we used the fact that $40\gamma(1-\gamma)/q^2\le10/q^2\le5/8$ and $(1-x)^{1/x}\ge2^{-10/4}$ for $x\le5/8$.
Putting these back to \Cref{eq:clm:tvdist_gamma_biased_shift_3} gives
$$
\tvdist{\Dcal_0-\Dcal_1}
\ge\frac2q\cdot2^{-50\gamma(t-1)/q^2}
$$
as desired.
\end{proof}

\end{document}